\documentclass[twoside,11pt]{article}
\usepackage{jair, theapa, rawfonts}

\jairheading{81}{2024}{1-40}{11/2023}{10/2024}
\ShortHeadings{Proof Theory and Decision Procedures for Deontic STIT Logics}
{Lyon \& van Berkel}
\firstpageno{1}

\usepackage{lineno}
\usepackage[normalem]{ulem}

\usepackage{graphicx,amsmath,amsthm,amssymb,stmaryrd}
\usepackage[all]{xy}
\usepackage{pgf,changepage,relsize,tcolorbox}
\usepackage[utf8]{inputenc}
\usepackage{enumerate,rotating,appendix,bussproofs,framed}
\usepackage{upgreek}
\usepackage{mathtools}  
\usepackage{graphicx}
\usepackage{xfrac}
\usepackage[colorinlistoftodos]{todonotes}
\usepackage{color, multicol}
\usepackage{rotating}
\usepackage{colortbl}
\usepackage{marvosym}
\usepackage[ruled,linesnumbered]{algorithm2e} % Enables the writing of pseudo code.

\newcommand{\noon}{\texttt{n} }%{\texttt{noon}}
\newcommand{\clean}{\texttt{c} }%{\texttt{clean}}
\newcommand{\car}{\texttt{car}}
\newcommand{\foot}{\texttt{f} }%{\texttt{foot}}
\newcommand{\post}{\texttt{p} }%{\texttt{post}}

\newcommand{\Oy}{\otimes_y}

\newcommand{\Ody}{\ominus_y}

\newcommand{\no}{\texttt{n}}

\newcommand{\fo}{\texttt{f}}

\newcommand{\cl}{\texttt{c}}
\newcommand{\po}{\texttt{p}}
 
\usepackage{tikz}
\usetikzlibrary{decorations.pathreplacing,trees}
\usetikzlibrary{positioning,arrows,calc,decorations.markings}
\usetikzlibrary{shapes.geometric, arrows, calc, fit}
\tikzset{
modal/.style={>=stealth',shorten >=1pt,shorten <=1pt,auto,node distance=1.5cm,
semithick},
world/.style={circle,draw,minimum size=0.5cm,fill=gray!15},
point/.style={circle,draw,inner sep=0.5mm,fill=black},
reflexive above/.style={->,loop,looseness=7,in=120,out=60},
reflexive below/.style={->,loop,looseness=7,in=240,out=300},
reflexive left/.style={->,loop,looseness=7,in=150,out=210},
reflexive right/.style={->,loop,looseness=7,in=30,out=330}
}
\tikzstyle{level 1}=[level distance=1.5cm, sibling distance=1.5cm]
\tikzstyle{level 2}=[level distance=1.5cm, sibling distance=2cm]
\tikzstyle{bag} = [text width=8em, text centered]
\usetikzlibrary{patterns}
\usetikzlibrary{trees}

\EnableBpAbbreviations

\def\phi{\varphi}

\newcommand{\sect}{Section}
\newcommand{\fig}{Figure}
\newcommand{\dfn}{Definition}
\newcommand{\lem}{Lemma}
\newcommand{\thrm}{Theorem}

\newcommand{\rmk}{Remark}

\newcommand{\ifandonlyif}{\textit{iff} }
\newcommand{\iffi}{\textit{iff} }

\newcommand{\alg}{Algorithm}

\newtheorem{theorem}{Theorem}
\newtheorem{lemma}[theorem]{Lemma}
\newtheorem{corollary}[theorem]{Corollary}
\newtheorem{remark}[theorem]{Remark}
\newtheorem{definition}[theorem]{Definition}

\newtheorem{example}[theorem]{Example}

\definecolor{green2}{RGB}{154, 205, 50}
\definecolor{comk}{RGB}{154, 205, 50}
\definecolor{Gray1}{gray}{0.9}
\definecolor{Gray2}{gray}{0.7}
\newcommand {\Diam}{\Diamond}
\newcommand{\lang}{\mathcal{L}}

\newcommand{\Oi}{\otimes_{i}}
\newcommand{\lb}{\langle}
\newcommand{\rb}{\rangle}
\newcommand{\boxr}{(\Box)}
\newcommand{\diar}{(\Diamond)}
\newcommand{\stitdiam}{\lb i \rb}
\newcommand{\stitbox}{[i]}
\newcommand{\agdia}{\lb i \rb}
\newcommand{\agbox}{[i]}
\newcommand{\Ag}{Ag}
\newcommand{\Var}{Var}

\newcommand{\ODi}{\ominus_i}
\newcommand{\sti}{\mathsf{STIT}}
\newcommand{\opt}{I}

\newcommand{\negnnf}[1]{\neg{#1}}

\newcommand{\ioacond}{\textbf{(C2)}}

\newcommand{\baselogic}{\mathsf{DS}}

\newcommand{\dsnk}{\mathsf{DS}_{n}^{k}}

\newcommand{\R}{R}

\newcommand{\gtdsnk}{\mathsf{G3DS}_{n}^{k}}

\newcommand{\provenk}{\mathtt{Prove}_{n}^{k}}
\newcommand{\scid}{\mathbf{(C_{id})}}
\newcommand{\sccon}{\mathbf{(C_{\land})}}
\newcommand{\scdis}{\mathbf{(C_{\lor})}}
\newcommand{\scbox}{\mathbf{(C_{\Box})}}
\newcommand{\scdia}{\mathbf{(C_{\Diamond})}}
\newcommand{\scodi}{\mathbf{(C_{\ODi})}}
\newcommand{\scoi}{\mathbf{(C_{\Oi})}}
\newcommand{\scagbox}{\mathbf{(C_{\agbox})}}
\newcommand{\scagdia}{\mathbf{(C_{\agdia})}}
\newcommand{\scref}{\mathbf{(C_{ref})}}
\newcommand{\sceuc}{\mathbf{(C_{euc})}}
\newcommand{\scdii}{\mathbf{(C_{D2})}}
\newcommand{\scdiii}{\mathbf{(C_{D3})}}

\newcommand{\scapc}{\mathbf{(C_{APC})}}
\newcommand{\ripath}{\sim^{\rel}_{i}}
\newcommand{\thrd}{\mathcal{T}}
\newcommand{\ioat}{\mathsf{IOA}^{\thrd}\!\!}
\newcommand{\gt}{G}

\newcommand{\rel}{\mathcal{R}}
\newcommand{\qel}{\mathcal{Q}}
\newcommand{\strucset}{\mathsf{StrR}}
\newcommand{\mstamod}{M^{\Lambda}}
\newcommand{\wstamod}{W^{\Lambda}}
\newcommand{\rstamod}{R_{[i]}^{\Lambda}}
\newcommand{\istamod}{\opt_{\Oi}^{\Lambda}}
\newcommand{\vstamod}{V^{\Lambda}}

\newcommand{\der}{\gtdsnk \vdash}
\newcommand{\reli}{R_{[i]}}
\newcommand{\relio}{I_{\Oi}}

\newcommand{\sub}{(\mathsf{Sub})}
\newcommand{\cut}{(\mathsf{Cut})}
\newcommand{\wk}{(\mathsf{Wk})}

\newcommand{\settlc}{(\Box^{*})}
\newcommand{\Oilc}{(\Oi^{*})}
\newcommand{\stitlc}{([i]^{*})}
\newcommand{\stitdialc}{(\agdia^{*})}

\newcommand{\stitrefl}{(\mathsf{Ref}_{i})}
\newcommand{\stiteucl}{(\mathsf{Euc}_{i})}

\newcommand{\stittrans}{(\mathsf{Tra}_{i})}
\newcommand{\stitsym}{(\mathsf{Sym}_{i})}
\newcommand{\dtwo}{(\mathsf{D2}_{i})}

\newcommand{\ioa}{(\mathsf{IOA})}
\newcommand{\id}{(\mathsf{id})}
\newcommand{\choicer}{(\mathsf{APC}_{i}^{k})}

\newcommand{\sar}{\Rightarrow}
\newcommand{\Oir}{(\Oi)}
\newcommand{\ODir}{(\ominus_{i})}

\newcommand{\disr}{(\vee)}
\newcommand{\conr}{(\wedge)}
\newcommand{\settr}{(\Box)}
\newcommand{\settdiar}{(\Diamond)}
\newcommand{\dthree}{(\mathsf{D3}_{i})}
\newcommand{\stitr}{(\agbox)}
\newcommand{\stitdiar}{(\agdia)}

\newcommand{\ioaop}{\mathtt{IoaOp}}

\newcommand{\sufo}{\mathrm{sufo}}

\newcommand{\lab}{Lab}
\newcommand{\I}[1]{#1^{I}}
\newcommand{\ideal}{I_{\otimes_{i}}}

\begin{document}

\title{Proof Theory and Decision Procedures\\ for Deontic STIT Logics} 

\author{\name Tim S. Lyon \email timothy\_stephen.lyon@tu-dresden.de \\
       \addr Technische Universit\"at Dresden,\\ %Institute of Artificial Intelligence, \\ 
       N\"othnitzer Str. 46, 01069 Dresden, Germany
       \AND
             \name Kees van Berkel \email kees@logic.at \\
             \addr TU Wien, Institute for Logic and Computation \\
             Favoritenstra{\ss}e 9, 1040, Vienna, Austria}

\maketitle

\begin{abstract}
\noindent This paper provides a set of cut-free complete sequent-style calculi for deontic $\sti$ (`See To It That') logics used to formally reason about choice-making, obligations, and norms in a multi-agent setting. We leverage these calculi to write a proof-search algorithm deciding deontic, multi-agent $\sti$ logics with (un)limited choice and introduce a loop-checking mechanism to ensure the termination of the algorithm. Despite the acknowledged potential for deontic reasoning in the context of autonomous, multi-agent scenarios, this work is the first to provide a syntactic decision procedure for this class of logics. Our proof-search procedure is designed to provide verifiable witnesses/certificates of the (in)validity of formulae, which permits an analysis of the (non)theoremhood of formulae and act as explanations thereof. We show how the proof system and decision algorithm can be used to automate normative reasoning tasks such as \emph{duty checking} (viz. determining an agent's obligations relative to a given knowledge base), \emph{compliance checking} (viz. determining if a choice, considered by an agent as potential conduct, complies with the given knowledge base), and \emph{joint fulfillment checking} (viz. determining whether under a specified factual context an agent can jointly fulfill all their duties).
\end{abstract}

\section{Introduction}\label{Sect:Intro}

The logic of `Seeing To It That’, commonly referred to as $\sti$, is a formal framework developed by \citeauthor{BelPerXu01} \citeyear{BelPerXu01} for the analysis of agential choice in multi-agent settings. Next to its philosophical merits, the formalism has acquired a central position in the investigation of agency within the domain of artificial intelligence (AI). The framework has been adopted and applied to temporal \cite{CiuLor17} and epistemic \cite{Bro11} reasoning, 
 as well as to the investigation of legal concepts in AI \cite{LorSar15}. Most notably, $\sti$ has been employed for the analysis of \textit{normative choice}, viz.\ multi-agent choice on the basis of obligations, prohibitions, and permissions. Logics dealing with such concepts are commonly referred to as \textit{deontic}, i.e.\ relating to duty. In fact, the 
  potential of applying $\sti$ to the deontic setting was recognized from the outset \cite{Bar93,%BelPer88,
  HorBel95}.
  More recently, the potential of deontic $\sti$ for reasoning in the context of autonomous vehicles has been identified and explored \cite{ArkBriBel05,SheAbb21}.

The work 
by \citeauthor{Hor01} \citeyear{Hor01} set the stage for the investigation of deontic $\sti$ as a study in its own right. Horty's deontic $\sti$ continues to receive considerable attention in the literature, and it is often recognized that the framework could serve as a suitable basis for developing automated reasoning tools within the context of agent-based, normative reasoning. Despite an increasing demand for automated normative reasoning tools for artificial agents~\cite{ArkBriBel05,BerLyo19a,SheAbb20}, such tools
have yet to be satisfactorily developed for modal logics of agency. 
Although \citeauthor{ArkBriBel05}\ \citeyear{ArkBriBel05} obtained 
 preliminary results on automated reasoning for deontic $\sti$, three essential components 
 were left unaddressed:  
correct termination of the implementation, constructive generation of explicit proofs, and explainability in the case of failed proof-search. In response, this article provides reasoning tools for 
deontic $\sti$ logics, which address all three components. 
 In particular, we focus on Horty's \citeyear{Hor01} deontic $\sti$ logic for \textit{dominance ought}, which deals with a utilitarian evaluation of agent-dependent obligations. The well-grounded philosophical motivation of dominance ought 
 made the framework the most extensively studied deontic $\sti$ logic with a 
wide variety of applications. 

Murakami \citeyear{Mur05} provided a sound a complete axiomatization  for the logic of dominance ought and 
\citeauthor{BerLyo19b} \citeyear{BerLyo19b} showed how utilitarian $\sti$ models---those initially provided for dominance ought by  Horty \citeyear{Hor01}---can be transformed into equivalent relational models. The first upshot of the latter result is that we can safely deal with deontic $\sti$ in the more modular relational setting, omitting reference to utilities. 
The second upshot is that the alternative semantics facilitates the development of an effective 
\emph{proof-theory} for deontic $\sti$, i.e.\ the formulation of a set of rules (called a \emph{calculus} or \emph{proof system}) capable of deriving all and only the theorems of a particular deontic $\sti$ logic, and which may be leveraged for applications or to study (meta-)logical properties. Preliminary works in the proof-theory of (deontic) $\sti$ logics, provided in \cite{Wan06,BerLyo19a,LyoBer19,Negri2020,BerLyo21,Lyo21thesis}, show the potential of proof theory in relation to automated reasoning for deontic $\sti$ logics, which is the primary concern of this article.

Our work is of a twofold nature, reflected by the following 
two aims.  
The first aim is our central aim and constitutes the main contribution of this article: 
\begin{quote}
(AIM I) Supply automated reasoning procedures for the deontic $\sti$ logic of dominance ought (and variations thereof) which produce certificates of the (in)validity of input formulae.
\end{quote}
As mentioned above, the automation of $\sti$ reasoning is motivated by the need for explicit normative reasoning in agentive settings. For this reason, aim I explicates the need for certificates, i.e.\ the reasoning procedure must generate explicit proofs of successful proof-search and counter-models in the case of failed proof-search. Such certificates can be seen as explicit \textit{reasons}
 that explain an input's (in)validity. This relates to our second aim: 
\begin{quote}
(AIM II) Illustrate how the obtained calculi and proof-search algorithm can be utilized for   
agent-based normative reasoning tasks. 
\end{quote}
We demonstrate how the developed proof systems and accompanying proof-search algorithm can be employed for reasoning tasks concerning obligations and choices.
 The constructed proofs and counter-models will serve as certificates that provide reasons for an agent's obligations and choices. The upshot of this approach, is that such certificates can be accessed, analyzed, and assessed, which are important aspects of more transparent and explainable normative reasoning in AI.

Let us  discuss these two aims in more detail.

\subsection{Aim I: Proof Theory}\label{sect1a_aimi}

Modern proof theory has its roots in the work of Gerhard Gentzen \citeyear{Gen35a,Gen35b} who introduced the \emph{sequent calculus} framework for classical and intuitionistic logic. The characteristic feature of the sequent calculus framework is the utilization of \emph{sequents} in defining inference rules and deriving theorems (sequents are formulae of the form $\phi_{1}, \ldots, \phi_{n}\sar \psi_{1}, \ldots, \psi_{k}$, interpreted as stating that if all formulae $\phi_{i}$ in the antecedent hold, then some formula $\psi_{j}$ in the conclusion holds). The framework proved to be well-suited for the construction of \emph{analytic calculi} (i.e.\ calculi that possess the \emph{subformula property}), meaning that every formula used to \emph{reach} the conclusion of a proof, \emph{occurs} as a subformula in the conclusion of the proof. 
The 
analyticity of a calculus is highly valuable in designing proof-search algorithms that take a formula as input and attempt to construct a proof by applying inference rules in reverse (typically allowing for a counter-model to be constructed if a proof is not found).

Despite its success and utility, 
it was found that the sequent calculus framework is often not expressive enough to construct analytic calculi for many logics of interest, e.g.\ see 
\cite{Wan02}. In response, a diverse number of proof-theoretic formalisms 
 extending Gentzen's sequent calculus framework were introduced, 
 allowing for the construction of analytic calculi for larger classes of logics. Such formalisms include (prefixed) tableaux~\cite{Fit72,Fit14}, display calculi~\cite{Bel82,Wan94}, hypersequent calculi~\cite{Avr96,Lah13}, labeled calculi~\cite{Gab96,Sim94,Vig00}, and nested calculi~\cite{Bul92,Kas94,Bru09,Pog09}.

In this paper, we utilize proof systems for deontic $\sti$ logics within the \textit{labeled sequent} formalism \cite{Sim94,Gab96,Vig00,Hei05,Neg05}. This formalism is particularly suitable for handling modal logics, since it allows for semantic information to be explicitly represented in the syntax of sequents. Its name refers  
to the use of labels, representing worlds in a relational model, that prefix formulae. 
The proof systems used in this paper first appeared in the PhD thesis of \citeauthor{Lyo21thesis} \citeyear{Lyo21thesis} and are extensions/variants of the labeled calculi for (deontic) $\sti$ logics provided by both authors in a sequence of papers~\cite{BerLyo19a,LyoBer19,BerLyo21}. In these works,  
 the problem of designing proof-search algorithms for deontic $\sti$ logics 
 was left open; this problem will be addressed 
 in this paper.
 
We adopt 
 the labeled formalism since it offers a variety of advantages. 
First, building calculi within the labeled formalism is straightforward as calculi are easily obtained by transforming the semantic clauses and frame properties associated with the relational semantics of a logic into inference rules~\cite{Sim94,SchTis11}. Additionally, the labeled formalism provides uniform and modular presentations for many classes of logics, where the addition, deletion, or modification of inference rules from one calculus yields a calculus for another logic within the considered class~\cite{Sim94,Vig00}. This modularity proves useful in our setting as the addition/deletion of inference rules 
transforms our proof systems into sound and complete proof systems for various (non-)deontic $\sti$ logics, including 
the non-deontic 
traditional $\sti$ logics developed by Belnap, Perloff, and Xu \citeyear{BelPerXu01}. 
 Consequently, all of our results straightforwardly transfer to the non-deontic, multi-agent setting. 
Last, labeled calculi typically exhibit favorable proof-theoretic properties such as the height-preserving invertibility of rules and cut admissibility~\cite{Sim94,Vig00,Neg05}. 
We leverage such properties to successfully realize aim I.

\subsection{Aim II: Normative Reasoning Tasks and Explanations}\label{sect1a_aimii}

The constructive stance taken in this work is deliberate. As our proof-search algorithm also supports counter-model extraction, we automate the construction of explicit proofs of derivability and underivability. Our approach, thus, yields insight into the reasons for a formula's (non-)theoremhood, providing 
information that explains logical inferences relative to agents, choices, and obligations. This is essential for the secondary aim of our work.

Explanation is a critically important topic in AI \cite{BirCot17,MilHowSon17,Gunetal19,Mil19}. It allows us to understand and evaluate why intelligent systems---such as artificial agents---produce particular decisions, predictions, and actions. An often observed consequence of improved explainability is that it increases our trust in such systems. 
 Especially in light of the rapid increase of autonomous intelligent systems, we also find an increasing demand for formal approaches to explainable agency, which is the ability of intelligent agents to explain their choices and actions \cite{LanMaeSriCho17}. Similarly, the demand for explanatory methods for normative reasoning has increased \cite{BerStr22}. In particular, the work by \citeauthor{BerStr24} \citeyear{BerStr24} proposes dialogue models to construct, what are called, deontic explanations for defeasible normative reasoning. The deontic formalisms in these works do not contain explicit concepts of agency.

%Similarly, demand for explanatory methods for normative reasoning increases \cite{BerStr22}.

Explanation is strongly connected to notions such as justification and interpretability \cite{BirCot17}. Following \citeauthor{YeJoh95} \citeyear{YeJoh95}, a `justification’ denotes a type of explanation that gives an “explicit description of the causal argument or rationale behind each inferential step taken” (p.~158)\footnote{This is a common definition 
in AI. In ethics, however, `explanation’ and `justification’ are often considered 
distinct, namely, the former providing reasons for why something happened in the way it happened and the latter providing reasons for why it is good that it happened \cite{Alv17}.} and `interpretability’ signifies, in the words of \citeauthor{BirCot17} \citeyear{BirCot17} that  a system’s “operations can be understood by a human, either through introspection or through a produced explanation” (p.~8). Hence, 
it can be seen that constructive proof-search procedures 
have explanatory value by yielding justifications in terms of proofs---i.e.\ transparent representations of a step-by-step inferential procedure---and facilitating interpretability via explicit inferential operations that can be inspected.

Logical reasoning is suitable for transparent reasoning tasks, avoiding the opaqueness issues of black box reasoning as often ascribed to machine learning procedures \cite{Gunetal19}.  Although such logical methods might not be optimal for real-time reasoning in which decisions must be made fast, they are particularly useful for transparent reasoning procedures for which we can formally specify various reasoning tasks 
and facilitate explanation in the way defined above. For instance, in this paper, we illustrate how the proof theory of deontic $\sti$ can generate certificates 
of potential (non-)compliance of agents' behavior with respect to a set of obligations. 
 By providing \textit{procedural} proof-search algorithms with automated counter-model extraction, we address 
three core demands for automated agent-based normative reasoning: 
correctness and termination, 
transparency of reasoning, and 
explainability in terms of justification and interpretability. 

We also apply our deontic $\sti$ calculi to address three normative reasoning tasks: 
First, 
our formalism can be applied to determine an agent’s obligations given a certain situation in which the agent resides. That is, given a situation---i.e.\ a given knowledge base consisting of obligations and facts---the agent can check whether they are bound by certain obligations. We refer to this application as \textit{duty checking}. Second, 
the calculi can be used to determine whether a certain choice, considered by the agent as potential conduct, complies with the given knowledge base of obligations and facts. We refer to this second application as \textit{compliance checking}. Third, 
the calculi can be employed to check whether under a specified factual context an agent can (still) jointly fulfill all of their obligations. We call this application \textit{joint fulfillment checking}. These three applications of the framework demonstrate three conceptually different normative reasoning tasks. 
Furthermore, we discuss the formal similarities between these 
applications from the perspective of proof theory.

\subsection{Related Work}

There have been previous attempts to automate or mechanize reasoning with (deontic) $\sti$ logics. An important first step was made by  \citeauthor{ArkBriBel05}\ \citeyear{ArkBriBel05} who encoded a natural deduction system for Horty's deontic $\sti$ in Athena, allowing for the formal verification of proofs, but omitting the specification of a terminating proof-search algorithm. \citeauthor{SheAbb21} \citeyear{SheAbb21} wrote a model-checking algorithm for a variant of Horty's deontic $\sti$ combined with $\mathsf{CTL}^{\ast}$, which takes a model and obligation as input and determines whether the obligation holds. This significantly differs from our algorithm which takes an arbitrary formula as input and decides the formula, outputting either a proof or counter-model thereof. Proof-theoretic approaches to formal reasoning with non-deontic $\sti$ logics include the tableau systems of \citeauthor{Wan06} \citeyear{Wan06} and the labeled systems of \citeauthor{Negri2020} \citeyear{Negri2020}. The former work leaves the specification of a tableau-based decision algorithm open, while the latter work takes a similar approach to the authors' work in this and previous papers~\cite{LyoBer19}, providing a labeled calculus for non-deontic $\sti$ logic and attempting to secure decidability via terminating proof-search. Dalmonte et al.\ \citeyear{Daletal21} develop proof-search with counter-model extraction for a class of non-$\sti$ deontic
logics.

As the work of \citeauthor{Negri2020} \citeyear{Negri2020} shares similarities with ours, the authors would like to emphasize the main differences between the labeled calculi applied in this paper and those defined by \citeauthor{Negri2020} \citeyear{Negri2020} for non-deontic multi-agent $\sti$ logic: 
 (i) We consider a more general class of $\sti$ calculi, which includes the deontic extension of traditional multi-agent $\sti$ \cite{Hor01,Mur05} with and without the $\sti$ principle for limiting the amount of choices available to agents \cite[Ch. 17]{BelPerXu01}. (ii) In contrast to \cite{Negri2020}, our calculi adopt a \textit{relational} semantics for $\sti$, forgoing the more involved---yet redundant \cite{HerSch08}---branching-time structures in the context of 
atemporal 
$\sti$ logics. As a consequence of using relational semantics, our calculi are syntactically simpler and more economical which, for instance, significantly decreases the number of rules in the calculi.
 (iii) We implement an explicit loop-checking mechanism for our proof-search algorithm to ensure correct and terminating proof-search for (deontic) multi-agent $\sti$ logics (with a limited as well as unlimited choice constraint). At the end of Section~\ref{Sect:Calculi}, we provide an argument as to why proof-search algorithms for (non-deontic) $\sti$ logics without loop-checking are susceptible to problematic cases (i.e.\ termination may fail; e.g. in
\citeauthor{Negri2020} \citeyear{Negri2020}). 

Last, 
(non-)deontic $\sti$ logics were shown decidable via the finite model property by~\citeauthor{BelPerXu01}\ \citeyear{BelPerXu01}  and \citeauthor{Mur05} \citeyear{Mur05}, respectively. \citeauthor{BalHerTro08} \citeyear{BalHerTro08} studied the complexity of deciding traditional $\sti$ logics, finding that single-agent $\sti$ logic is $\mathrm{NP}$-complete and multi-agent $\sti$ is $\mathrm{NEXPTIME}$-complete. These approaches are model-theoretic, 
whereas our work is 
the first to establish the decidability of (non-)deontic $\sti$ logics \emph{syntactically} by means of terminating proof-search.

\subsection{Outline of the Paper}

\sect~\ref{Sect:Logics} introduces deontic $\sti$ logics and \sect~\ref{Sect:Calculi} introduces their cut-free labeled sequent calculi. \sect~\ref{Sect:Proof-search} is dedicated to a discussion of our 
proof-search algorithm, and 
the proposed loop-checking mechanism, which lets us decide (non-)deontic multi-agent $\sti$ logics with (un)limited choice.
In \sect~\ref{Sect:Applications}, we apply proof-search to the three agent-based normative reasoning tasks.  
We conclude and discuss future work in \sect~\ref{Sect:Conclusion}.

\section{Logical Preliminaries}\label{Sect:Logics}

Deontic $\sti$ logics enable reasoning about agents' choices and obligations in multi-agent scenarios by employing a variety of modal operators. In this section, we discuss the interpretation of 
our deontic $\sti$ languages and provide a formal semantics for its formulae.

\textit{Choices.} Different choices may be available to different agents at different moments in time. The characteristic feature of traditional $\sti$ logic is the use of an instantaneous  \emph{choice} operator $\agbox$ for each agent $i$, which informally expresses that ``agent $i$ sees to it that'' (some proposition holds). The operator is instantaneous in the sense that choice refers to what an agent can directly see to at a given moment. In a multi-agent world, a single agent cannot uniquely determine the future by acting. For instance, when I decide to go to a concert, it may be that my friends join me, but also they may stay at home. Nevertheless, if I see to it that I go to the concert, this excludes a future continuation where I stayed at home. Hence, what an agent can do via exercising choice is to constrain or limit the possible courses of events; the modal operator $\agbox$ models this idea.

\textit{Settledness.} Certain states of affairs cannot be altered by any of the agents' (joint) choices at a given moment. Such states of affairs are \textit{settled true} for that moment. Basic $\sti$ logic includes a \emph{settledness} operator $\Box$. For instance, 
that it is currently settled true that it is Tuesday, means  
that no choice is available to any of the agents to see to it that today is \textit{not} Tuesday. In such cases, we sometimes say that what is settled true is realized independently of any of the agents' choices. The settledness operator plays an essential role in characterizing the relations between different choices of agents. 
Following \citeauthor{HorBel95} \citeyear{HorBel95}, we focus on single-moment scenarios.

\textit{Obligations.} In the context of $\sti$, obligations prescribe certain choices over others and are captured by an agent-specific \emph{deontic choice} operator $\otimes_{i}$ for each agent $i$. We interpret $\otimes_{i}$ as 
``agent $i$ ought to see to it that'' (some proposition holds) \cite{Hor01}. 
To illustrate, the formula $\otimes_{i} \texttt{concert}$ is informally read as ``agent $i$ ought to see to it that she attends the concert'' (e.g.\ because she made a promise to a friend). 

For each of the above modal operators (viz.\ $\agbox$, $ \Box$, and $\Oi$, respectively), we include the \emph{duals} (viz.\ $\agdia$, $\Diamond$, and $\ODi$). 
We read $\agdia \phi$ as `the state of affairs $\phi$ may result from a choice made by agent $i$' and take $\Diam\phi$ to express that `the state of affairs $\phi$ is currently possible'. We read $\ODi \phi$ as `agent $i$ is permitted to see to it that the state of affairs $\phi$ holds' (cf. permission as the dual of obligation \cite{HilMcN13}). Beyond the aforementioned modalities, our languages also include disjunction $\lor$, conjunction $\land$, and classical negation. The multi-agent language $\lang_{n}$ with $n \in \mathbb{N}$ is defined as follows:

\begin{definition}[The Language $\lang_{n}$]\label{def:language} Let $\Ag := \{0,1, \ldots ,n\} \subset \mathbb{N}$ be a finite set of agent labels 
and let $\Var :=\{p_0,p_1,p_2, \ldots\}$ be a denumerable set of propositional variables. The language $\lang_{n}$ is defined via the following grammar in BNF: 
$$
\phi ::= p \ | \ \negnnf{p} \ | \ \phi \lor \phi \ | \ \phi \land \phi \ | \ \Box \phi \ | \ \Diamond \phi \ | \ [i] \phi \ | \ \agdia \phi \ | \ \Oi\phi \ | \ \ODi \phi
$$
where $i\in \Ag$ and $p \in \Var$.
\end{definition}

\newcommand{\rightj}{\texttt{right}\_\texttt{jade}}
\newcommand{\rightk}{\texttt{right}\_\texttt{kai}}
\newcommand{\leftj}{\texttt{left}\_\texttt{jade}}
\newcommand{\leftk}{\texttt{left}\_\texttt{kai}}
\newcommand{\col}{\texttt{coll}}

Our formulae in $\lang_{n}$ are in negation normal form, which allows us to simplify the sequents employed in our calculi and reduce the number of rules. 
 This means that we restrict applications of negation to propositional variables. 
Then, we define the negation $\negnnf{\phi}$ of an arbitrary formula $\phi \in \lang_{n}$ to be the formula where every operator $\lor$, $\land$, $\Box$, $\Diamond$, $\agbox$, $\agdia$, $\Oi$, $\ODi$ is replaced with its dual $\land$, $\lor$, $\Diamond$, $\Box$, $\agdia$, $\agbox$, $\ODi$, $\Oi$ (respectively), every propositional variable $p$ is replaced with its negation $\negnnf{p}$, and $\negnnf{p}$ is replaced with its positive version $p$. We define $\phi \rightarrow \psi := \negnnf{\phi} \lor \psi$, $\phi \leftrightarrow \psi := (\phi \rightarrow \psi) \land (\psi \rightarrow \phi)$, $\top := p \lor \negnnf{p}$, and $\bot := p \land \negnnf{p}$. In what follows, we refer to the \textit{complexity} of a formula in the usual way, 
corresponding to the number of symbols in a formula. 

To illustrate $\lang_n$, 
let Jade ($j$) be an agent cycling to her office in London. Then, the formula $\otimes_j\leftj\land\Diam\leftj\land [j]\lnot\leftj$ describes that 
``Jade is obliged to cycle on the left-hand side of the road, it is realizable that she cycles on the left side, and she chooses not to do so,'' in other words, 
Jade sees to it that her obligation is violated.

We interpret $\mathcal{L}_n$ formulae over special types of relational models, called 
\emph{$\dsnk$-models}. The parameter $n$ refers to the number of agents involved and $k$ to the maximal number of choices for each agent at each moment, and which imposes no maximum when $k=0$. 
Since our logic concerns instantaneous (i.e.\ atemporal) choice-making, it suffices to consider single-moment frameworks \cite{BalHerTro08}, forgoing the use of 
traditional branching time structures often employed in atemporal \textsf{STIT} logics~\cite{BelPerXu01} (the 
$\sti$ language $\lang_n$ is not expressive enough to reason about branching time structures).

\begin{definition}[$\dsnk$-frames, -models]\label{def:frames-models} Let $n \in \mathbb{N}$ and for each agent label $i \in \Ag$, let us define $\R_{[i]}(w) := \{v\in W \ | \ (w,v) \in R_{[i]}\}$. A \emph{$\dsnk$-frame} is defined to be a tuple of the form $F := (W, \{\R_{[i]} \ | \ i \in Ag\}, \{ \opt_{\Oi} \ | \ i \in Ag\} )$ with $W$ a non-empty set of worlds $w,v,u,\ldots$ and:
\begin{center}
\begin{tabular}{p{1em} p{30pt} p{450pt}}
 & {\rm \textbf{(C1)}} & For all $i\in Ag$, $\R_{[i]} \subseteq W{\times} W$ is an equivalence relation. \\
 & {\rm \textbf{(C2)}} & For all $u_{1}, \ldots ,u_{n} \in W$, $\bigcap_{i \in Ag} \R_{[i]}(u_{i}) \neq \emptyset$. \\
 & {\rm \textbf{(C3)}} & If $k > 0$, then for all $i \in \Ag$ and $w_{0}, w_{1}, \ldots, w_{k} \in W$, \\
 & & $$\displaystyle{ \bigvee_{0 \leq m \leq k-1 \text{, } m+1 \leq j \leq k} R_{[i]} w_{m}w_{j}}.$$
\end{tabular}
\end{center}
\begin{center}
\begin{tabular}{p{1em} p{30pt} p{450pt}}
 & {\rm \textbf{(D1)}} & For all $i\in Ag$, $\opt_{\Oi}\subseteq W$.\\

 & {\rm \textbf{(D2)}} & For all $i\in Ag$, $\opt_{\Oi} \neq \emptyset$. \\
 & {\rm \textbf{(D3)}} & For all $i \in Ag$ and $w, v \in W$, if $w \in \opt_{\Oi}$ and $v \in R_{\agbox}(w)$, then $v \in \opt_{\Oi}$.\\
\end{tabular}
\end{center}
A \emph{$\dsnk$-model} is a tuple $M := (F,V)$ where $F$ is a frame and $V{:}\ \Var \to \mathcal{P}(W)$ is a valuation function mapping propositional variables to subsets of $W$. 
\end{definition}

Following \dfn~\ref{def:frames-models}, a $\dsnk$-model consists of a set 
$W$ representing a single moment in time at which the agents from $\Ag$ are making decisions and where each world in $W$ designates a 
different configuration of states of affairs that could hold at that moment. Conditions (\textbf{Ci}) and (\textbf{Di}) refer, respectively, to the \emph{choice properties} and \emph{deontic properties}.  
 Condition \textbf{(C1)} partitions $W$ into \emph{choice-cells} (i.e.\ equivalence classes) for each $i \in \Ag$, where each choice-cell represents a choice available to agent $i$ at the considered moment.
 The partition imposed on $W$ must satisfy the \emph{independence of agents} condition, and may satisfy a \emph{limited choice} condition \cite{BelPerXu01}. Independence of agents, expressed by condition \textbf{(C2)}, stipulates that no agent can keep another agent from exercising an available choice at the moment of choice, i.e.\label{ioa}\  regardless of the choices made by the agents, some set of possible states of affairs ensues. Formally,  \textbf{(C2)} ensures that any combination of choices made by the agents is consistent, i.e.\ if we select a choice-cell for each agent, their intersection 
is non-empty.
 The limited choice condition, captured by \textbf{(C3)}, stipulates that each agent has a maximum of $k$ choices available to choose from at each moment and is only enforced when $k > 0$ (i.e.\ 
 imposing no limitation on the number of choices when $k=0$). Formally, \textbf{(C3)} states that if the limit is $k$, then for $k+1$ many worlds at least two worlds must be in the same choice-cell.\footnote{
 \citeauthor{BelPerXu01} \citeyear[Ch. 9]{BelPerXu01} argue that one may want to limit the amount of choices 
 due to the possibility of `busy choosers’, which are agents making infinitely many choices in a finite period. For some $\sti$ logics, busy choosers have undesirable effects, e.g.\ where refraining from refraining does not equal doing.}
 The condition \textbf{(D1)} ensures that all deontically optimal worlds at the present moment for any agent $i$ (which are those worlds contained in $\opt_{\Oi}$) are realizable. Condition \textbf{(D2)} ensures that, for each agent $i$, at least one deontically optimal world exists at a moment, and \textbf{(D3)} states that if a deontically optimal world exists in a choice, then every world in that choice is deontically optimal. In other words,
\textbf{(D1)}--\textbf{(D3)} ensure the existence of an obligatory 
choice for each agent $i$ at a given moment. Last, each $\dsnk$-model contains a \emph{valuation function} $V$ mapping propositional variables to sets of worlds.

\begin{definition}[Semantics]\label{def:semantics} Let $M$ be a $\dsnk$-model and let $w\in W$ of $M$. The \emph{satisfaction} of a formula $\phi\in \lang_{n}$ in $M$ at $w$ is recursively defined accordingly:

\begin{itemize}
\vspace{-0.2cm}
\item[1.] $M, w \Vdash p$ \iffi $w \in V(p)$

\item[2.] $M, w \Vdash \negnnf{p}$ \iffi $w \not\in V(p)$

\item[3.] $M, w \Vdash \phi \wedge \psi$ \iffi $M, w \Vdash \phi$ and $M, w \Vdash \psi$

\item[4.] $M, w \Vdash \phi \vee \psi$ \iffi $M, w \Vdash \phi$ or $M, w \Vdash \psi$

\item[5.] $M, w \Vdash \Box \phi$ \iffi $\forall u \in W$, $M, u \Vdash \phi$

\item[6.] $M, w \Vdash \Diamond \phi$ \iffi $\exists u \in W$, $M, u \Vdash \phi$

\item[7.] $M, w \Vdash [i] \phi$ \iffi $\forall u \in \R_{[i]}(w)$, $M, u \Vdash \phi$

\item[8.] $M, w \Vdash \lb i \rb \phi$ \iffi $\exists u \in \R_{[i]}(w)$, $M, u \Vdash \phi$

\item[9.] $M, w \Vdash \Oi \phi$ \iffi $\forall u \in \opt_{\Oi}$, $M, u \Vdash \phi$

\item[10.] $M, w \Vdash \ODi \phi$ \iffi $\exists u \in \opt_{\Oi}$, $M, u \Vdash \phi$
\end{itemize}
\vspace{-0.2cm}
 A formula $\phi$ is \emph{$\dsnk$-valid} (written $\Vdash_{\dsnk} \phi$) \iffi for every $\dsnk$-model $M$, it  is satisfied at every world in the domain $W$ of $M$. 
We define the logic $\dsnk$ as the set of $\dsnk$-valid formulae.
\end{definition}

The semantic evaluation of the propositional connectives implies that our logic extends 
classical propositional logic. 
We note in passing that sound and complete axiomatizations of $\dsnk$ with respect to the above semantic characterization 
are provided by \citeauthor{Lyo21thesis} \citeyear{Lyo21thesis}.

Let us consider an example. Suppose there are two agents, Jade and Kai (where $Ag:=\{j,k\}$), who are cycling toward each other on a two lane road.  Both have two choices: keep cycling left or keep cycling right. Let $\leftj$ and $\rightj$, and $\leftk$ and $\rightk$ denote the states of affairs in which Jade and Kai cycle left and right, respectively. In the graphical illustration of this model (see \fig~\ref{fig:model-example}), the set $W = \{w_{1}, w_{2}, w_{3}, w_{4}\}$ represents a 
moment in time with four possible situations, e.g.\  $w_{1}$ denotes the situation in which both Jade and Kai cycle on the left.  
The model in \fig~\ref{fig:model-example} shows that Jade has two choices, one consisting of the worlds
$w_{1}$ and $w_{3}$, and one consisting of 
$w_{2}$ and $w_{4}$, whereas Kai has one choice consisting of 
$w_{1}$ and $w_{2}$, and one consisting of 
$w_{3}$ and $w_{4}$. Graphically, the horizontal choice-cells outlined with dotted lines `\textbf{\textperiodcentered } \textbf{\textperiodcentered } \textbf{\textperiodcentered }' represent Kai's choices and the vertical choice-cells outlined with dashed lines `$\textbf{\-- \-- \--}$' represent Jade's choices. There are two cases in which a collision between Jade and Kai can be avoided: when both keep cycling left and when both keep cycling right. These two scenarios are represented by the intersecting choices leading to $w_1$, respectively $w_4$ in Figure~\ref{fig:model-example}. A collision $\col$ holds at those intersecting choices where one of the two cycles left and the other right.\footnote{Avoiding a collision is expressed as: $\Box(\lnot\col\rightarrow((\rightj\land\rightk)\lor(\leftj\land \leftk)))$).}
In \fig~\ref{fig:model-example} the propositional variable $\leftj$ is mapped to the set $\{w_{1},w_{3}\}$, meaning that in these possible worlds Jade cycles left. 
The negation $\negnnf{p}$ of a propositional variable $p$ holds at those worlds where $p$ does \emph{not} hold (e.g.\ $\lnot\col$ holds at $w_{1}$ and $w_{3}$ in our example model).

\begin{figure}[t]
\label{fig:model-example}
\begin{center}
\scalebox{0.92}{
\begin{tikzpicture}

\node[%point
] (w1) [] {$w_1$}; 

\node[%point
] (w2) [right=of w1, xshift=0mm,yshift=0mm] {$w_2$};

\node[%point
] (w3) [below=of w1, xshift=0mm,yshift=2mm] {$w_3$};

\node[%point
] (w4) [below=of w2, xshift=0mm,yshift=2mm] {$w_4$};

\node[draw, very thick, %dotted,
 color=black, rounded corners, inner xsep=15pt, inner ysep=15pt, fit=(w1) (w4)] (m1) {}; %C

\node[very thick, rounded corners,  inner xsep=10pt, inner ysep=10pt, fit=(w1) (w3), pattern=north west lines, pattern color=Gray2] (c11) {}; %C

\node[draw,  thick, densely dashed, color=black, rounded corners, inner xsep=10pt, inner ysep=10pt, fit=(w1) (w3)] (c1) {}; %C

\node[draw,  thick, densely dashed, color=black, rounded corners, inner xsep=10pt, inner ysep=10pt, fit=(w2) (w4)] (c2) {};

\node[%preaction={fill, red},
 very thick, %purple
, rounded corners,  inner xsep=5pt, inner ysep=5pt, fit=(w1) (w2), pattern=north east lines, pattern color=Gray2] (c12) {}; %C

\node[draw,  thick, dotted, color=black %purple
, rounded corners, inner xsep=5pt, inner ysep=5pt, fit=(w1) (w2)] (c1) {}; %C

\node[draw,  thick, dotted, color=black %purple
, rounded corners, inner xsep=5pt, inner ysep=5pt, fit=(w3) (w4)] (c2) {};

\node[%point
] (w11) [] {$w_1$}; 

\node[%point
] (w21) [right=of w1, xshift=0mm,yshift=0mm] {$w_2$};

\node[%point
] (w31) [below=of w1, xshift=0mm,yshift=2mm] {$w_3$};

\node[%point
] (w41) [below=of w2, xshift=0mm,yshift=2mm] {$w_4$};

\node[] (m2) [above=of m1,xshift=-45mm, yshift=5mm] {};

\node[] (m3) [above=of m1,xshift=-15mm, yshift=10mm] {\textcolor{white}{10}};

\node[] (m4) [above=of m1,xshift=15mm, yshift=10mm] {};

\node[] (m5) [above=of m1,xshift=45mm, yshift=5mm] {};

\node[] (root) [below=of m1] {};

\path[-] (root) edge[ thick,black]  (m1);

\path[-] (w1) edge[ thick,black] node[left,xshift=-2.5mm,yshift=9mm] {$\leftj$} node[left,xshift=-1mm,yshift=4mm] {$\leftk$}  (m3);
\path[-] (w2) edge[ thick,black] node[right,xshift=2.5mm,yshift=9mm] {$\rightj$} node[right,xshift=1mm,yshift=4mm] {$\leftk$} 
node[right,xshift=0mm,yshift=-1mm] {$\col$}
(m4);
\path[-] (w3) edge[ thick,black] node[left,xshift=-6mm,yshift=5mm] {$\leftj$} node[left,xshift=-1mm,yshift=0mm] {$\rightk$} 
node[left,xshift=4mm,yshift=-4mm] {$\col$} (m2);
\path[-] (w4) edge[ thick,black] node[right,xshift=5.5mm,yshift=5mm] {$\rightj$} node[right,xshift=1mm,yshift=0mm] {$\rightk$} (m5);
\end{tikzpicture}
}
\end{center}
\vspace{-0.5cm}
\caption{An illustration of the $\baselogic_{n}$-model $M := (F,V)$ built atop the $\baselogic_{n}$-frame $F := (W,\{R_{[j]}, R_{[y]}\},\{\opt_{\otimes_{j}}, \opt_{\otimes_{y}}\})$ such that $W := \{w_{1}, w_{2}, w_{3}, w_{4}\}$, $R' := \{(w_{i},w_{i}) \ | \ 1 \leq i \leq 4\}$, Jade's relation $R_{[j]} := \{(w_{1}, w_{3}), (w_{3}, w_{1}), (w_{2}, w_{4}), (w_{4}, w_{2})\} \cup R'$, Kai's relation $R_{[k]} := \{(w_{1}, w_{2}), (w_{2}, w_{1}), (w_{3}, w_{4}), (w_{4}, w_{3})\} \cup R'$, $\opt_{\otimes_{j}} := \{w_{1},w_{3}\}$, $\opt_{\otimes_{k}} := \{w_{1},w_{2}\}$, and  $V := \{(\leftj,\{w_{1},w_{3}\}),(\rightj,\{w_{2},w_{4}\}),\allowbreak(\leftk,\{w_{1},w_{2}\}),\allowbreak(\rightk,\{w_{3},w_{4}\}),(\col,\{w_{2},w_3\})\}$.
}
\end{figure}

  Although there are two outcomes in the cycling scenario of Figure~\ref{fig:model-example} in which a collision is avoided, we assume that UK traffic law obliges cyclists to always cycle on the left-hand side of the road. In Figure~\ref{fig:model-example}, Jade's choice consisting of $w_{1}$ and $w_{3}$ is obligatory for Jade and Kai's choice consisting of $w_{1}$ and $w_{2}$ is obligatory for Kai. Thus, based on the assumed traffic laws, if both Jade and Kai commit to their obligatory choices, they jointly see to it that a collision is avoided. In Figure~\ref{fig:model-example}, obligatory choices are shaded. In relation to this, due to fact that $\dsnk$-models ensure that at least one obligatory choice exists for each agent $i \in \Ag$, each agent's set of obligations is consistent, i.e.\ no agent is obliged to see to it that $\phi$ and $\negnnf{\phi}$ for some proposition $\phi$. It can be straightforwardly checked that the example $\dsnk$-model in \fig~\ref{fig:model-example} satisfies all properties from Definition~\ref{def:frames-models}.

\section{Proof Theory for Deontic STIT Logics}\label{Sect:Calculi}

We present a class of labeled sequent calculi, 
 which generalize Gentzen-style sequent systems 
by including semantic information directly in the syntax of sequents. The calculi 
are variations of the calculi  
given by \citeauthor{BerLyo21} \citeyear{BerLyo21}, and were first defined by \citeauthor{Lyo21thesis} \citeyear{Lyo21thesis}. The syntactic expressions 
in our calculi, called \emph{labeled sequents}---or \emph{sequents} for short---encode semantic information by prefixing formulae with labels (representing worlds in a relational model) and include relational atoms (encoding the accessibility relations). 

\begin{definition}[Sequent]\label{def:language_seq} Let $\Ag = \{0,1, \ldots, n\}$ be a finite set of agent labels and let $\lab := \{ w, u, v, \ldots \}$ be a denumerable set of labels. A \emph{sequent} is defined to be an expression of the form $\Lambda := \rel \sar \Gamma$, where $\rel$ is a (potentially empty) finite set of \emph{relational atoms} of the form $\R_{[i]}wu$ and $\opt_{\Oi}w$, and $\Gamma$ is a (potentially empty) finite set of \emph{labeled formulae} of the form $w : \phi$, where $i$ ranges over $\Ag$, $\phi$ ranges over $\lang_{n}$, and $w,u$ range over $\lab$.
\end{definition}

We use $\rel$, $\qel$, $\ldots$ (potentially annotated) to denote sets of relational atoms, and use $\Gamma$, $\Delta$, $\ldots$ (potentially annotated) to denote sets of labeled formulae. We often use $\Lambda := \rel \sar \Gamma$ to denote an entire sequent, where we refer to $\rel$ as the \emph{antecedent} of $\Lambda$, and $\Gamma$ as the \emph{consequent} of $\Lambda$. Intuitively, a sequent $\rel \sar \Gamma$ is interpreted as stating that if all relational atoms in $\rel$ hold, then some labeled formula $w : \phi \in \Gamma$ holds (see \dfn~\ref{def:sequent-semantics-dsn} for the semantic interpretation of a sequent). We use the notation $Lab(\rel)$, $Lab(\Gamma)$, and $Lab(\Lambda)$ to represent the set of labels contained in $\rel$, $\Gamma$, and $\Lambda$, respectively; e.g.\ $Lab(R_{[2]}wv \sar u : p) = \{w,u,v\}$. Furthermore, for a label $w$ and set of labeled formulae $\Gamma$, we define $\Gamma \restriction w := \{\phi \ | \ w : \phi\}$ to be the set of formulae prefixed with the label $w$.

The $\gtdsnk$ calculi 
for each logic $\dsnk$, where $|\Ag| = n \in \mathbb{N}$ and $k \in \mathbb{N}$, are given a uniform presentation in \fig~\ref{fig:calculus}. The $\id$ rule encodes the fact that at any world in a $\dsnk$-model, either a propositional atom $p$ holds or $\negnnf{p}$ holds, and acts as a closure rule during proof-search. 
We refer to any sequent that is the conclusion of $\id$ as an \emph{initial sequent}.  The $\disr$, $\conr$, $\diar$, $\boxr$, $\ODir$, $\Oir$, $\stitdiar$, and $\stitr$ rules are obtained by transforming the semantic clauses of the corresponding logical connectives into inference rules (\dfn~\ref{def:semantics}). The structural rules $\stitrefl$ and $\stiteucl$, taken together, encode the fact that each $R_{\agbox}$ relation is an equivalence relation (i.e.\ is both reflexive and Euclidean) as dictated by condition \textbf{(C1)}. The conditions \textbf{(C2)}, \textbf{(C3)}, \textbf{(D2)}, and \textbf{(D3)} are transformed into the rules $\ioa$, $\choicer$, $\dtwo$, and $\dthree$, respectively, with \textbf{(D1)} encoded in the $\ODir$ rule. 

We note that $\choicer$ is a rule schema encoding that agent $i$ is limited to at most $k$ choices; $\mathsf{APC}$ stands for `axiom scheme for possible choices’ \cite[p. 437]{BelPerXu01}.   
 When $k = 0$, $\choicer$ is omitted from the calculus thus enforcing no upper-bound on the number of choices. 
 When $k > 0$, $\choicer$ contains $k(k+1)/2$ premises $\rel, R_{\agbox}x_{m}x_{j} \sar \Gamma$ with $0 \leq m \leq k-1$ and $m+1 \leq j \leq k$. For instance, if $k = 1$, then $\choicer$ takes the form
\begin{center}
\AxiomC{$\rel, R_{\agbox}w_{0}w_{1} \sar \Gamma$}
\RightLabel{$(\mathsf{APC}^{1}_{i})$}
\UnaryInfC{$\rel \sar \Gamma$}
\DisplayProof
\end{center}
and if $k = 2$, then $\choicer$ takes the form 
\begin{center}
\AxiomC{$\rel, R_{\agbox}w_{0}w_{1} \sar \Gamma$}
\AxiomC{$\rel, R_{\agbox}w_{0}w_{2} \sar \Gamma$}
\AxiomC{$\rel, R_{\agbox}w_{1}w_{2} \sar \Gamma$}
\RightLabel{$(\mathsf{APC}^{2}_{i})$.}
\TrinaryInfC{$\rel \sar \Gamma$}
\DisplayProof
\end{center}

\begin{figure}[t]
\noindent\hrule

\begin{center}
\begin{tabular}{c c} 

\AxiomC{ }
\RightLabel{$\id$}
\UnaryInfC{$\rel \sar w:p, w:\negnnf{p}, \Gamma$}
\DisplayProof

&

\AxiomC{$\rel \sar w: \phi, \Gamma$}
\AxiomC{$\rel \sar w: \psi, \Gamma$}
\RightLabel{$\conr$}
\BinaryInfC{$\rel \sar w: \phi \wedge \psi, \Gamma$}
\DisplayProof

\end{tabular}
\end{center}

\begin{center}
\begin{tabular}{c @{\hskip 2em} c}

\AxiomC{$\rel \sar w: \phi, w : \psi, \Gamma$}
\RightLabel{$\disr$}
\UnaryInfC{$\rel \sar w: \phi \vee \psi, \Gamma$}
\DisplayProof

&

\AxiomC{$\rel,\R_{[1]}w_{1}u,\dots, \R_{[n]}w_{n}u \sar \Gamma$}
\RightLabel{$\ioa^{\dag}$}
\UnaryInfC{$\rel \sar \Gamma$}
\DisplayProof

\end{tabular}
\end{center}

\begin{center}
\begin{tabular}{c @{\hskip 1em} c @{\hskip 1em} c}

\AxiomC{$\rel \sar u: \phi, \Gamma$}
\RightLabel{$\settr^{\dag}$}
\UnaryInfC{$\rel \sar w: \Box \phi, \Gamma$}
\DisplayProof

&

\AxiomC{$\rel \sar w: \Diamond \phi, u: \phi, \Gamma$}
\RightLabel{$\settdiar$}
\UnaryInfC{$\rel \sar w: \Diamond \phi, \Gamma$}
\DisplayProof

&

\AxiomC{$\rel, R_{[i]}ww \sar \Gamma$}
\RightLabel{$\stitrefl$}
\UnaryInfC{$\rel \sar \Gamma$}
\DisplayProof
\end{tabular}
\end{center}

\begin{center}
\begin{tabular}{c c} 

\AxiomC{$\rel, \opt_{\Oi}u \sar w : \ominus_{i} \phi, u : \phi, \Gamma$}
\RightLabel{$\ODir$}
\UnaryInfC{$\rel, \opt_{\Oi}u \sar w : \ominus_{i} \phi, \Gamma$}
\DisplayProof

&

\AxiomC{$\rel, R_{[i]}wu \sar w: \agdia \phi, u:\phi, \Gamma$}
\RightLabel{$\stitdiar$}
\UnaryInfC{$\rel, R_{[i]}wu \sar w: \agdia \phi, \Gamma$}
\DisplayProof

\end{tabular}
\end{center}

\begin{center}
\begin{tabular}{c c c}

\AxiomC{$\rel, \R_{[i]}wu \sar u: \phi, \Gamma$}
\RightLabel{$\stitr^{\dag}$}
\UnaryInfC{$\rel \sar w: [i] \phi, \Gamma$}
\DisplayProof

&

\AxiomC{$\rel, \opt_{\Oi}u \sar u : \phi, \Gamma$}
\RightLabel{$\Oir^{\dag}$}
\UnaryInfC{$\rel \sar w : \Oi \phi, \Gamma$}
\DisplayProof

&

\AxiomC{$\rel, \opt_{\Oi}u \sar \Gamma$}
\RightLabel{$\dtwo^{\dag}$}
\UnaryInfC{$\rel \sar \Gamma$}
\DisplayProof

\end{tabular}
\end{center}

\begin{center}
\begin{tabular}{c @{\hskip 1em} c}

\AxiomC{$\rel, R_{[i]}wu, R_{[i]}wv, R_{[i]}uv \sar \Gamma$}
\RightLabel{$\stiteucl$}
\UnaryInfC{$\rel, \R_{[i]}wu, \R_{[i]}wv \sar \Gamma$}
\DisplayProof

&

\AxiomC{$\rel, R_{[i]}wu, \opt_{\Oi}w, \opt_{\Oi}u \sar \Gamma$}
\RightLabel{$(\mathsf{D3}_{i})$}
\UnaryInfC{$\rel, R_{[i]}wu, \opt_{\Oi}w \sar \Gamma$}
\DisplayProof
\end{tabular}
\end{center}

\begin{center}
\begin{tabular}{c}
\AxiomC{$\Big\{ \rel, R_{\agbox}w_{m}w_{j} \sar \Gamma \ \Big| \ 0 \leq m \leq k-1 \text{, } m+1 \leq j \leq k \Big\}$}
\RightLabel{$\choicer$\index{$\choicer$}}
\UnaryInfC{$\rel \sar \Gamma$}
\DisplayProof
\end{tabular}
\end{center}

\hrule
\caption{The $\gtdsnk$ calculi \cite{Lyo21thesis}, where $|\Ag| = n \in \mathbb{N}$ and $k \in \mathbb{N}$. The superscript $\dag$ on $\settr$, $\stitr$, $\Oir$, $\ioa$, and $(\mathsf{D2}_{i})$ indicates that $u$ is a eigenvariable, i.e.\ it does not occur in the conclusion. We have $\stitdiar$, $\stitr$, $\ODir$, $\Oir$, $\stitrefl$, $\stiteucl$, $\choicer$, $\dtwo$, and $\dthree$ rules for each $i \in \Ag$. $\choicer$ is omitted 
when $k = 0$.} 
\label{fig:calculus}
\end{figure}

A \emph{derivation} (or, \emph{proof}) is built by sequentially applying instances of the rules in $\gtdsnk$ to initial sequents (or to \emph{assumptions} in certain cases), and if a sequent $\Lambda$ is derivable in $\gtdsnk$, i.e.\ $\Lambda$ is the \emph{conclusion}, then we write $\der \Lambda$. 
Each derivation has the form of a tree with the conclusion acting as the root and initial sequents forming the leaves. 
 The \emph{height} of a derivation is the longest path of sequents from the conclusion to an initial sequent in the derivation. The relational atoms and labeled formulae that are explicitly presented in the conclusion of a rule are called \emph{principal}, and those that are explicitly presented in the premises and are not principal are called \emph{auxiliary}.  
Let us consider an example derivation.

\begin{example} We show how to derive $\otimes_i p \rightarrow \Diam [i] p = \ominus_i \negnnf{p} \lor \Diam[i]p$ in $\gtdsnk$.

\begin{center}
\scalebox{0.93}{
\AXC{}
\RL{$\id$}
\UIC{$I_{\otimes_{i}} u, I_{\otimes_{i}} z, R_{[i]}zu \sar u: \negnnf{p}, x: \ODi \negnnf{p}, x: \Diam[i]p, z: [i]p, u:p$}
\RL{$\ODir$}
\UIC{$I_{\otimes_{i}} u, I_{\otimes_{i}} z, R_{[i]}zu \sar x: \ODi \negnnf{p}, x: \Diam[i]p, z: [i]p, u:p$}
\RL{$\dthree$}
\UIC{$I_{\otimes_{i}} z, R_{[i]}zu \sar x: \ODi \negnnf{p}, x: \Diam[i]p, z: [i]p, u:p$}
\RL{$([i])$}
\UIC{$I_{\otimes_{i}} z \sar x: \ODi \negnnf{p},  x: \Diam[i]p, z: [i]p$}
\RL{$\settdiar$}
\UIC{$I_{\otimes_{i}} z \sar x: \ODi \negnnf{p},x: \Diam[i]p$}
\RL{$\dtwo$}
\UIC{$\sar x: \ODi \negnnf{p}, x: \Diam[i]p$}
\RL{$(\lor)$}
\UIC{$\sar x: \ODi \negnnf{p} \lor \Diam[i]p$}
\DisplayProof
}
\end{center}
We point out that $\otimes_{i} p\rightarrow \Diam [i] p$ is an instance of the 
principle of \textit{Ought-implies-Can}, 
which is a central theorem of deontic $\sti$ \cite{Hor01,BerLyo21}.
\end{example}

Since the syntax of our sequents explicitly incorporates semantic information, 
such objects can be interpreted as abstractions of a $\dsnk$-model \cite{Lyo21thesis}. 
 This interpretation (given below) gives 
a correspondence between the semantics of Section~\ref{Sect:Logics} and our calculi.

\begin{definition}[Sequent Semantics] 
\label{def:sequent-semantics-dsn} Let $M = (W, \{R_{[i]} \ | \ i \in \Ag\}, \{\opt_{\Oi} \ | \ i \in \Ag\}, V)$ be a $\dsnk$-model with $I : \ \lab \mapsto W$ an \emph{interpretation function} mapping labels to worlds. The \emph{satisfaction} of a relational atom $R_{\agbox}wu$ or $\opt_{\Oi} w$ (written $M, I \models R_{\agbox}wu$, respectively $M, I \models \opt_{\Oi} w$) and a labeled formula $w : \phi$ (written $M, I \models w : \phi$) is defined as: 
\begin{itemize}

\item $M, I \models R_{\agbox}wu$ \ifandonlyif $(\I{w},\I{u}) \in R_{\agbox}$

\item $M, I \models \opt_{\Oi} w$ \ifandonlyif $\I{w} \in \opt_{\Oi}$

\item $M, I \models w : \phi$ \ifandonlyif $M, \I{w} \Vdash \phi$

\end{itemize}

A labeled sequent $\Lambda := \rel \sar \Gamma$ is \emph{satisfied} in $M$ with $I$ (written, $M,I \models \Lambda$) \ifandonlyif if  $M, I \models R_{\agbox}wu$ and $M, I \models \opt_{\Oi} w$ for all $R_{\agbox}wu, \opt_{\Oi} w \in \rel$, then $M, I \models w : \phi$ for some $w : \phi \in \Gamma$. We say that a
labeled sequent $\Lambda$ is \emph{falsified} in $M$ with $I$ \ifandonlyif $M, I \not\models \Lambda$, that is, $\Lambda$ is not satisfied by $M$ with $I$. Last, a labeled sequent $\Lambda$ is \emph{$\dsnk$-valid}\index{$\dsnk$-valid} (written $\models \Lambda$) \ifandonlyif it is satisfiable in every $\dsnk$ model $M$ with every interpretation function $I$. We say that a
labeled sequent $\Lambda$ is \emph{$\dsnk$-invalid}\index{$\dsnk$-invalid} \ifandonlyif $\not\models \Lambda$, i.e.\ $\Lambda$ is not $\dsnk$-valid.
\end{definition}

As shown in \thrm~\ref{thm:gtdsnk-properties} below, the $\gtdsnk$ calculi for deontic $\sti$ logics possess fundamental properties such as the \textit{height-preserving (hp) admissibility} of label substitutions $\sub$ and weakenings $\wk$ (see \fig~\ref{fig:structural-rules}). Moreover, all rules of the calculi are \textit{height-preserving invertible}, and each calculus admits a syntactic proof of cut admissibility. We briefly clarify the meaning of these properties, which will recur throughout the course of this paper, and after, list 
the properties possessed by each $\gtdsnk$ calculus. 

Let the rule $(r)$ be of the form:
\begin{center}
\AxiomC{$\Lambda_{1}$}
\AxiomC{$\ldots$}
\AxiomC{$\Lambda_{n}$}
\RightLabel{$(r)$}
\TrinaryInfC{$\Lambda$}
\DisplayProof
\end{center}
We say that the rule $(r)$ is (\emph{hp}-)\emph{admissible} in a calculus \iffi if $\Lambda_{1}, \ldots, \Lambda_{n}$ have proofs (with heights $h_{1}, \ldots, h_{n}$, respectively), then $\Lambda$ has a proof (with a height $h \leq \max\{h_{1}, \ldots, h_{n}\}$). We let the $i$-\emph{inverse} of $(r)$, written $(\hat{r}_{i})$, be the rule obtained by taking the conclusion of $(r)$ to be the premise of $(\hat{r}_{i})$ and the $i^{th}$ premise of $(r)$ to be the conclusion of $(\hat{r}_{i})$. We say that an $n$-ary rule $(r)$ is (\emph{hp}-)\emph{invertible} in a calculus \iffi $(\hat{r}_{i})$ is  (hp-)admissible for each $1 \leq i \leq n$. Last, a \emph{label substitution} $(w/u)$ is defined in the usual way, that is, for a labeled sequent $\Lambda$, $\Lambda(w/u)$ is obtained by replacing each occurrence of $u$ by $w$ in $\Lambda$. For example, $(R_{\agbox}wu, \opt_{\Oi}u \sar u : \phi)(v/u) = (R_{\agbox}wv, \opt_{\Oi}v \sar v : \phi)$.

Admissibility and invertibility properties serve a twofold purpose: first, such properties are useful in establishing the completeness of the $\gtdsnk$ calculi (\thrm~\ref{thm:completness-gtdsn}), and second, they are useful in showing the decidability of our logics via proof-search (\sect~\ref{Sect:Proof-search}).

\begin{figure}[t]
\noindent\hrule
\begin{center}
\begin{tabular}{c @{\hskip .5em} c @{\hskip .5em} c} 
\AxiomC{$\rel \sar \Gamma$}
\RightLabel{$\sub$}
\UnaryInfC{$\rel(w/u) \sar \Gamma(w/u)$}
\DisplayProof

&

\AxiomC{$\rel \sar \Gamma$}
\RightLabel{$\wk$}
\UnaryInfC{$\rel,\rel' \sar \Gamma,\Gamma'$}
\DisplayProof

&

\AxiomC{$\rel \sar w : \phi,\Gamma$}
\AxiomC{$\rel \sar w : \negnnf{\phi},\Gamma$}
\RightLabel{$\cut$}
\BinaryInfC{$\rel \sar \Gamma$}
\DisplayProof
\end{tabular}
\end{center}
\hrule
\caption{The set $\strucset$ of structural rules consists of the rules above.}
\label{fig:structural-rules}
\end{figure}

\begin{theorem}[{$\gtdsnk$ Properties ~\cite[Ch.~3.4]{Lyo21thesis}}]\label{thm:gtdsnk-properties}\label{thm:completness-gtdsn}\label{thm:soundness-gtdsn}
Let $|\Ag| = n \in \mathbb{N}$ and $k \in \mathbb{N}$.
\begin{enumerate}

\item For all $\phi \in \lang_{n}$, $\gtdsnk \vdash \rel \sar w : \phi, w : \negnnf{\phi}, \Gamma$;

\item The $\sub$ and $\wk$ rules are hp-admissible in $\gtdsnk$;

\item All rules in $\gtdsnk$ are hp-invertible;

\item The rule $\cut$ is admissible in $\gtdsnk$;

\item $\gtdsnk$ is sound, i.e.\ if $\vdash_{\gtdsnk} \Lambda$, then $\models \Lambda$;

\item $\gtdsnk$ is complete, i.e.\ $\Vdash_{\dsnk} \phi$, then $\vdash_{\gtdsnk} w : \phi$.

\end{enumerate}

\end{theorem}

\section{Proof-search and Decidability}\label{Sect:Proof-search}

We now put the $\gtdsnk$ calculi to use and demonstrate that each system serves as a basis for automated reasoning with the corresponding logic $\dsnk$. In particular, we design a proof-search procedure, referred to as $\provenk$ (see \alg~\ref{alg:ProveDSk-logical} below), that decides the validity of each formula $\phi \in \lang_{n}$, and which additionally provides witnesses for the answers it yields. 
That is to say, $\provenk$ bottom-up applies rules from $\gtdsnk$ to the input formula, and when a proof is found (i.e.\ proof-search succeeds), it follows that the input formula is valid; when a proof is not found (i.e.\ proof-search fails), we show that a counter-model witnessing the invalidity of the formula can be extracted.

An interesting feature of our proof-search algorithm is the inclusion of a novel loop-checking mechanism, which ensures that proof-search avoids entering an infinite loop. Although our loop-checking mechanism is sophisticated, it is necessary. In fact, at the end of the section, we give a counter-example to terminating proof-search in the absence of loop-checking, illustrating the need for such a mechanism.

The first 
tool we introduce, to be used during proof-search, is the notion of an \emph{$\rel^{i}$-path}. 
Intuitively, an $\rel^{i}$-path exists between two labels in a set $\rel$ of relational atoms when one label can be reached from the other by means of an undirected sequence of relational atoms of the form $R_{\agbox}wu$. Semantically, an $\rel^{i}$-path existing between two labels
$w$ and $u$ means that both are interpreted as worlds within the same choice-cell.

\begin{definition}[$\rel^{i}$-Path]\label{def:Ri-path} Let $w \sim_{i} u \in \{R_{[i]}wu, R_{[i]}uw\}$ and let $\rel$ be a set of relational atoms. An \emph{$\rel^{i}$-path} of relational atoms from a label $w$ to $u$ occurs in $\rel$ (written $w \sim^{\rel}_{i} u$) \ifandonlyif $w = u$, $w \sim_{i} u$, or there exist labels $v_{i}$ (with $i \in \{0,\ldots,n\}$) such that $w \sim_{i} v_{0}, \ldots, v_{n} \sim_{i} u$ occurs in $\rel$.
\end{definition}

By making use of the above definition, it is straightforward to verify the following:

\begin{lemma}\label{lem:ripath-equiv-rel}
If $\Lambda = \rel \sar \Gamma$ is a sequent, then $\ripath$ is an equivalence relation over $\lab(\Lambda)$.
\end{lemma}

The use of $\rel^{i}$-paths is uniquely beneficial in managing 
bottom-up rule applications of $\ioa$ during proof-search. In essence, given a sequent $\Lambda$ and $n = |\Ag|$ many labels $w_{1}, \ldots, w_{n} \in \lab(\Lambda)$, we check to see if a label $u \in \lab(\Lambda)$ exists such that $w_{i} \sim_{i}^{\rel} u$ for every $i \in \{1,\ldots,n\}$, and if not, then $\ioa$ is applied bottom-up to the sequent during proof-search. This avoids unnecessary applications of the $\ioa$ rule during proof-search since if such a label $u \in \lab(\Lambda)$ exists, then $u$ is interpreted as a world in the intersection of the $n = |\Ag|$ many choice-cells that respectively contain $w_{1}, \ldots, w_{n}$, meaning $\ioa$ need not be bottom-up applied relative to these labels as they already satisfy \textbf{(C2)}.

We introduce a second helpful notion in the management of applications of the $\ioa$ rule, namely, the notion of a \emph{thread}, which is defined below.

\begin{definition}[Thread]\label{def:thread} 
 We define a \emph{thread} from $\rel_{0} \sar \Gamma_{0}$ to $\rel_{h} \sar \Gamma_{h}$ to be a sequence $\thrd = (\rel_{i} \sar \Gamma_{i})_{i=0}^{h}$ of sequents such that
\begin{enumerate}

\item for each $i \in \{0, \ldots, h-1\}$, each sequent $\rel_{i+1} \sar \Gamma_{i+1}$ is obtained from $\rel_{i} \sar \Gamma_{i}$ by a bottom-up application of a rule in $\gtdsnk$ with the former sequent serving as a premise and the latter serving as the conclusion, and

\item if a sequent $\rel_{i+1} \sar \Gamma_{i+1}$ is obtained from $\rel_{i} \sar \Gamma_{i}$ by a bottom-up application of $\ioa$, meaning that $\rel_{i+1} = \rel_{i}, R_{[1]}w_{1}u, \ldots, R_{[n]}w_{1}u$, then none of the labels $w_{1}, \ldots, w_{n}$ were introduced via a prior application of $\ioa$.

\end{enumerate}
Occasionally, we say that a thread $\thrd$ is \emph{from} $\rel_{0} \sar \Gamma_{0}$ or \emph{to} $\rel_{h} \sar \Gamma_{h}$ if we only aim to indicate the first or last element of $\thrd$. Last, we call $h$ the \emph{height} of the thread $\thrd$ and $\rel_{h} \sar \Gamma_{h}$ a \emph{top sequent}.
\end{definition}

The first condition in the above definition states that a thread corresponds to a single path in a `partial' $\gtdsnk$ proof, i.e.\ the path may not necessarily end at an initial sequent. The second condition stipulates that bottom-up applications of $\ioa$ are only applied to labels \emph{not} introduced by prior $\ioa$ applications. We shall see that our proof-search procedure always constructs threads with this second property by definition as this ensures the termination of proof-search (see \lem~\ref{lem:termination-label-intro} and \thrm~\ref{thm:termination-prove} below).

Rather than considering single bottom-up applications of the $\ioa$ rule during proof-search, we define an operation, referred to as $\ioaop$, which sequentially applies $\ioa$ to certain collections of labels from a given sequent.

\begin{definition}[$\ioat$-tuple, $\ioat$-satisfied, $\ioaop$]\label{def:ioa-terminology-operation} Let $\Lambda = \rel \sar \Gamma$ be a sequent, $|\Ag| = n$, and $\thrd = (\rel_{i} \sar \Gamma_{i})_{i=0}^{h}$ be a thread from $\rel_{0} \sar \Gamma_{0}$ to $\Lambda$, i.e.\ $\rel = \rel_{h}$ and $\Gamma = \Gamma_{h}$. 

We define an \emph{$\ioat$-label} in $\Lambda$ to be a label introduced by a bottom-up application of $\ioa$ in $\thrd$, and a \emph{non-$\ioat$-label} to be a label in $\Lambda$ that is not an $\ioat$-label. A tuple $(w_{1}, \ldots, w_{n}) \in \lab(\Lambda)$ of length $n$ is defined to be an \emph{$\ioat$-tuple} \iffi for each $i \in \{1,\ldots,n\}$, the label $w_{i}$ is a non-$\ioat$-label.

We say that a tuple $(w_{1}, \ldots, w_{n}) \in \lab(\Lambda)^{n}$ is \emph{$\ioat$-satisfied} \iffi there exists an $\ioat$-label $u \in Lab(\Lambda)$ such that $w_{i} \sim^{\rel}_{i} u$ for all $i \in Ag$. We say that $\Lambda$ is $\ioat$-satisfied \iffi each tuple in $\lab(\Lambda)^{n}$ is $\ioat$-satisfied.

Let $\mathbf{I}$ be the set of all $\ioat$-tuples of $\Lambda$ that are not $\ioat$-satisfied. We define
$$
\mathtt{IoaOp}(\Lambda) := \rel, \bigcup_{\vec{w} \in \mathbf{I}} R_{[i]}w_{i}u \sar \Gamma
$$
with $u$ fresh for each $\ioat$-tuple $\vec{w} = (w_{1}, \ldots, w_{n}) \in \mathbf{I}$.
\end{definition}

There are two reasons for the introduction of the above operation: first, after an application of $\ioaop$ we are ensured that if we transform our labeled sequent into a 
counter-model, 
the model will satisfy condition $\ioacond$ (independence of agents)  which is required in proving our proof-search procedure correct (see \lem~\ref{lem:IOA-sat}). Second, this operation limits bottom-up applications of $\ioa$ to only certain labels, namely, labels that were not introduced by bottom-up applications of $\ioa$. This ensures that bottom-up applications of $\ioa$ do not enter into an infinite loop during proof-search, introducing a fresh label ad infinitum.

We may disregard proofs where the $\stitr$ rule is applied to $\ioat$-labels, and so, $\ioat$-labels will never give rise to new relational atoms of the form $R_{\agbox}uv$ during proof-search, ultimately facilitating termination. This point is discussed in more detail in  \rmk~\ref{rmk:admiss-rules-use}.

\begin{lemma}\label{lem:IOA-sat}
Let $\thrd = (\rel_{i} \sar \Gamma_{i})_{i=0}^{h}$ be a thread where $\Lambda = \rel_{h-1} \sar \Gamma_{h-1}$ and $\Lambda' = \rel_{h} \sar \Gamma_{h} = \ioaop(\Lambda)$. If $(w_{1}, \ldots, w_{n}) \in \lab(\Lambda')^{n}$, then 
$(w_{1}, \ldots, w_{n})$ is $\ioat$-satisfied in $\Lambda'$. 
\end{lemma}

\begin{proof} Let $(w_{1}, \ldots, w_{n}) \in \lab(\Lambda')$. We have two cases: either (i) $(w_{1}, \ldots, w_{n}) \in \lab(\Lambda)^{n}$, or (ii) $(w_{1}, \ldots, w_{n}) \not\in \lab(\Lambda)^{n}$.

(i) Suppose that $(w_{1}, \ldots, w_{n}) \in \lab(\Lambda)^{n}$ and let $\{i_{1}, \ldots, i_{m}\} \subseteq \{1, \ldots, n\}$ such that each $w_{i_{j}}$ with $j \in \{1, \ldots, m\}$ is an $\ioat$-label. Then, since each $\ioat$-label $w_{i_{j}}$ was introduced via a bottom-up application of $\ioa$, there exists a label $u_{i_{j}}$ such that $R_{[i_{j}]}u_{i_{j}}w_{i_{j}} \in \rel_{h-1}$. Let $(v_{1}, \ldots, v_{n})$ be the tuple obtained by replacing each occurrence of $w_{i_{j}}$ in $(w_{1}, \ldots, w_{n})$ with $u_{i_{j}}$. Observe that $(v_{1}, \ldots, v_{n})$ is an $\ioat$-tuple since no $v_{i}$ with $i \in \{1,\ldots, n\}$ is an $\ioat$-label by \dfn~\ref{def:thread} (i.e.\ 
threads only allow bottom-up applications of the $\ioa$ rule to non-$\ioat$-labels). Then, $(v_{1}, \ldots, v_{n})$ is $\ioat$-satisfied in $\Lambda'$ as $\ioaop$ was applied, implying the existence of a label $z$ such that $v_{i} \sim_{i}^{\rel_{h}} z$ for all $i \in \{1, \ldots, n\}$. Since $R_{[i_{j}]}u_{i_{j}}w_{i_{j}} \in \rel_{h-1} \subseteq \rel_{h}$, we also have that $u_{i_{j}} \sim_{i}^{\rel_{h}} w_{i_{j}}$ for all $j \in \{1, \ldots, m\}$, implying that $u_{i_{j}} \sim_{i}^{\rel_{h}} z$ for each $u_{i_{j}}$ by \lem~\ref{lem:ripath-equiv-rel}. It therefore follows that $(w_{1}, \ldots, w_{n})$ is $\ioat$-satisfied in $\Lambda'$.

(ii) Let us assume that $(w_{1}, \ldots, w_{n}) \not\in \lab(\Lambda)^{n}$, meaning that the following set
$$
N_{1} := \{i \in \{1, \ldots, n\} \ | \ w_{i} \in \lab(\Lambda') \setminus \lab(\Lambda)\}
$$
is non-empty. We define $N_{2} := \{1, \ldots, n\} \setminus N_{1}$, and note that $N_{1}$ contains the indicies of all fresh labels $w_{i}$ introduced by $\ioaop$ to $\Lambda'$ and $N_{2}$ contains the indicies of all labels $w_{j}$ occurring in $\Lambda$. Since each $w_{i}$ with $i \in N_{1}$ was introduced by $\ioaop$, we know that there exists a label $u_{i} \in \lab(\Lambda)$ such that $R_{\agbox}u_{i}w_{i} \in \rel_{h}$, i.e.\ $w_{i} \sim^{\rel_{h}}_{i} u_{i}$. If we consider the tuple $(v_{1}, \ldots, v_{n})$ obtained by replacing each occurrence of $w_{i}$ in $(w_{1}, \ldots, w_{n})$ with each such $u_{i}$, then $(v_{1}, \ldots, v_{n}) \in \lab(\Lambda)$, implying that $(v_{1}, \ldots, v_{n})$ is $\ioat$-satisfied by case (i) above. Hence, there exists a label $z$ such that $v_{i} \sim^{\rel}_{i} z$ for each $i \in \{1, \ldots, n\}$. For any label $w_{i}$ in $(w_{1}, \ldots, w_{n})$ with $i \in N_{1}$, we have that $w_{i} \sim^{\rel_{h}}_{i} u_{i}$ and $u_{i} \sim^{\rel_{h}}_{i} z$, meaning that $w_{i} \sim^{\rel}_{i} z$ for each $i \in N_{1}$ by \lem~\ref{lem:ripath-equiv-rel}. Since $w_{j} \sim^{\rel_{h}}_{j} z$ holds for each $j \in N_{2}$ as well, we have that $(w_{1}, \ldots, w_{n})$ is indeed $\ioat$-satisfied in $\Lambda'$.
\end{proof}

We now introduce and prove the admissibility of a selection of inference rules (shown in \fig~\ref{fig:admiss-rules-for-proof-search}). The $\settlc$, $\Oilc$, and $\stitlc$ rules are particularly helpful in controlling how $\Box$, $\Oi$, and $\agbox$ formulae are processed during proof-search. The $\settlc$ and $\Oilc$ rules allow for the label of the principal formula to be changed, and the $\stitlc$ rule allows for the principal formula $u : \agbox \phi$ to be `unpacked' at a label $w$ so long as a relational atom of the form $R_{\agbox}wu$ is present in the sequent. The $\stitdialc$ rule plays a crucial role in ensuring the correctness of our proof-search algorithm (see \thrm~\ref{thm:correctness-prove}, the $\stitdiam \xi$ case, case (ii\textsuperscript{*}) of the inductive step on p.~\pageref{pref:NB-on-admissible-stitdia-rule}). Last, the $\stitsym$ rule is helpful in showing the admissibility of $\stittrans$, which is used in establishing the admissibility of $\stitdialc$.

Such rules will be used to change the functionality of proof-search so that labels are only introduced via the $\settlc$, $\Oilc$, $\stitr$, and $\stitlc$ rules in a ``treelike manner'' as witnessed by \lem~\ref{lem:gen-tree-is-tree} below and which gives rise to a loop-checking mechanism (see \dfn~\ref{def:blocking}) employed in proof-search to ensure termination. We come back to this later, in \rmk~\ref{rmk:admiss-rules-use}.

\begin{figure}[t]
\noindent\hrule

\begin{center}
\begin{tabular}{c c}
\AxiomC{$\rel, \R_{[i]}wu, \R_{[i]}uw \sar \Gamma$}
\RightLabel{$\stitsym$}
\UnaryInfC{$\rel, \R_{[i]}wu \sar \Gamma$}
\DisplayProof

&

\AxiomC{$\rel, \R_{[i]}wu, \R_{[i]}uv, \R_{[i]}wv  \sar \Gamma$}
\RightLabel{$\stittrans$}
\UnaryInfC{$\rel, \R_{[i]}wu, \R_{[i]}uv \sar \Gamma$}
\DisplayProof
\end{tabular}
\end{center}

\begin{center}
\begin{tabular}{c c}
\AxiomC{$\rel, \R_{[i]}wu, \R_{[i]}wv \sar u : [i] \phi, v : \phi, \Gamma$}
\RightLabel{$\stitlc^{\dag}$}
\UnaryInfC{$\rel, \R_{[i]}wu \sar u : [i] \phi, \Gamma$}
\DisplayProof

&

\AxiomC{$\rel \sar w : \Box \phi, \Gamma$}
\RightLabel{$\settlc$}
\UnaryInfC{$\rel \sar u : \Box \phi, \Gamma$}
\DisplayProof
\end{tabular}
\end{center}

\begin{center}
\begin{tabular}{c c}
\AxiomC{$\rel, R_{\agbox}wu \sar w : \agdia \phi, u : \agdia \phi, u : \phi, \Gamma$}
\RightLabel{$\stitdialc$}
\UnaryInfC{$\rel, R_{\agbox}wu \sar w : \agdia \phi, \Gamma$}
\DisplayProof

&

\AxiomC{$\rel \sar w : \Oi \phi, \Gamma$}
\RightLabel{$\Oilc$}
\UnaryInfC{$\rel \sar u : \Oi \phi, \Gamma$}
\DisplayProof
\end{tabular}
\end{center}

\hrule
\caption{Admissible rules. The side condition $\dag$ stipulates that $v$ must be an eigenvariable.}
\label{fig:admiss-rules-for-proof-search}
\end{figure}

\begin{lemma}\label{lem:admissible-rules} All rules in \fig~\ref{fig:admiss-rules-for-proof-search} are admissible in $\gtdsnk$.
\end{lemma}

\begin{proof} We first show the admissibility of the rules $\stitsym$, $\stittrans$, $\settlc$, and $\Oilc$. The first two rules
leverage the hp-admissibility of $\wk$ and the latter two rules
rely on the hp-invertibility of $\settr$ and $\Oir$.

\begin{center}
\scalebox{0.93}{
\begin{tabular}{c c}
\AxiomC{$\rel, \R_{[i]}wu, \R_{[i]}uw \sar \Gamma$}
\RightLabel{$\wk$}
\dashedLine
\UnaryInfC{$\rel, \R_{[i]}ww, \R_{[i]}wu, \R_{[i]}uw \sar \Gamma$}
\RightLabel{$\stiteucl$}
\UnaryInfC{$\rel, \R_{[i]}ww, \R_{[i]}wu \sar \Gamma$}
\RightLabel{$\stitrefl$}
\UnaryInfC{$\rel, \R_{[i]}wu \sar \Gamma$}
\DisplayProof

&

\AxiomC{$\rel \sar w : \Box \phi, \Gamma$}
\RightLabel{\thrm~\ref{thm:gtdsnk-properties}-3}
\dashedLine
\UnaryInfC{$\rel \sar v : \phi, \Gamma$}
\RightLabel{$\settr$}
\UnaryInfC{$\rel \sar u : \Box \phi, \Gamma$}
\DisplayProof
\end{tabular}
}
\end{center}

\begin{center}
\scalebox{0.93}{
\resizebox{\columnwidth}{!}{
\begin{tabular}{c c}
\AxiomC{$\rel, \R_{[i]}wu, \R_{[i]}uv, \R_{[i]}wv  \sar \Gamma$}
\RightLabel{$\wk$}
\dashedLine
\UnaryInfC{$\rel, \R_{[i]}wu, \R_{[i]}uw, \R_{[i]}uv, \R_{[i]}wv  \sar \Gamma$}
\RightLabel{$\stiteucl$}
\UnaryInfC{$\rel, \R_{[i]}wu, \R_{[i]}uw, \R_{[i]}uv \sar \Gamma$}
\RightLabel{$\stitsym$}
\dashedLine
\UnaryInfC{$\rel, \R_{[i]}wu, \R_{[i]}uv \sar \Gamma$}
\DisplayProof

&

\AxiomC{$\rel \sar w : \Oi \phi, \Gamma$}
\RightLabel{\thrm~\ref{thm:gtdsnk-properties}-3}
\dashedLine
\UnaryInfC{$\rel, \ideal v \sar v : \phi, \Gamma$}
\RightLabel{$\Oir$}
\UnaryInfC{$\rel \sar u : \Oi \phi, \Gamma$}
\DisplayProof
\end{tabular}
}
}
\end{center}
To show the admissibility of $\stitlc$, we invoke the hp-admissibility of $\wk$ as well as $\cut$ admissibility (\thrm~\ref{thm:gtdsnk-properties}).

\begin{center}
\scalebox{0.93}{
\begin{tabular}{c c}
$\Pi_{1} = \Bigg \{$

&

\AxiomC{}
\RightLabel{\thrm~\ref{thm:gtdsnk-properties}}
\dashedLine
\UnaryInfC{$\rel, \R_{[i]}wu, \R_{[i]}ww \sar u : \agdia [i] \phi, u : \agdia \negnnf{\phi}, u : \agbox \phi, w : \agbox \phi, \Gamma$}
\RightLabel{$\stitdiar$}
\UnaryInfC{$\rel, \R_{[i]}wu, \R_{[i]}ww \sar u : \agdia [i] \phi, u : \agdia \negnnf{\phi}, w : [i] \phi, \Gamma$}
\RightLabel{$\stitrefl$}
\UnaryInfC{$\rel, \R_{[i]}wu \sar u : \agdia [i] \phi, u : \agdia \negnnf{\phi}, w : [i] \phi, \Gamma$}
\DisplayProof
\end{tabular}
}
\end{center}

\begin{center}
\scalebox{0.93}{
\begin{tabular}{c c}
$\Pi_{2} = \Bigg \{$

&

\AxiomC{$\rel, \R_{[i]}wu, \R_{[i]}wv \sar u : [i] \phi, v : \phi, \Gamma$}
\RightLabel{$\stitr$}
\UnaryInfC{$\rel, \R_{[i]}wu \sar w : [i] \phi, u : [i] \phi, \Gamma$}
\RightLabel{$\wk$}
\dashedLine
\UnaryInfC{$\rel, \R_{[i]}wu \sar w : [i] \phi, u : [i] \phi, u : \agdia \agbox \phi, \Gamma$}
\AxiomC{$\Pi_{1}$}
\dashedLine
\RightLabel{$\cut$}
\dashedLine
\BinaryInfC{$\rel, \R_{[i]}wu \sar w : \agbox \phi, u : \agdia \agbox \phi, \Gamma$}
\RightLabel{$\stitdiar$}
\UnaryInfC{$\rel, \R_{[i]}wu \sar u : \agdia \agbox \phi, \Gamma$}
\RightLabel{$\wk$}
\dashedLine
\UnaryInfC{$\rel, \R_{[i]}wu \sar u : \agdia \agbox \phi, u : [i] \phi, \Gamma$}
\DisplayProof
\end{tabular}
}
\end{center}
\smallskip

\begin{center}
\scalebox{0.93}{
\begin{tabular}{c c}
$\Pi_{3} = \Bigg \{$

&

\AxiomC{}
\RightLabel{\thrm~\ref{thm:gtdsnk-properties}}
\dashedLine
\UnaryInfC{$\rel, \R_{[i]}wu, \R_{[i]}uz, \R_{[i]}uv, \R_{[i]}zv  \sar z : \agdia \negnnf{\phi}, v : \phi, v : \negnnf{\phi}, \Gamma$}
\RightLabel{$\stitdiar$}
\UnaryInfC{$\rel, \R_{[i]}wu, \R_{[i]}uz, \R_{[i]}uv, \R_{[i]}zv   \sar z : \agdia \negnnf{\phi}, v : \phi, \Gamma$}
\RightLabel{$\stiteucl$}
\UnaryInfC{$\rel, \R_{[i]}wu, \R_{[i]}uz, \R_{[i]}uv \sar z : \agdia \negnnf{\phi}, v : \phi, \Gamma$}
\RightLabel{$(\agbox) \cdot 2$}
\UnaryInfC{$\rel, \R_{[i]}wu \sar u : \agbox \agdia \negnnf{\phi}, u : \agbox \phi, \Gamma$}
\DisplayProof
\end{tabular}
}
%}
\end{center}

\smallskip
\begin{center}
\scalebox{0.93}{
\AxiomC{$\Pi_{3}$}
\AxiomC{$\Pi_{2}$}
\RightLabel{$\cut$} 
\dashedLine
\BinaryInfC{$\rel, \R_{[i]}wu \sar u : \agbox \phi, \Gamma$}
\DisplayProof
}
\end{center}
The admissibility of $\stitdialc$ follows from the hp-admissibility of $\wk$ and cut admissibility. 

\smallskip
\begin{center}
\scalebox{0.93}{
\begin{tabular}{c c}
$\Pi_{1} = \Bigg \{$

&

\AxiomC{$\rel, R_{\agbox}wu \sar w : \agdia \phi, u : \agdia \phi, u : \phi, \Gamma$}
\RightLabel{$\stitdiar$}
\UnaryInfC{$\rel, R_{\agbox}wu \sar w : \agdia \phi, u : \agdia \phi, \Gamma$}
\RightLabel{$\wk$}
\dashedLine
\UnaryInfC{$\rel, R_{\agbox}wu \sar w : \agdia \phi, w : \agdia \agdia \phi, u : \agdia \phi, \Gamma$}
\RightLabel{$\stitdiar$}
\UnaryInfC{$\rel, R_{\agbox}wu \sar w : \agdia \phi, w : \agdia \agdia \phi, \Gamma$}
\DisplayProof
\end{tabular}
}
\end{center}

\smallskip

\begin{center}
\scalebox{0.93}{
\begin{tabular}{c c}
$\Pi_{2} = \Bigg \{$

&

\AxiomC{}
\RightLabel{\thrm~\ref{thm:gtdsnk-properties}}
\dashedLine
\UnaryInfC{$\rel, R_{\agbox}wu, R_{\agbox}wv', R_{[i]}wu', R_{[i]}u'v'  \sar v' : \negnnf{\phi}, v' : \phi, w : \agdia \phi, \Gamma$}
\RightLabel{$\stitdiar$}
\UnaryInfC{$\rel, R_{\agbox}wu, R_{\agbox}wv', R_{[i]}wu', R_{[i]}u'v'  \sar v' : \negnnf{\phi}, w : \agdia \phi, \Gamma$}
\RightLabel{$\stittrans$}
\UnaryInfC{$\rel, R_{\agbox}wu, R_{[i]}wu', R_{[i]}u'v'  \sar v' : \negnnf{\phi}, w : \agdia \phi, \Gamma$}
\RightLabel{$([i])\cdot 2$}
\UnaryInfC{$\rel, R_{\agbox}wu \sar w : \agbox \agbox \negnnf{\phi}, w : \agdia \phi, \Gamma$}
\DisplayProof
\end{tabular}
}
\end{center}
\smallskip

\begin{center}
\scalebox{0.93}{
\AxiomC{$\Pi_{1}$}
\AxiomC{$\Pi_{2}$}
\RightLabel{$\cut$}
\dashedLine
\BinaryInfC{$\rel \sar w : \agdia \phi, \Gamma$}
\DisplayProof
}\end{center}
\vspace{-0.3cm}
\end{proof}

We now introduce the fundamental component of our loop-checking mechanism, i.e.\ the notion of a \emph{generation tree}. A generation tree is a graph that tracks how certain labels are introduced within a thread, and can be used to check if a label $w$ has an ancestor $u$ associated with the same set of formulae (in which case we say that $w$ is \emph{blocked}; see \dfn~\ref{def:blocking} below). In such a case, the $\stitr$ rule need not be applied to the leaf of the path containing $w$ and $u$ in the 
tree. As discussed in \lem~\ref{lem:termination-label-intro} below, this loop-checking mechanism bounds the depth of the generation tree, ultimately securing terminating proof-search.

\begin{definition}[Generation Tree]\label{def:generation-tree} Let $\thrd = (\rel_{i} \sar \Gamma_{i})_{i=0}^{h}$ be a thread from $\emptyset \sar w : \phi$ to $\Lambda$. We define the \emph{generation tree} $\gt_{\Lambda}^{\thrd} = (V,E)$ as follows:
\begin{itemize}

\item $w \in V$;

\item if $\settr$ is applied bottom-up in $\thrd$ with $w : \Box \phi$ principal and $u : \phi$ auxiliary, then add $u$ to $V$ and $(w,u)$ to $E$;

\item if $\Oir$ is applied bottom-up in $\thrd$ with $w : \Oi \phi$ principal, and $\ideal u$ and $u : \phi$ auxiliary, then add $u$ to $V$ and $(w,u)$ to $E$;

\item if $\dtwo$ is applied bottom-up in $\thrd$ with $\ideal u$ auxiliary, 
add $u$ to $V$ and $(w,u)$ to $E$;

\item if $\stitr$ is applied bottom-up in $\thrd$ with $u : \agbox \phi$ principal and $v : \phi$ auxiliary, then add $v$ to $V$ and $(u,v)$ to $E$;

\item if $\stitlc$ is applied bottom-up in $\thrd$ with $R_{\stitbox}uv$ and $v : \agbox \phi$ principal, and $R_{\stitbox}uz$ and $z : \phi$ auxiliary, then add $z$ to $V$ and $(u,z)$ to $E$.\footnote{We note that 
$u$ must already occur in $V$, so we are indeed permitted to add $(u,z)$ to $E$ (see \lem~\ref{lem:gen-tree-is-tree}).}

\end{itemize}
Note that in the $\settr$ and $\Oir$ cases, the edge $(w,u)$ added to $E$ always ensures that $u$ is a child of the root $w$. The significance of this is explained later on in \rmk~\ref{rmk:admiss-rules-use}.
\end{definition}

\begin{lemma}\label{lem:gen-tree-is-tree}
Let $\thrd = (\rel_{i} \sar \Gamma_{i})_{i=0}^{h}$ be a thread from $\emptyset \sar w : \phi$. The generation tree $\gt^{\thrd}_{\Lambda} = (V,E)$ is a tree.
\end{lemma}

\begin{proof} We prove the result by induction on the height of the thread $\thrd$.

\textit{Base case.} If $h = 0$, then $\thrd = \emptyset \sar w : \phi$, implying that $\gt^{\thrd}_{\Lambda} = (V,E)$ with $V = \{w\}$ and $E = \emptyset$, which is a tree.

\textit{Inductive step.} Let $\thrd = (\rel_{i} \sar \Gamma_{i})_{i=0}^{h+1}$ be a thread to $\Lambda$ with $\thrd' = (\rel_{i} \sar \Gamma_{i})_{i=0}^{h}$ a thread to $\Lambda'$. By IH, we know that $\gt^{\thrd'}_{\Lambda'} = (V',E')$ is a tree. If $\Lambda'$ was obtained from $\Lambda$ by a rule other than $\settr$, $\Oir$, $\dtwo$, $\stitr$, or $\stitlc$ then $\gt^{\thrd}_{\Lambda} = \gt^{\thrd'}_{\Lambda'}$ by \dfn~\ref{def:generation-tree}. Alternatively, if $\Lambda'$ was obtained from $\Lambda$ by $\settr$, $\Oir$, $\dtwo$, or $\stitr$, then due to the eigenvariable condition imposed on each rule, $\gt^{\thrd}_{\Lambda} = (V,E)$ where $V = V' \cup \{v\}$ and $E = E' \cup \{(u,v)\}$ with $u \in V'$ and $v$ fresh. Last, concerning the $\stitlc$ case, recall that the $\ioa$ rule is never applied to an $\ioat$-label in a thread $\thrd$ by \dfn~\ref{def:thread} and~\ref{def:ioa-terminology-operation}. As such, if $\stitlc$ is bottom-up applied to the sequent $\rel, R_{\agbox}uv \sar v : \agbox \phi, \Gamma$ to obtain the sequent $\rel, R_{\agbox}uv, R_{\agbox}uz \sar v : \agbox \phi, z : \phi, \Gamma$, then $u$ cannot be an $\ioat$-label, meaning, either $u = w$ or the label was introduced by a prior bottom-up application of $\settr$, $\Oir$, $\dtwo$, $\stitr$, or $\stitlc$. Therefore, $u \in V'$ and $\gt^{\thrd}_{\Lambda} = (V,E)$ where $V = V' \cup \{z\}$ and $E = E' \cup \{(u,z)\}$ with $z$ fresh. Hence, in the $\settr$, $\Oir$, $\dtwo$, $\stitr$, and $\stitlc$ cases, $\gt^{\thrd}_{\Lambda}$ is nothing more than $\gt^{\thrd'}_{\Lambda'}$ with a new edge protruding from a vertex to a fresh node, showing that $\gt^{\thrd}_{\Lambda}$ is indeed a tree.
\end{proof}

We adopt a notion of blocking in our setting, used to bound the number of labels introduced during proof-search. It is based on 
the work of \citeauthor{HorSat04} \citeyear{HorSat04}, who introduced loop-checking by means of (in)direct blocking in tableaux for description logics, and \citeauthor{TiuIanGor12} \citeyear{TiuIanGor12}, who applied a loop-checking mechanism for grammar logics in the context of nested sequent calculi. Both methods work in fundamentally the same way: in the former, a tableau is constructed, which is a tree whose nodes are sets of description logic formulae, and in the latter, a proof of nested sequents is constructed, where each nested sequent is a tree whose nodes are sets of grammar logic formulae. In both cases, the trees `grow' throughout proof-search and when two nodes are encountered along the same branch with the same set of formulae, one stops growing the branch, thus limiting the size of the tableau or nested sequent. 
 In our setting, although our labeled sequents do not form a tree, we may leverage generation trees to employ a similar 
 mechanism.

\begin{definition}[Blocking]\label{def:blocking}
 Let $\Lambda = \rel \sar \Gamma$ be a sequent and $\thrd = (\rel_{i} \sar \Gamma_{i})_{i=0}^{h}$ be a thread from $\emptyset \sar w : \phi$ to $\Lambda$. We say that $u$ is \emph{directly blocked} by its proper ancestor $v$ in the generation tree $\gt^{\thrd}_{\Lambda}$ \iffi (i) $v$ is not the root of $\gt^{\thrd}_{\Lambda}$, (ii) $\opt_{\Oi}u \in \rel$ \iffi $\opt_{\Oi}v \in \rel$, for all $i \in Ag$, and (iii) $\Gamma \restriction u = \Gamma \restriction v$. We refer to such a $u$ that is directly blocked as a \emph{loop node}, and 
 to such a $v$ as a \emph{loop ancestor}. We say that $u$ is \emph{indirectly blocked} \iffi it has a proper ancestor in the generation tree $\gt^{\thrd}_{\Lambda}$ that is directly blocked. Last, we say that $u$ is \emph{blocked} in $\Lambda$ \iffi it is directly or indirectly blocked in the generation tree $\gt^{\thrd}_{\Lambda}$, and 
 that $u$ is \emph{unblocked} otherwise. 
\end{definition}

\begin{remark}\label{rmk:ioat-label-unblocked} By the definition of a generation tree (\dfn~\ref{def:generation-tree}), no label introduced via a bottom-up application of $\ioa$ (i.e.\ an $\ioat$-label) will occur as a node in the generation tree. Hence, all $\ioat$-labels will be considered \emph{unblocked} by \dfn~\ref{def:blocking} above.
\end{remark}

Next, we introduce \emph{saturation conditions}, used to determine when the remaining rules of $\gtdsnk$ are no longer applicable during proof-search and `enough' information has been introduced to ensure that a counter-model can be extracted in the case of failure. 
Our conditions are motivated by conditions enforced in proof-search algorithms for 
description logics~\cite{HorSat04} and  multi-modal grammar logics~\cite{TiuIanGor12}.

\begin{definition}[Saturation Conditions] Let $\thrd$ be a thread to $\Lambda$. We define the \emph{saturation conditions} accordingly:
\begin{description}

\item[$\scid$]  for each $w \in \lab(\Lambda)$, if $w : \phi \in \Gamma$, then $w : \negnnf{\phi} \not\in \Gamma$;

\item[$\scdis$] for each $w \in \lab(\Lambda)$, if $w : \phi \lor\psi \in \Gamma$, then $w :\phi, w : \psi \in \Gamma$;

\item[$\sccon$] for each $w \in \lab(\Lambda)$, if $w : \phi \land \psi \in \Gamma$, then either $w :\phi$ or $w : \psi \in \Gamma$;

\item[$\scdia$] for each $w, u \in \lab(\Lambda)$, if $w : \Diamond \phi \in \Gamma$, then $u : \phi \in \Gamma$;

\item[$\scbox$] for each $w \in \lab(\Lambda)$, if $w : \Box \phi \in \Gamma$, then there exists an unblocked $u$ in the generation tree $\gt^{\thrd}_{\Lambda}$ such that $u : \phi \in \Gamma$;

\item[$\scagdia$] for each $w, u \in \lab(\Lambda)$ and $i \in \Ag$, if $w : \agdia \phi \in \Gamma$ and $\reli wu \in \rel$, then $u : \phi, u : \agdia \phi \in \Gamma$;

\item[$\scagbox$] for each $w \in \lab(\Lambda)$ and $i \in \Ag$, if $w : \agbox \phi \in \Gamma$ with $w$ unblocked in the generation tree $\gt^{\thrd}_{\Lambda}$, then there exists a $u \in \lab(\Lambda)$ such that $\reli wu \in \rel$ and $u : \phi \in \Gamma$;

\item[$\scref$] for each $w \in \lab(\Lambda)$, $\reli ww \in \rel$;

\item[$\sceuc$] for each $w, u, v \in \lab(\Lambda)$, if $\reli wu, \reli wv \in \rel$, then $\reli uv \in \rel$;

\item[$\scodi$] for each $w, u \in \lab(\Lambda)$ and $i \in \Ag$, if $w : \ODi \phi \in \Gamma$ and $\relio u \in \rel$, then $u : \phi \in \Gamma$;

\item[$\scoi$] for each $w \in \lab(\Lambda)$ and $i \in \Ag$, if $w : \ODi \phi \in \Gamma$, then there exists an unblocked $u$ in the generation tree $\gt^{\thrd}_{\Lambda}$ such that $\relio u \in \rel$ and $u : \phi \in \Gamma$;

\item[$\scdii$] for each $i \in \Ag$ there exists an unblocked $w$ in the generation tree $\gt^{\thrd}_{\Lambda}$ such that $\relio w \in \rel$;

\item[$\scdiii$] for each $w, u \in \lab(\Lambda)$ and $i \in \Ag$, if $\relio w, \reli wu \in \rel$, then $\relio u \in \rel$;

\item[$\scapc$] for each $w_{0}, \ldots, w_{k} \in \lab(\Lambda)$ and $i \! \in \! \Ag$, $R_{\agbox}w_{j}w_{m} \in \rel$ for some $j,m \in \{0, \ldots, k\}$ with $j \neq m$.

\end{description}
We refer to $\Lambda$ as \emph{stable} \iffi it satisfies all above conditions and every $\ioat$-tuple of the sequent is $\ioat$-satisfied.
\end{definition}

Our proof-search algorithm $\provenk$ is shown in \alg~\ref{alg:ProveDSk-logical} below, and decides if a formula is $\dsnk$-valid, meaning we take $n$ and $k$ to be fixed, input parameters. $\provenk$ is split into three parts due to its length. We note that the algorithm applies rules from $\gtdsnk$ and admissible rules from \fig~\ref{fig:admiss-rules-for-proof-search} bottom-up on an input of the form $\emptyset \sar w_{0} : \phi$ with $\phi \in \lang_{n}$. Each recursive call of the algorithm corresponds to a bottom-up application of a rule, and at any given step, $\provenk$ has generated a finite number of threads from $\emptyset \sar w_{0} : \phi$ that form a `partial proof' of the input. As we will see below, if the input is valid, then eventually a proof in $\gtdsnk$ may be extracted from successful proof-search, and if the input is invalid, then a counter-model may be extracted from failed proof-search. We also note that if $k = 0$, i.e.\ we do not bound the number of choices available to an agent, then we omit lines 15--20 and do not enforce the saturation condition $\scapc$. 

The  saturation conditions, along with $\ioat$-satisfiability, are used to guide the computation of $\provenk$. Such conditions serve a twofold role: first, they are used to determine the (in)applicability of a rule during proof-search, thus letting us control when rules are bottom-up applied, and yielding termination of the algorithm (\thrm~\ref{thm:termination-prove}). Second, such conditions ensure that enough information has entered into a sequent during proof-search to be able to extract a counter-model in the case where the input formula is invalid (\thrm~\ref{thm:correctness-prove}). In particular, if a stable sequent is encountered during proof-search, then it can be transformed into a finite $\dsnk$-model (see \dfn~\ref{def:stability-model} and \lem~\ref{lem:stability-model-is-DSn} below). 
We remark that it is decidable to check if a sequent (which is a finite object) is stable.

The proof-search strategy of our algorithm $\provenk$ is, intuitively, as follows: First, $\provenk$ is split into three parts, namely, lines 5-20, lines 21-40, and lines 41-60. Lines 5-20 bottom-up apply the inference rules $\stitrefl$, $\stiteucl$, $(\mathsf{D3}_{i})$, and $\choicer$, which only affect relational atoms. Lines 21-40 bottom-up apply the inference rules $\disr$, $\conr$, $\settdiar$, $\stitdialc$, and $\ODir$, which only affect labeled formulae. All rules in  lines 5-20 and 21-40 are \emph{non-generating} in the sense that they do not introduce new labels. Conversely, lines 41-60 apply the \emph{generating} rules $\settr$, $\Oir$, $\dtwo$, $\stitr$, $\stitlc$, and $\ioaop$ (i.e.\ a sequence of $\ioa$ applications) that bottom-up introduce new, fresh labels to sequents. We note that $\settr$ and $\Oir$ are preceded by applications of $\settlc$ and $\Oilc$, respectively, as discussed in more detail in the remark below. Since all rules of $\gtdsnk$ are hp-invertible (\thrm~\ref{thm:gtdsnk-properties}-3), we can bottom-up apply our rules in any order. However, we find it conceptually easier to prove facts about $\provenk$ by 
grouping rules of a similar functionality together.

\begin{remark}\label{rmk:admiss-rules-use} As it is important for understanding $\provenk$, we discuss our use of the admissible rules $\settlc$, $\Oilc$, and $\stitlc$. Lines 41-43 bottom-up apply $\settlc$ followed by $\settr$ and lines 44-47 bottom-up apply $\Oilc$ followed by $\Oir$. In each case, the label $w$ of $w : \Box \psi$ and $w : \Oi \psi$ is changed (via the $\settlc$ and $\Oilc$ rules) to the `source' label $w_{0}$ before bottom-up applying $\settr$ and $\Oir$, respectively. (NB. We always assume that the input to $\provenk$ is a sequent of the form $\emptyset \sar w_{0} : \phi$, meaning $w_{0}$ will be the root of every generation tree constructed throughout proof-search.) This has the effect that each label introduced by the bottom-up application of the subsequent $\settr$ or $\Oir$ is a child of the root of the generation tree. Observe that by the definition of a blocked node (\dfn~\ref{def:blocking}), that a node in the generation tree cannot be blocked by the root of the generation tree. This means that fresh labels generated by $\settr$ and $\Oir$ are unblocked by definition, which is crucial for constructing a counter-model when proof-search fails (see \dfn~\ref{def:stability-model} and \lem~\ref{lem:stability-model-is-DSn} below) as our counter-model is built using only unblocked labels. Hence, labels generated by bottom-up applications of $\settr$ and $\Oir$ will always exist within our counter-model, which will be vital in proving the correctness of $\provenk$ (see \thrm~\ref{thm:correctness-prove} below). 

Lines 52-57 bottom-up apply $\stitr$ to unblocked non-$\ioat$-labels and $\stitlc$ to $\ioat$-labels. Applying $\stitr$ to only unblocked non-$\ioat$-labels limits the growth of sequents during proof-search since repeated bottom-up applications of $\stitr$ will eventually produce blocked non-$\ioat$-labels as discussed in the proof of termination (see \thrm~\ref{thm:termination-prove} below). Our reason for the use of $\stitlc$ however, is that since if a sequent $\rel, R_{\stitbox}wu \sar \Gamma$ is encountered during proof-search with $u : \stitbox \phi \in \Gamma$ and $u$ an $\ioat$-label, then applying $\stitlc$ lets us `unpack' $u : \stitbox \phi$ and obtain $\rel, R_{\stitbox}wu, R_{\stitbox}wv \sar \Gamma, v : \phi$ with $v$ fresh rather than $\rel, R_{\stitbox}wu, R_{\stitbox}uv \sar \Gamma, v : \phi$. This prevents chains of relational atoms from being introduced that protrude outward from $\ioat$-labels. In fact, none of the generating rules are ever applied to an $\ioat$-label, meaning such labels can be ignored in our generation tree, and thus do not give rise to infinite loops during proof-search, as they never give rise to fresh labels during proof-search.
\end{remark}

We now explain how a counter-model (which we call a \emph{stability model}) can be obtained from failed proof-search, and then show the correctness of the proof-search algorithm.

\begin{algorithm}[t]
\KwIn{A Labeled Sequent: $\Lambda = \rel \sar \Gamma$}
\KwOut{A Boolean: \texttt{True}, \texttt{False}}

\uIf{$w : \phi, w :\negnnf{\phi} \in \Gamma$}
     {\Return \texttt{True};}

\uIf{$\rel \sar \Gamma$ is stable}
     {\Return \texttt{False};}
     
\uIf{for some $w \in \lab(\Lambda)$, $\reli ww \not\in \rel$}
    {
    Let $\rel' := \reli ww, \rel$;\\
    \Return $\provenk$($\rel' \sar \Gamma$)
    }
    
\uIf{for some $w, u, v \in \lab(\Lambda)$, $\reli wu, \reli wv \in \rel$, but $\reli uv \not\in \rel$}
    {
    Let $\rel' := \reli uv, \rel$;\\
    \Return $\provenk$($\rel' \sar \Gamma$);
    }
    
\uIf{for some $w, u \in \lab(\Lambda)$, $\relio w, \reli wu \in \rel$, but $\relio u \not\in \rel$}
    {
    Let $\rel' := \relio u, \rel$;\\
    \Return $\provenk$($\rel' \sar \Gamma$);
    }

\uIf{for some $w_{0}, \ldots, w_{k} \in \lab(\Lambda)$, $R_{\agbox}w_{j}w_{m} \not\in \rel$ for each $j,m \in \{0, \ldots, k\}$ with $j \neq m$}
    {
    Let $\rel_{m,j} := R_{\agbox}w_{m}w_{j}, \rel$ (with $0 \leq m \leq k - 1$ and $m + 1 \leq j \leq k$);\\
    \uIf{$\provenk (\rel_{m,j} \sar \Gamma) = \mathtt{False}$ for some $m$ and $j$}
        {
        \Return \texttt{False};
        }\Else{
        \Return \texttt{True};
        }
    }
     
\caption{$\provenk$}\label{alg:ProveDSk-logical}
\end{algorithm}

\setcounter{algocf}{0}
\begin{algorithm}
%\hrulefill\\
\setcounter{AlgoLine}{20}

\uIf{there exists a $w : \phi \lor \psi \in \Gamma$, but either $w : \phi \not\in \Gamma$ or $w : \psi \not\in \Gamma$}
     {Let $\Gamma' := w :\phi, w :\psi, \Gamma$;\\
     \Return $\provenk$($\rel \sar \Gamma'$);}
     
\uIf{there exists a $w : \phi \land \psi \in \Gamma$, but $w : \phi, w : \psi \not\in \Gamma$}
{
    Let $\Gamma_{1} := w :\phi, \Gamma$;\\
    Let $\Gamma_{2} := w :\psi, \Gamma$;\\
    \uIf{$\provenk (\rel \sar \Gamma_{i}) = \mathtt{False}$ for some $i \in \{1,2\}$}
    {
    \Return \texttt{False};
    }\Else{
    \Return \texttt{True};
    }
}
    
\uIf{$w : \Diamond \phi \in \Gamma$, but $u : \phi \not\in \Gamma$ for some $u \in \lab(\Lambda)$}
    {
    Let $\Gamma' := u :\phi, \Gamma$;\\
     \Return $\provenk$($\rel \sar \Gamma'$);
    }
    
\uIf{$w : \agdia \phi \in \Gamma$, but $u : \phi, u : \agdia \phi \not\in \Gamma$ for some $u \in \lab(\Lambda)$ such that $R_{\agbox}wu \in \rel$}
    {
    Let $\Gamma' := u :\phi, u : \agdia \phi, \Gamma$;\\
     \Return $\provenk$($\rel \sar \Gamma'$);
    }
    
\uIf{$w : \ODi \phi \in \Gamma$, but $u : \phi \not\in \Gamma$ for some $u \in \lab(\Lambda)$ such that $\relio u \in \rel$}
    {
    Let $\Gamma' := u :\phi, \Gamma$;\\
     \Return $\provenk$($\rel \sar \Gamma'$);
    }
\caption{$\provenk$ (Continued)}
\end{algorithm}

\setcounter{algocf}{0}
\begin{algorithm}\label{alg:with-loop-check}
\setcounter{AlgoLine}{40}

\uIf{$w : \Box \phi \in \Gamma$, but $u : \phi \not\in \Gamma$ for all $u \in \lab(\Lambda)$}
    {
     Replace $w : \Box \phi \in \Gamma$ with $w_{0} : \Box \phi$ and let $\Gamma' := v :\phi, \Gamma$ with $v$ fresh;\\
     \Return $\provenk$($\rel \sar \Gamma'$);
    }
    
\uIf{$w : \Oi \phi \in \Gamma$, but $u : \phi \not\in \Gamma$ for all $u \in \lab(\Lambda)$ with $\ideal u \in \rel$}
    {
     Replace $w : \Oi \phi \in \Gamma$ with $w_{0} : \Oi \phi$;\\
     Let $\rel' := \relio v, \rel$ and $\Gamma' := v :\phi, \Gamma$ with $v$ fresh;\\
     \Return $\provenk$($\rel' \sar \Gamma'$);
    }

\uIf{for all $w \in \lab(\Lambda)$, $\relio w \not\in \rel$}
    {
    Let $\rel' := \relio u, \rel$ with $u$ fresh;\\
    \Return $\provenk$($\rel' \sar \Gamma$);
    }
    
\uIf{$w : \agbox \phi \in \Gamma$ with $w$ unblocked in $\gt^{\thrd}_{\Lambda}$, but $u : \phi \not\in \Gamma$ for all $u \in \lab(\Lambda)$ such that $R_{\agbox}wu \in \rel$}
    {
    \uIf{$w$ is not an $\ioat$-label}
        {
        Let $\rel' := R_{\agbox}wv, \rel$ and $\Gamma' := v :\phi, \Gamma$ with $v$ fresh;\\
    \Return $\provenk$($\rel' \sar \Gamma'$);
        }
    \uIf{$w$ is an $\ioat$-label}
        {
        \tcp{As $w$ is an $\ioat$-label, $\rel$ is of the form $R_{\agbox}vw, \rel'$.}
        Let $\rel'' := R_{\agbox}vz, R_{\agbox}vw, \rel'$ and $\Gamma' := z :\phi, \Gamma$ with $z$ fresh;\\
    \Return $\provenk$($\rel'' \sar \Gamma'$);
        }
    }
    
\uIf{for some non-$\ioat$-labels $w_{1}, \ldots, w_{n} \in \lab(\Lambda)$, there does not exist an $\ioat$-label $u \in \lab(\Lambda)$ such that $w_{i} \ripath u$ for each $i \in \Ag$}
    {
    Let $\Lambda' := \ioaop(\Lambda)$;\\
    \Return $\provenk$($\Lambda'$);
    }
\caption{$\provenk$ (Continued)}
\end{algorithm}

\begin{definition}[Stability Model]\label{def:stability-model} Let $\thrd$ be a thread from $\emptyset \sar w_{0} : \phi$ to $\Lambda := \rel \sar \Gamma$ generated by $\provenk(\emptyset \sar w_{0} : \phi)$. We define the \emph{stability model} relative to $\Lambda$ to be $\mstamod :=$ $(\wstamod,\{\rstamod \ | \ i \in \Ag \}, \{\istamod \ | \ i \in \Ag\}, \vstamod)$ such that
\begin{itemize}

\item $\wstamod$ is the set of unblocked labels in $\Lambda$;

\item $\rstamod$ is the Euclidean closure of the set $S^{\Lambda}_{[i]}$, where $(w,u) \in S^{\Lambda}_{[i]}$ \iffi (i) $\reli wu \in \rel$ with $w$ and $u$ both unblocked in $G_{\Lambda}^{\thrd}$, or (ii) $\reli wv \in \rel$ with $w$ unblocked, and $v$ directly blocked by $u$ in $G_{\Lambda}^{\thrd}$;\footnote{Since $\Lambda$ is stable, we know it satisfies the saturation condition $\scref$, meaning $(w,w) \in S^{\Lambda}_{[i]}$ for each $w \in \wstamod$. Hence, taking the Euclidean closure of $S^{\Lambda}_{[i]}$ implies that $\rstamod$ will be both reflexive and Euclidean,  i.e.\ it will be an equivalence relation.}

\item $\istamod = \{w \ | \ \relio w \in \rel \text{ and $w$ is unblocked.}\}$;

\item $w \in \vstamod(p)$ \iffi $w : \negnnf{p} \in \Gamma$ and $w$ is unblocked.

\end{itemize}
\end{definition}

\begin{lemma}\label{lem:reacability-is-equiv-rel} Let $\Lambda := \rel \sar \Gamma$ be a stable sequent with $w,u \in \wstamod$. If $w \ripath u$, then $\rstamod wu$.
\end{lemma}

\begin{proof} If $w \ripath u$, then since $\Lambda$ is stable, and satisfies both $\scref$ and $\sceuc$, it follows that $\rstamod wu$.
\end{proof}

We now show that any stability model is indeed a $\dsnk$-model of 
\dfn~\ref{def:frames-models}.

\begin{lemma}\label{lem:stability-model-is-DSn}
Let $\thrd$ be a thread from $\emptyset \sar w_{0} : \phi$ to the stable sequent $\Lambda := \rel \sar \Gamma$ generated by $\provenk(\emptyset \sar w_{0} : \phi) = \mathtt{false}$. Then, the stability model $\mstamod$ is a finite $\dsnk$-model.
\end{lemma}

\begin{proof} First, since $\Lambda$ is a stable sequent generated from $\provenk(\emptyset \sar w_{0} : \phi)$, 
$w_{0}\in 
\wstamod$ ensures its non-emptiness. We show that the model satisfies \textbf{(C1)}--\textbf{(C3)} and \textbf{(D1)}--\textbf{(D3)}.

\vspace{-0.1cm}
\begin{itemize}

\item[\textbf{(C1)}] To prove that \textbf{(C1)} holds, we show that for all $i \in \Ag$, $\rstamod \subseteq \wstamod \times \wstamod$ is an equivalence relation. By \dfn~\ref{def:stability-model}, we know that $\rstamod$ is the Euclidean closure of the $S^{\Lambda}_{[i]}$ relation. Since $\Lambda$ is stable, we know that it satisfies the saturation condition $\scref$, implying that $\rstamod$ is reflexive and Euclidean, and so, is an equivalence relation.

\item[\textbf{(C2)}] Suppose that $w_{1}, \ldots, w_{n} \in \wstamod$. Since $\Lambda$ is saturated, we know that it is $\ioat$-satisfied, meaning that an $\ioat$-label $u \in \lab(\Lambda)$ exists such that $w_{i} \ripath u$ for each $i \in \Ag$. By \rmk~\ref{rmk:ioat-label-unblocked}, we know that $u$ is unblocked, and so, $u \in \wstamod$. By the $\scref$ and $\sceuc$ saturation conditions, it follows from the fact that $w_{i} \ripath u$ for each $i \in \Ag$ that $R_{\agbox}w_{i}u \in \rel$ for each $i \in \Ag$. Since $w_{1}, \ldots, w_{n} \in \wstamod$, each label is unblocked by \dfn~\ref{def:stability-model}, and as mentioned above, $u$ is unblocked as well, implying that $\rstamod w_{i}u$ for each $i \in \Ag$. Thus, $\ioacond$ is satisfied in this case.

\item[\textbf{(C3)}] Suppose that $w_{0}, \ldots, w_{k} \in \wstamod$ and let $i \in \Ag$. Since $\Lambda$ is stable and satisfies $\scapc$, we know that $R_{\agbox}w_{j}w_{m} \in \rel$ holds for some $j,m \in \{0, \ldots, k\}$.  By definition, $\rstamod w_{j}w_{m}$ holds, implying the desired result.

\item[\textbf{(D1)}] Follows immediately from the definition of $\istamod$.

\item[\textbf{(D2)}] Let $i \in \Ag$. Since $\Lambda$ is stable, we know that $\Lambda$ satisfies $\scdii$, meaning that there exists an unblocked $w \in \wstamod$ such that $\relio w \in \rel$. Hence, by the definition of $\istamod$, $w \in \istamod$, implying that $\istamod$ is non-empty.

\item[\textbf{(D3)}] Let $i \in \Ag$, $w,u \in \wstamod$, and suppose that $w \in \istamod$ and $\rstamod wu$. Let us define $v \sim_{i}^{S} z$ \iffi $(v,z)$ or $(z,v) \in S^{\Lambda}_{\agbox}$. (NB. See \dfn~\ref{def:stability-model} above for the definition of $S^{\Lambda}_{\agbox}$.) We prove the claim by induction on the minimal length of an undirected path $w \sim_{i}^{S} v_{1}, \ldots, v_{n} \sim_{i}^{S} u$ from $w$ to $u$ in $S^{\Lambda}_{\agbox}$ and omit the case when the path is of length zero (i.e.\ $w = u$) as the case is trivial.

\textit{Base case.} Suppose that $(w,u) \in S^{\Lambda}_{\agbox}$; the case when $(u,w) \in S^{\Lambda}_{\agbox}$ is shown similarly. We have two cases to consider: either (i) $R_{\agbox}wu \in \rel$, or (ii) $R_{\agbox}wu \not\in \rel$.

(i) If $R_{\agbox}wu \in \rel$, then since $w \in \istamod$, we know that $\ideal w \in \rel$ by \dfn~\ref{def:stability-model}. By the saturation condition $\scdiii$, it follows that $\ideal u \in \rel$, and since $u \in \wstamod$, we know that $u$ is unblocked, implying that $u \in \istamod$.

(ii) If $R_{\agbox}wu \not\in \rel$, then $(w,u) \in S^{\Lambda}_{\agbox}$ must hold due to clause (ii) of the definition of $S^{\Lambda}_{\agbox}$ (see \dfn~\ref{def:stability-model} above). Therefore, there is a loop node $v$ with $u$ its loop ancestor such that $R_{\agbox}wv \in \rel$. Since $w \in \istamod$, we know that $\ideal w \in \rel$, implying that $\ideal v \in \rel$ by the saturation condition $\scdiii$. By \dfn~\ref{def:blocking}, it follows that $\ideal u \in \rel$ since $u$ is the loop ancestor of $v$. Since $u$ is unblocked, we have that $u \in \istamod$.

\textit{Inductive step.} Let $w \sim_{i}^{S} v_{1}, \ldots, v_{n} \sim_{i}^{S} u$ be a minimal undirected path between $w$ and $u$ whose length is $n+1$. By IH, we know that $v_{n} \in \istamod$, implying that $\ideal v_{n} \in \rel$ since $v_{n}$ is unblocked. (NB. Every label participating in $S^{\Lambda}_{\agbox}$ and $\rstamod$ is unblocked by definition.) Suppose that $(u,v_{n}) \in S^{\Lambda}_{\agbox}$; the case when $(v_{n},u) \in S^{\Lambda}_{\agbox}$ is shown similarly. We have two cases to consider: either (i) $R_{\agbox}uv_{n} \in \rel$, or (ii) $R_{\agbox}uv_{n} \not\in \rel$.

(i) If $R_{\agbox}uv_{n} \in \rel$, then by 
conditions $\scref$ and $\sceuc$, we know 
$R_{\agbox}v_{n}u \in \rel$. We may conclude that $u \in \istamod$ by an argument similar to case (i) of the base case above.

(ii) Suppose that $R_{\agbox}uv_{n} \not\in \rel$. Then, $(u,v_{n}) \in S^{\Lambda}_{\agbox}$ must hold due to clause (ii) of the definition of $S^{\Lambda}_{\agbox}$, implying the existence of a loop node $z$ with $v_{n}$ its loop ancestor such that $R_{\agbox}uz \in \rel$. Since $\ideal v_{n} \in \rel$ and $v_{n}$ is the loop ancestor of $z$, we have that $\ideal z \in \rel$. By the fact that $R_{\agbox}uz \in \rel$ and by saturation conditions $\scref$ and $\sceuc$, we know that $R_{\agbox}zu \in \rel$, showing that $\ideal u \in \rel$. It follows that $u \in \istamod$.

\end{itemize}
Last, we note that by condition $\scid$, the valuation $\vstamod$ is well-defined, and since $\mstamod$ was extracted from $\Lambda$, which is a finite object, $\mstamod$ will be finite. 
\end{proof}

We are now in the position to prove correctness.

\begin{theorem}[Correctness]\label{thm:correctness-prove} Let $n, k \in \mathbb{N}$ and $\phi \in \lang_{n}$.
\begin{itemize}

\item[(i)] If $\provenk(\emptyset \sar w_{0} :\phi) = \mathtt{true}$, then there is a proof of $\emptyset \sar w_{0} :\phi$ in $\gtdsnk$ witnessing that $\phi$ is $\dsnk$-valid. 

\item[(ii)] If $\provenk(\emptyset \sar w_{0} :\phi) = \mathtt{false}$, then there is finite $\dsnk$-model witnessing that $\phi$ is $\dsnk$-invalid.

\end{itemize}
\end{theorem}

\begin{proof} (i) 
Every instruction in $\provenk$ can be seen as a bottom-up application of a rule in $\gtdsnk$ or of an admissible rule from \fig~\ref{fig:admiss-rules-for-proof-search} with $\ioaop$ equivalent to a series of bottom-up $\ioa$ applications. 
 We can invoke \thrm~\ref{thm:gtdsnk-properties}-1 to obtain proofs of all top sequents.

(ii) Let us suppose that $\provenk(\emptyset \sar w_{0} :\phi) = \mathtt{false}$. It follows that $\provenk(\emptyset \sar w_{0} :\phi)$ generated a thread $\thrd = (\rel_{i} \sar \Gamma_{i})_{i=0}^{h}$ from $\emptyset \sar w_{0} : \phi$ to $\Lambda := \rel \sar \Gamma$ with $\Lambda$ stable. We may construct the stability model $\mstamod$ as specified in \dfn~\ref{def:stability-model}. By \lem~\ref{lem:stability-model-is-DSn}, we know that $\mstamod$ is a finite $\dsnk$-model. Let us now prove that for any $u \in \wstamod$ with $u : \psi \in \Gamma$, $\mstamod , u \not\Vdash \psi$. We show this by induction on the complexity of $\psi$ and note that since $w_{0} : \phi \in \Gamma$, the above claim implies that $\mstamod , w_{0} \not\Vdash \phi$, meaning that $w_{0} : \phi$ is not provable in $\gtdsnk$ by soundness (\thrm~\ref{thm:soundness-gtdsn}). We only consider the % cases of the 
inductive step when $\psi$ is of the form $\agdia \xi$, $\ODi \xi$, $\agbox \xi$, and $\Oi \xi$ since all remaining cases are simple or similar.

\textbf{$\agdia \xi$.} Suppose $u : \agdia \xi \in \Gamma$ and assume that for some arbitrary $v \in \wstamod$, $\rstamod uv$. We aim to show that $\mstamod, v \Vdash \xi$, and have two cases to consider: either (i) $R_{\agbox}uv \in \rel$, or (ii) $R_{\agbox}uv \not\in \rel$.

(i) If $R_{\agbox}uv \in \rel$, then since $\Lambda$ is stable, we know that $v : \xi \in \Gamma$. Hence, by IH and the definition of $\rstamod$ we have that $\mstamod, v \not\Vdash \xi$. 

(ii) Suppose that $R_{\agbox}uv \not\in \rel$. We define $u \sim_{i}^{S} v$ \iffi $(u,v)$ or $(v,u) \in S^{\Lambda}_{\agbox}$, and prove that $v : \xi, v : \agdia \xi \in \Gamma$ by induction on the minimal length of an undirected path $u \sim_{i}^{S} z_{1}, \ldots, z_{n} \sim_{i}^{S} v$ from $u$ to $v$ in $S^{\Lambda}_{\agbox}$. (NB. See \dfn~\ref{def:stability-model} above for the definition of $S^{\Lambda}_{\agbox}$.) Note that the path cannot be of length $0$, since then $u = v$, implying that $R_{\agbox}uu \not\in \rel$ by our assumption, which contradicts the fact that $R_{\agbox}uu \in \rel$ by the saturation condition $\scref$.

\textit{Base case.} Suppose that $(u,v) \in S^{\Lambda}_{\agbox}$; the case when $(v,u) \in S^{\Lambda}_{\agbox}$ is shown similarly. Then, $(u,v) \in S^{\Lambda}_{\agbox}$ must hold due to clause (ii) of the definition of $S^{\Lambda}_{\agbox}$ (\dfn~\ref{def:stability-model}). Hence, there is a loop node $z$ with $v$ its loop ancestor such that $R_{\agbox}uz \in \rel$. Since $u : \agdia \xi \in \Gamma$, we know that $z : \xi, z : \agdia \xi \in \Gamma$, implying that $v : \xi, v : \agdia \xi \in \Gamma$ as $v$ is the loop ancestor of $z$.

\textit{Inductive step.} Let $u \sim_{i}^{S} z_{1}, \ldots, z_{n} \sim_{i}^{S} v$ be a minimal undirected path between $u$ and $v$ whose length is $n+1$. By IH, we know that $z_{n} : \xi, z_{n} : \agdia \xi \in \Gamma$. Suppose that $(v,z_{n}) \in S^{\Lambda}_{\agbox}$; the case when $(z_{n},v) \in S^{\Lambda}_{\agbox}$ is shown similarly. We have two cases to consider: either (i) $R_{\agbox}vz_{n} \in \rel$, or (ii) $R_{\agbox}vz_{n} \not\in \rel$.

(i\textsuperscript{*}) If $R_{\agbox}vz_{n} \in \rel$, then by saturation conditions $\scref$ and $\sceuc$, we know 
$R_{\agbox}z_{n}v \in \rel$. We may
conclude that $v : \xi, v :\agdia \xi \in \Gamma$ since $\Lambda$ satisfies the saturation condition $\scagdia$.

(ii\textsuperscript{*}) Suppose that $R_{\agbox}vz_{n} \not\in \rel$. Then, $(v,z_{n}) \in S^{\Lambda}_{\agbox}$ must hold due to clause (ii) of the definition of $S^{\Lambda}_{\agbox}$, implying the existence of a loop node $z$ with $z_{n}$ its loop ancestor such that $R_{\agbox}vz \in \rel$. Since $z_{n}$ is the loop ancestor of $z$, we have that $z : \xi, z : \agdia \xi \in \Gamma$. By the fact that $R_{\agbox}vz \in \rel$ and by saturation conditions $\scref$ and $\sceuc$, we know that $v : \xi, v : \agdia \xi \in \Gamma$.\footnote{This case of our proof demonstrates the utility of the $\stitdialc$ rule. Since $z : \agdia \xi$ occurs in $\Gamma$, and $R_{\agbox}vz \in \rel$ implies $R_{\agbox}zv \in \rel$ by the saturation conditions $\scref$ and $\sceuc$, we know that $v : \xi, v : \agdia \xi \in \Gamma$. If we utilized the $\stitdiar$ rule rather than $\stitdialc$, we would not immediately have that $z : \agdia \xi \in \Gamma$.}
\label{pref:NB-on-admissible-stitdia-rule}

Hence, regardless of the length of a minimal path between $u$ and $v$, we have that $v : \xi, v : \agdia \xi \in \Gamma$, showing that $\mstamod, v \not\Vdash \xi$ by IH. As $v$ is arbitrary, may conclude that $\mstamod, u \not\Vdash \agdia \xi$.

\textbf{$\ODi \xi$.} Suppose $u : \ODi \xi \in \Gamma$. Since $\Lambda$ is stable we know that $v : \xi \in \Gamma$ for all $v$ such that $\relio v \in \rel$. Hence, by IH and the definition of $\istamod$ we have that $\mstamod, v \not\Vdash \xi$ for all $v \in \wstamod$ such that $\istamod v$, implying $\mstamod, u \not\Vdash \ODi \xi$.

\textbf{$\agbox \xi$.} Suppose $u : \agbox \xi \in \Gamma$. Since $\Lambda$ is stable, we know that 
there exists a $v \in \lab(\Lambda)$ such that $\reli uv \in \rel$ and $v :\xi \in \Gamma$. There are two cases: 
 (i) $v$ is a loop node or (ii) $v$ is not.

(i) If $v$ is a loop node, then there exists a loop ancestor $z$ (which is therefore unblocked) such that $z : \xi \in \Gamma$ and $\rstamod uz$. Hence, by IH, we have that for some $z \in \wstamod$, $\mstamod, z \not\Vdash \xi$.

(ii) By IH, the definition of $\wstamod$, and the definition of $\rstamod$, it follows that there exists a $v \in \wstamod$ such that $\rstamod uv$ and $\mstamod, v \not\Vdash \xi$. 

Hence, $\mstamod, u \not\Vdash \agbox \xi$.

\textbf{$\Oi \xi$.} Suppose $u : \Oi \xi \in \Gamma$. Since $\Lambda$ is stable, we know that there exists an unblocked $v \in \lab(\Lambda)$ such that $\relio v \in \rel$ and $v :\xi \in \Gamma$. By IH, the definition of $\wstamod$, and the definition of $\istamod$, it follows that there exists a $v \in \istamod$ such that $\istamod v$ and $\mstamod, v \not\Vdash \xi$. Hence, $\mstamod, u \not\Vdash \Oi \xi$.
\end{proof}

Let us now argue that $\provenk$ terminates regardless of its input $\emptyset \sar w_{0} : \phi$ for $\phi \in \lang_{n}$. We first define the set of subformulae of a given formula $\phi$ (Definition~\ref{def:subformula}), which due to its finiteness 
limits the formulae that can occur during proof-search. We then prove 
that only a finite number of labels can be created during proof-search, implying that only a finite number of sequents can be generated, and yielding the termination of $\provenk$.

\begin{definition}[Subformula]\label{def:subformula} Let $\phi$ be a formula in $\lang_{n}$. We define the set $\sufo(\phi)$ of subformulae of $\phi$ recursively as follows:
\begin{itemize}

\item $\sufo(\phi) = \{\phi\}$ for $\phi \in \{p, \negnnf{p} \ | \ p \in Var\}$;

\item $\sufo(\triangledown \psi) = \{\triangledown \psi\} \cup sufo(\psi)$ for $\triangledown\in \{\Box, \Diamond\} \cup \{\Oi, \agbox, \ODi, \agdia \ | \ i \in Ag\}$;

\item  $\sufo(\psi \circ \chi) = \sufo(\psi) \cup \sufo(\chi)$ for $\circ \in \{\land, \lor\}$.

\end{itemize}
\end{definition}

\begin{lemma}\label{lem:termination-label-intro} Let $n, k \in \mathbb{N}$ and $\phi \in \lang_{n}$.
\begin{enumerate}

\item For each $\Box \psi \in \sufo(\phi)$, $\settlc$ and $\settr$ are applied bottom-up at most once in any thread during the computation of $\provenk(\emptyset \sar w_{0} : \phi)$;

\item For each $\Oi \psi \in \sufo(\phi)$, $\Oilc$ and $\Oir$ are applied bottom-up at most once in any thread during the computation of $\provenk(\emptyset \sar w_{0} : \phi)$;

\item For each $i \in \Ag$, $\dtwo$ is applied bottom-up at most once in any thread during the computation of $\provenk(\emptyset \sar w_{0} : \phi)$;

\item $\stitr$ and $\stitlc$ will be applied only a finite number of times in any thread of the computation of $\provenk(\emptyset \sar w_{0} : \phi)$;

\item $\ioaop$ will be applied only a finite number of times in any thread of the computation of $\provenk(\emptyset \sar w_{0} : \phi)$.

\end{enumerate}
\end{lemma}

\begin{proof} Claims 1 and 2 are straightforward since once $\settr$ and $\Oir$ are applied bottom-up (which is preceded by an application of $\settlc$ and $\Oilc$, respectively, as discussed in \rmk~\ref{rmk:admiss-rules-use}) with $w_{0} : \Box \psi$ or $w_{0} : \Oi \psi$ principal, then the saturation conditions $\scbox$ and $\scoi$ will hold for \emph{any} labeled formula of the form $u : \Box \psi$ or $u : \Oi \psi$ for the remainder of proof-search. Similarly, once $\dtwo$ is applied bottom-up for an agent $i \in \Ag$, it never need be applied for that agent again as $\scdii$ will continue to be satisfied for the remainder of proof-search. We therefore dedicate the remainder of the proof to showing claims 4 and 5.

4. Observe that any generation tree defined relative to a thread $\thrd$ generated during the computation of $\provenk(\emptyset \sar w_{0} : \phi)$ has a finite branching factor bounded by the number of agents plus the number of $\Box$, $\Oi$, and $\agbox$ subformulae of $\phi$. Moreover, since each label $u$ can only be associated with formulae from $\{\ideal u \ | \ i \in \Ag\} \cup \sufo(\phi)$, which is a finite set, we know the depth of the generation tree is %also
bounded by a finite number as eventually a loop node must occur along any given path in the generation tree. Hence, since each bottom-up application of $\stitr$ and $\stitlc$ corresponds to an edge in a generation tree, we have that $\stitr$ and $\stitlc$ can be bottom-up applied only a finite number of times in any thread.

5. Note that only $\ioaop$ and bottom-up applications of $\settr$, $\Oir$, $\dtwo$, $\stitr$, and $\stitlc$ add fresh labels during proof-search. By claims 1--4 above, we know that if we consider any thread of proof-search then eventually all possible labels that could be added via $\settr$, $\Oir$, $\dtwo$, $\stitr$, and $\stitlc$ will have been added. After such a point, $\ioaop$ will eventually be applied if the top sequent of a thread is not yet $\ioat$-satisfied, yielding an $\ioat$-satisfied sequent $\Lambda$ due to \lem~\ref{lem:IOA-sat}. Since no other bottom-up applications of rules introduce labels that could violate the $\ioat$-satisfiability of a sequent, all sequents generated after $\Lambda$ will be $\ioat$-satisfied as well, meaning that $\ioaop$ will never be applied again. Hence, $\ioaop$ will only be applied a finite number of times in any thread during proof-search.
\end{proof}

\begin{theorem}[Termination]\label{thm:termination-prove} For each $n,k \in \mathbb{N}$ and $\phi \in \lang_{n}$, $\provenk(\emptyset \sar w_{0} : \phi)$ terminates.
\end{theorem}

\begin{proof} By \lem~\ref{lem:termination-label-intro} above we know that only a finite number of labels will be introduced during the computation of $\provenk(\emptyset \sar w_{0} : \phi)$. Furthermore, each label will only be associated with a finite number of formulae from $\sufo(\phi)$ and a finite number of relational atoms. Hence, only a finite number of possible sequents can be generated during proof-search. Since each recursive call of proof-search, i.e.\ bottom-up application of a rule, creates a strictly larger sequent (as each recursive call introduces at least one new relational atom or labeled formula to a sequent), $\provenk$ never generates the same sequent twice in any given thread, implying that the algorithm will eventually terminate due to the finite space of sequents that can be generated.
\end{proof}

Last, as a consequence of the above we obtain decidability and the finite model property.

\begin{corollary}[Decidability and FMP]\label{cor:fmp}
For each $n, k \in \mathbb{N}$, the logic $\dsnk$ is decidable and has the finite model property.
\end{corollary}

\begin{figure}[t]
    \centering
\scalebox{0.92}{
\begin{tikzpicture}
\pgftransformscale{.8}

\node (w) [ ] {$\overset{\boxed{\Diam [1] p, \Diam [2] q, [1]p, [2]q}}{w}$};

\node (u0) [below left=of w,xshift=2cm] {$\overset{\boxed{p, [1]p, [2]q}}{u_{0}}$};

\node (u1) [below right=of w,xshift=-2cm] {$\overset{\boxed{q, [1]p, [2]q}}{u_{1}}$};

\node (u2) [below=of u0,yshift=-.3cm] {$\overset{\boxed{q, [1]p, [2]q}}{u_{2}}$};

\node (u3) [below=of u2,yshift=-.3cm] {$\overset{\boxed{p, [1]p, [2]q}}{u_{3}}$};

\node (u4) [below=of u3,yshift=-.3cm] {$\overset{\boxed{q, [1]p, [2]q}}{u_{4}}$};

\node (u5) [below=of u1,xshift=-2cm,yshift=-.75cm] {$\vdots$};

\node (u6) [below=of u4,yshift=1.25cm] {$\vdots$};

%%%Edges
\draw[-] (w) to node[above] {1} (u0);
%\draw[->] (u0) to [bend right=15] node[below] {1} (w);

\draw[-] (w) to node[above] {2} (u1);
%\draw[->] (u1) to [bend right=15] node[above] {2} (w);

\draw[-] (u0) to node[left] {2} (u2);
%\draw[->] (u2) to [bend right=15] node[right] {2} (u0);

\draw[-] (u1) to node[above] {1} (u0);
%\draw[->] (u0) to [bend right=10] node[below] {1} (u1);

\draw[-] (u2) to node[left] {1} (u3);

\draw[-] (u1) to node[above] {2} (u2);

\draw[-] (u3) to node[left] {2} (u4);

\draw[-] (u1) to node[above] {1} (u3);

\node[draw, dashed, color=black, rounded corners, inner xsep=20pt, inner ysep=5pt, fit=(w)(u0)(u1)] (c1) {}; 

\node (cc1) [below right=of c1,xshift=-1cm,yshift=1.5cm] {$\mathrm{S_{1}}$};

\node[draw, dashed, color=black, rounded corners, inner xsep=20pt, inner ysep=5pt, fit=(w)(u2)(u1)(c1)] (c2) {}; 

\node (cc2) [below right=of c2,xshift=-1cm,yshift=1.5cm] {$\mathrm{S_{2}}$};

\node[draw, dashed, color=black, rounded corners, inner xsep=20pt, inner ysep=5pt, fit=(w)(u3)(u1)(c2)] (c3) {}; 

\node (cc3) [below right=of c3,xshift=-1cm,yshift=1.5cm] {$\mathrm{S_{3}}$};

\end{tikzpicture}
}
\caption{Three `snapshots' of proof-search without loop-checking, showing that proof-search will continue ad infinitum if loop-checking is not enforced.}\label{fig:loop}
\end{figure}

\subsection{Necessity of Loop-Checking in Proof-Search}

We show that proof-search is not guaranteed to terminate without loop-checking, thus motivating its introduction and use. 
We provide an example of a (non-deontic) $\sti$ formula 
$\Diam [1] p \lor \Diam [2] q$, such that proof-search in $\mathsf{G3DS}_{2}^{2}$ does not terminate if we disregard loop-checking. 
In order to improve readability, we represent labeled sequents as labeled graphs, and explain how these graphs are expanded during proof-search:

\begin{definition}[Graph of Labeled Sequent] Let $\Lambda = \rel \sar \Gamma$ be a labeled sequent. We define its \emph{corresponding graph} to be the triple $G(\Lambda) = (V,E,L)$ such that $V = \lab(\Lambda)$, $E = \{(w,u,i) \ | \ R_{\agbox}wu \in \rel\}$, and $L(w) = \Gamma \restriction w$.
\end{definition}

Now, let us suppose that our proof-search algorithm omits loop-checking, meaning that lines 51--57 (see p.~\pageref{alg:with-loop-check}) are 
removed from $\provenk$ and 
replaced with the following lines of pseudo-code that simply apply the $\stitr$ rule bottom-up when the condition is met.

\medskip

\begin{algorithm}[H]
\If{$w : \agbox \phi \in \Gamma$, but $u : \phi \not\in \Gamma$ for all $u \in \lab(\Lambda)$ such that $R_{\agbox}wu \in \rel$}
    {
    Let $\rel' := R_{\agbox}wv, \rel$ and $\Gamma' := v :\phi, \Gamma$ with $v$ fresh;\\
    \Return $\provenk$($\rel' \sar \Gamma'$);
    }
\caption{Applying the $\stitr$ Rule without Loop-checking}
\end{algorithm}

\medskip

Now, suppose our proof-search algorithm takes $\sar w : \Diam [1] p \lor \Diam [2] q$ as an input. Then, the derivation shown below will initially be constructed yielding the labeled node $w$ shown at the top of the graph in \fig~\ref{fig:loop} (which corresponds to the top sequent in the derivation). The loops/relational atoms $R_{[1]}ww, R_{[2]}ww$ have been omitted in the graph. For simplicity, we always omit the presentation of loops in \fig~\ref{fig:loop} and assume their presence.

\begin{center}
\scalebox{0.93}{
\AxiomC{$R_{[1]}ww, R_{[2]}ww \sar w : \Diam [1] p, \Diam [2] q, w : [1] p, w : [2] q$}
\RightLabel{$\diar$}
\UnaryInfC{$R_{[1]}ww, R_{[2]}ww \sar w : \Diam [1] p, \Diam [2] q, w : [1] p$}
\RightLabel{$\diar$}
\UnaryInfC{$R_{[1]}ww, R_{[2]}ww \sar w : \Diam [1] p, \Diam [2] q$}
\RightLabel{$\disr$}
\UnaryInfC{$R_{[1]}ww, R_{[2]}ww \sar w : \Diam [1] p \lor \Diam [2] q$}
\RightLabel{$(\mathsf{Ref}_{2})$}
\UnaryInfC{$R_{[1]}ww \sar w : \Diam [1] p \lor \Diam [2] q$}
\RightLabel{$(\mathsf{Ref}_{1})$}
\UnaryInfC{$\sar w : \Diam [1] p \lor \Diam [2] q$}
\DisplayProof
}
\end{center}

As the top sequent in the derivation above is not stable, our algorithm will continue proof-search, thus yielding the derivation shown below, whose top sequent corresponds to the graph shown in `snapshot 1', i.e.\ the dashed rectangle $\mathrm{S_{1}}$ in \fig~\ref{fig:loop}, though without the horizontal $1$-edge between $u_{0}$ and $u_{1}$. We refer to this graph as $G_{1}$. 
 To improve readability, and due to the size of the labeled sequents occurring in the derivation below, we have included ellipses to indicate the side formulae that are inherited from lower sequents. In 
 \fig~\ref{fig:loop}, we have opted to denote pairs of relational atoms that form a 2-cycle (e.g.\ $R_{[1]}wu_{0}, R_{[1]}u_{0}w$ and $R_{[1]}wu_{1}, R_{[2]}u_{1}w$) as a single, undirected edge to simplify presentation.

\begin{center}
\scalebox{0.93}{
\AxiomC{}
\UnaryInfC{$R_{[1]}wu_{1}, \ldots \sar u_{1} : q, u_{1} : [1] p, u_{1} : [2] q, \ldots$}
\RightLabel{$\diar$}
\UnaryInfC{$R_{[1]}wu_{1}, \ldots \sar u_{1} : q, u_{1} : [1] p, \ldots$}
\RightLabel{$\diar$}
\UnaryInfC{$R_{[1]}wu_{1}, R_{[2]}u_{1}w, \ldots \sar u_{1} : q, \ldots, w : [2] q$}
\RightLabel{$(\mathsf{Euc}_{2})$}
\UnaryInfC{$R_{[1]}wu_{1}, R_{[1]}u_{1}u_{1}, R_{[2]}u_{1}u_{1}, \ldots \sar u_{1} : q, \ldots, w : [2] q$}
\RightLabel{$(\mathsf{Ref}_{2})$}
\UnaryInfC{$R_{[1]}wu_{1}, R_{[1]}u_{1}u_{1}, \ldots \sar u_{1} : q, \ldots, w : [2] q$}
\RightLabel{$(\mathsf{Ref}_{1})$}
\UnaryInfC{$R_{[2]}wu_{1}, \ldots \sar u_{1} : q, \ldots, w : [2] q$}
\RightLabel{$([2])$}
\UnaryInfC{$R_{[1]}wu_{0}, \ldots \sar u_{0} : p, u_{0} : [1] p, u_{0} : [2] q, \ldots, w : [2] q$}
\RightLabel{$\diar$}
\UnaryInfC{$R_{[1]}wu_{0}, \ldots \sar u_{0} : p, u_{0} : [1] p, \ldots, w : [2] q$}
\RightLabel{$\diar$}
\UnaryInfC{$R_{[1]}wu_{0}, R_{[1]}u_{0}w, \ldots \sar u_{0} : p, \ldots, w : [2] q$}
\RightLabel{$(\mathsf{Euc}_{1})$}
\UnaryInfC{$R_{[1]}wu_{0}, R_{[1]}u_{0}u_{0}, R_{[2]}u_{1}u_{1}, \ldots \sar u_{0} : p, \ldots, w : [2] q$}
\RightLabel{$(\mathsf{Ref}_{2})$}
\UnaryInfC{$R_{[1]}wu_{0}, R_{[1]}u_{0}u_{0}, \ldots \sar u_{0} : p, \ldots, w : [2] q$}
\RightLabel{$(\mathsf{Ref}_{1})$}
\UnaryInfC{$R_{[1]}wu_{0}, \ldots \sar u_{0} : p, \ldots, w : [2] q$}
\RightLabel{$([1])$}
\UnaryInfC{$R_{[1]}ww, R_{[2]}ww \sar w : \Diam [1] p, \Diam [2] q, w : [1] p, w : [2] q$}
\DisplayProof
}
\end{center}

One can see that $\mathbf{(C_{[1]})}$ and $\mathbf{(C_{[2]})}$ are unsatisfied in $G_{1}$ due to the occurrence of $[1]p$ at $u_{1}$ and $[2]q$ at $u_{0}$, respectively. We assume w.l.o.g. that proof-search applies $([2])$ bottom-up, creating a new label/vertex $u_{2}$ emanating from $u_{0}$, which can be seen in \fig~\ref{fig:loop}. Proof-search will then apply $(\mathsf{Ref}_{1})$ and $(\mathsf{Ref}_{2})$, introducing loops at $u_{2}$, followed by an application of $(\mathsf{Euc}_{2})$, introducing a $2$-edge between $u_{2}$ and $u_{0}$. Recall that $k = 2$, that is, each agent has a maximum of two choices. Since no relational atom of the form $R_{[1]}vz$ with $v, z \in \{u_{0}, u_{1}, u_{2}\}$ and $v \neq z$ exists within our labeled sequent/graph constructed thus far, i.e.\ there is no $1$-edge between $u_{0}$, $u_{1}$, or $u_{2}$, the $\mathbf{(C_{APC})}$ saturation condition is not satisfied. Hence, proof-search will apply $(\mathsf{APC}_{1}^{2})$ bottom-up, yielding a premise with 
$R_{[1]}u_{1}u_{0}$, followed by an application of $(\mathsf{Euc}_{1})$, introducing 
$R_{[1]}u_{0}u_{1}$. After applying the $\settdiar$ rule bottom-up twice, introducing the formulae $[1]p$ and $[2]q$ at $u_{2}$, we obtain the labeled sequent/graph shown in `snapshot 2' in \fig~\ref{fig:loop}, i.e.\ the graph in the dashed rectangle $\mathrm{S_{2}}$, though without the diagonal $2$-edge between $u_{1}$ and $u_{2}$. We refer to this graph as $G_{2}$.

Observe that the $\mathbf{(C_{[1]})}$ saturation condition is not satisfied due to the occurrence of $[1]p$ at $u_{2}$ in $G_{2}$. Hence, 
the algorithm will apply the $([1])$ rule bottom up, introducing the label $u_{3}$ protruding from $u_{2}$, followed by applications of $(\mathsf{Ref}_{1})$, $(\mathsf{Ref}_{2})$, and $(\mathsf{Euc}_{1})$. One can see that $\mathbf{(C_{APC})}$ is not satisfied since a $2$-edge does not exist between $u_{1}$, $u_{2}$, or $u_{3}$, and so, $(\mathsf{APC}_{2}^{2})$ will be applied bottom-up. This rule application, in conjunction with an application of $(\mathsf{Euc}_{2})$,  introduces $R_{[2]}u_{1}u_{2}$ and $R_{[2]}u_{2}u_{1}$, which are represented as an undirected $2$-edge between $u_{1}$ and $u_{2}$. After two applications of the $\settdiar$ rule, we obtain the labeled sequent/graph in `snapshot 3' in \fig~\ref{fig:loop}, i.e.\ the graph in the dashed rectangle $\mathrm{S_{3}}$, although without the diagonal $1$-edge between $u_{1}$ and $u_{3}$. We call this graph $G_{3}$.

The $\mathbf{(C_{[1]})}$ saturation condition is not satisfied in $G_{3}$ due to the occurrence of $[2]q$ at $u_{3}$. As a consequence, our proof-search algorithm will apply the $([2])$ rule bottom up, introducing the label $u_{4}$ protruding from $u_{3}$, followed by applications of $(\mathsf{Ref}_{1})$, $(\mathsf{Ref}_{2})$, and $(\mathsf{Euc}_{2})$. At this stage, $\mathbf{(C_{APC})}$ is not satisfied in the graph since a $1$-edge does not exist between $u_{1}$, $u_{3}$, or $u_{4}$, and so, $(\mathsf{APC}_{1}^{2})$ will be applied bottom-up. This rule application yields a premise that, in conjunction with an application of $(\mathsf{Euc}_{1})$, will introduce $R_{[1]}u_{1}u_{3}$ and $R_{[1]}u_{3}u_{1}$ (represented as the undirected $1$-edge between $u_{1}$ and $u_{3}$). After two applications of the $\settdiar$ rule, we obtain the labeled sequent/graph shown in \fig~\ref{fig:loop}. By repeating the above pattern, one can see that the left path in the labeled sequent/graph will continue to grow through alternating bottom-up applications of the $([1])$ and $([2])$ rules, i.e.\ proof-search does not terminate for the formula $\Diam [1]p \lor \Diam [2]q$ in the absence of loop-checking.

We emphasize that, since the chosen formula is $\ODi$- and $\Oi$-free, the above example shows the importance of a loop-checking mechanism for (non-deontic) $\sti$ as well.\footnote{
Although the proof-search algorithm for non-deontic $\sti$ logic by \citeauthor{Negri2020} \citeyear{Negri2020} is susceptible to the above counter-example, we conjecture that a modified version of our loop-checking mechanism can be used to ensure termination in their algorithm.}

\section{Applications: Duty, Compliance, and Joint Fulfillment Checking}\label{Sect:Applications}

Proof-search algorithms make a logic suitable for reasoning tasks. In this last section, we discuss three tasks of particular interest to agent-based normative reasoning:

\begin{description}

\item \textit{Duty Checking:} Determine an agent’s obligations relative to a given knowledge base.

\item
\textit{Compliance Checking}: Determine if a choice, considered by an agent as potential conduct, complies with the given knowledge base.

\item
\textit{Joint Fulfillment Checking}: Determine whether under a specified factual context an agent can (still) jointly fulfill all their duties.
\end{description}
The referred knowledge base is a set consisting of obligations and facts and, in what follows, we distinguish between the norm base and the factual context of the knowledge base. We provide an example for each of these applications. For the sake of readability, we consider a single-agent setting in all cases, but the approach generalizes to multi-agent reasoning. In considering the first two tasks, we assume that the given knowledge base is consistent, unless stated otherwise (we note that consistency can also be determined through proof-search).

 As argued for in \sect~\ref{Sect:Intro}, the derivations and (counter-)models obtained through our applications 
have explanatory value\footnote{For some purposes, formal proofs may be less explanatory than informal ones 
containing reasoning gaps. This is especially true when the intended explainee is 
a layperson \cite{Mil19}. For experts, derivations and interpretable models serve as appropriate explanations; cf. \cite{BirCot17}. 
} since both provide justifications of (non-)theoremhood by representing a constructive step-by-step reasoning process interpretable by the explainee.

\subsection{Duty Checking}\label{sect5.1}

Consider the following scenario: Yara borrowed a hammer from a friend (Kai) in order to repair a leaking shed. Furthermore, Yara promised to return the hammer to Kai before noon. Let \noon be the proposition ``Yara arrives at Kai's place before noon (to return the hammer).''  We say that (due to the promise) Yara is under the obligation to return the hammer before noon, i.e.\  $\otimes_{y} \noon$. We assume that Yara is free to see to it that she arrives before noon or not, represented as $\Diam [y] \noon \land \Diam [y] \negnnf{\noon}$. Now, Yara knows that Kai lives around the corner and considers her means of going there, e.g.\ she could go by foot or take the car. Suppose Yara entertains the idea to go by foot, i.e.\ $\Diam[y]\foot$. Moreover, let us assume that Yara knows, as a fact, that walking will ensure that she arrives by noon, i.e.\ $\Box(\foot\rightarrow\noon)$. 
 
 The knowledge base for this scenario is the collection
$$
\Sigma:=\{\Oy\noon, \Diam[y] \noon,  \Diam[y]\lnot \noon, \Diam[y]\foot, \Diam[y]\negnnf{\foot}, \Box(\foot\rightarrow\noon)\},
$$
though we simplify it to $\Sigma:=\{\Oy\noon, \Diam[y] \negnnf{\noon}, \Diam[y]\foot, \Box(\foot\rightarrow\noon)\}$ as $\Diam[y]\noon$ and $\Diam[y]\negnnf{\foot}$ are implied by the other formulae. Since Yara is obliged to meet Kai, Yara wonders whether she has the obligation to go by foot, i.e.\ $\Oy\foot$. Accordingly, we check whether the formula
 $$
\phi := (\Oy\noon\land\Diam[y]\lnot \noon \land \Diam[y]\foot\land \Box(\foot\rightarrow\noon))\rightarrow \Oy\foot
$$
is valid, i.e.\ we can check whether $w : \phi$ is derivable in $\mathsf{G3DS}^{0}_{1}$, where $\Ag = \{y\}$ and the limited choice parameter $k = 0$ since we have not enforced an \textit{a priori} limit on Yara's number of choices. Hence, we check whether the duty $\otimes_y\foot$ is implied by the context $\Sigma$. 

To test the validity of $\phi$, let us run our proof-search algorithm $\mathtt{Prove}_{1}^{0}(\emptyset \sar w : \phi)$. Doing so, one finds that $\phi$ is invalid as the stable sequent $\Lambda = \rel \sar \Gamma_{0}, \Gamma_{1}, \Gamma_{2}, \Gamma_{3}$ is generated via proof-search, where each set in $\Lambda$ is as follows:
\begin{description}
 
\item $\rel := R_{[y]}ww, R_{[y]}uu, R_{[y]}vv, R_{[y]}zz, I_{\otimes_{y}}z$;

\item $\Gamma_{0} := w : \ominus_{y} \negnnf{\noon}, w : \Box \langle y \rangle \noon, w : \Box \langle y \rangle \negnnf{\foot}, w : \Diam(\foot \land \negnnf{\noon}), w : \foot \land \negnnf{\noon}, w : \foot, w : \otimes_{y} \foot$;

\item $\Gamma_{1} := u : \langle y \rangle \noon, u : \noon, u : \foot \land \negnnf{\noon}, u : \foot$;

\item $\Gamma_{2} := v : \langle y \rangle \negnnf{\foot}, v : \negnnf{\foot}, v : \foot \land \negnnf{\noon}, v : \negnnf{\noon}$;

\item $\Gamma_{3} := z : \foot, z : \foot \land \negnnf{\noon}, z : \negnnf{\noon}$.
 
\end{description}
As explained in the previous section, we may transform $\Lambda$ into a counter-model $\mstamod$ for $\phi$. We define the stability model $M := (W,R_{[y]},I_{\Oy},V)$ in accordance with \dfn~\ref{def:stability-model}, where $W=\{w,u,v,z\}$, $R_{[y]}=\{(w,w),(u,u),(v,v),(z,z)\}$, $I_{\Oy}=\{z\}$, and $V(\noon) = \{w,v,z\}$ and $V(\foot)=\{v\}$. A graphical representation of $\mstamod$ is presented in \fig~\ref{fig:countermodel_1}.

\begin{figure}
\begin{center}
\scalebox{0.93}{
\begin{tikzpicture}

\node[%point
] (w) [] {$w: \no,\negnnf{\fo}$}; 

\node[%point
] (v) [right=of w, xshift=0mm,yshift=0mm] {$v:\no,\fo$};

\node[%point
] (u) [below=of w, xshift=0mm,yshift=0mm] {$u: \negnnf{\no},\negnnf{\fo}$};

\node[%point
] (z) [right=of u, xshift=0mm,yshift=0mm] {\textcolor{white}{$z: \no,\negnnf{\fo}$}};

\node[draw, very thick, %dotted,
 color=black, rounded corners, inner xsep=10pt, inner ysep=10pt, fit=(w) (v) (u) (z)] (m1) [label={[xshift=3cm, yshift=-3.5cm]}] {}; %$R_{\Box}(w)$}] {}; %C

\node[
 very thick, %purple
, rounded corners,  inner xsep=6pt, inner ysep=6pt, fit=(z), pattern=north west lines, pattern color=Gray2] (c1) {}; %C

\node[%point
] (v1) [right=of u, xshift=0mm,yshift=0mm] {$z: \no,\negnnf{\fo}$};

\node[draw,  thick, densely dashed, color=black %purple
, rounded corners, inner xsep=6pt, inner ysep=6pt, fit=(v)] (c1) {}; %C

\node[draw,  thick, densely dashed, color=black %purple
, rounded corners, inner xsep=6pt, inner ysep=6pt, fit=(z)] (c1) {}; %C

\node[draw,  thick, densely dashed, color=black %purple
, rounded corners, inner xsep=6pt, inner ysep=6pt, fit=(w)] (c1) {}; %C

\node[draw,  thick, densely dashed, color=black %purple
, rounded corners, inner xsep=6pt, inner ysep=6pt, fit=(u)] (c1) {}; %C

\end{tikzpicture}
}
\end{center}
\vspace{-0.5cm}
\caption{The stability model $\mstamod$ falsifying $\phi := (\Oy\noon\land\Diam[y]\lnot \noon \land \Diam[y]\foot\land \Box(\foot\rightarrow\noon))\rightarrow \Oy\foot$ in the logic $\mathsf{DS}^{0}_{1}$.\label{fig:countermodel_1}} %
\end{figure}

One may readily verify that $\phi$ is indeed falsified on $\mstamod$ at $w$, showing that Yara is not obliged to go to Kai's place by foot. This fact is explained by the above counter-model, namely, by the world $z$ which shows that Yara can consistently satisfy her obligation to arrive before noon without going by foot (perhaps she can take the bus). 

However, if going by foot is the \emph{sole means} to ensure a timely arrival---i.e.\ if we additionally assume that $\Box(\noon\rightarrow\foot)$---then it becomes obligatory for Yara to travel by foot. In this case, we say that going by foot is a necessary and sufficient condition for Yara to fulfill her duties. 
The following derivation in $\mathsf{G3DS}^{0}_{1}$ (which may be found by means of our proof-search algorithm) witnesses her obligation to travel by foot under the assumption $\Box(\noon\rightarrow\foot)$ as opposed to $\Box(\foot\rightarrow\noon)$. 
The right branch of the proof, denoted by $\Pi$, is a straightforward adaptation of the left branch.

\bigskip

\hspace{-1.2cm}{
\scalebox{0.91}{
\AxiomC{}
\RightLabel{$\id$}
\UnaryInfC{$\ldots \sar \ldots, w : \no, u : \foot, u : \noon, u : \negnnf{\noon}$}
\RightLabel{$\ODir$}
\UnaryInfC{$\ldots \sar \ldots, w : \no, u : \foot, u : \noon$}
\RightLabel{$\settdiar$}

\AxiomC{}
\RightLabel{$\id$}
\UnaryInfC{$\ldots \sar \ldots, w : \no, u : \foot, u : \negnnf{\foot}$}
\RightLabel{$\settdiar$}

\RightLabel{$\conr$}
\BinaryInfC{$R_{[y]}ww, R_{[y]}uu, I_{\otimes_{y}}u \sar \ldots, w : \no, u : \foot, u : \noon \land \negnnf{\foot}$}
\RightLabel{$\settdiar$}
\UnaryInfC{$R_{[y]}ww, R_{[y]}uu, I_{\otimes_{y}}u \sar \ldots, w : \no, u : \foot$}
\RightLabel{$(\mathsf{Ref}_{y})$}
\UnaryInfC{$R_{[y]}ww, I_{\otimes_{y}}u \sar \ldots, w : \no, u : \foot$}
\RightLabel{$\Oir$}
\UnaryInfC{$R_{[y]}ww, \sar \ldots, w : \no$}

\AXC{$\Pi$}
\UnaryInfC{$R_{[y]}ww, \sar \ldots, w : \negnnf{\fo}$}
\RL{$\conr$}

\insertBetweenHyps{\hskip -15pt}
\BIC{$R_{[y]}ww \Rightarrow w: \Ody \negnnf{\no}, w: \Diam(\no \land \negnnf{\fo}), w : \no \land \negnnf{\fo},  w: \Oy \fo$}\RL{$\diar$}
\UIC{$R_{[y]}ww \Rightarrow w: \Ody \negnnf{\no}, w: \Diam(\no \land \negnnf{\fo}), w: \Oy \fo$}\RL{$(\lor) \times 2$}

\UIC{$R_{[y]}ww \Rightarrow w: \Ody \negnnf{\no} \lor \Diam(\no \land \negnnf{\fo}) \lor \Oy \fo$}\RL{$(\mathsf{Ref}_{y})$}
\UIC{$\Rightarrow w: \Ody \negnnf{\no} \lor \Diam(\no \land \negnnf{\fo}) \lor \Oy \fo$}\RL{=}
\dashedLine
\UIC{$\Rightarrow w: (\Oy \no \land \Box (\no\rightarrow \fo))\rightarrow \Oy \fo$}
\DisplayProof
}
}
\bigskip

We may transform the above into a proof of $(\Oy \no \land \Diam[y] \negnnf{\no} \land \Diam[y] \fo \land \Box (\no\rightarrow \fo))\rightarrow \Oy \fo$ by weakening in the additional assumptions accordingly, using $\wk$. Hence, 
Yara does have the obligation to see to it that she goes by foot (i.e.\ $\Oy\fo$) given the strengthened assumption that going by foot is the only means by which Yara can get to Kai's by noon (i.e.\ $\Box(\noon\rightarrow\foot)$). The above discussion generalizes to arbitrary finite scenarios.

\subsection{Compliance Checking}

To determine whether a considered choice complies with a given knowledge base, consisting of a norm base and a factual context, we must determine whether the considered choice does not conflict with an obligation entailed by the norm base under the given factual context (recall that we assume the knowledge base to be consistent in this scenario). 

We emphasize that checking whether an agent’s conduct is \textit{consistent} with a given knowledge base does not define compliance checking. Consider the simple example where Jade is obliged to drive on the left-hand side of the road, i.e.\ $\otimes_j\leftj$ and where she considers driving on the right-hand side instead, i.e.\ $[j]\rightj$. Furthermore, assume that driving left and right are mutually exclusive, i.e.\ $\Box(\leftj\rightarrow\lnot\rightj)$. Such a scenario can be consistently modeled by a single-agent two-choice moment with one choice---the obligatory one---representing Jade driving on the left and another choice with Jade cycling on the right. The scenario is consistent, but there is no compliance: Jade violates her obligation by driving on the right side of the road. A choice is not compliant whenever its performance entails a norm violation. 
Therefore, we check compliance by determining whether there is no explicit obligation to the contrary of what the agent considers as potential conduct.

Reconsider the 
scenario where Yara promised Kai to return the borrowed hammer before noon, i.e.\ $\otimes_y\noon$. She considers going there by foot or by car and both are available choices, i.e.\ $\Diam[y]\foot \land\Diam[y]\car$. Yara wonders whether taking the car---i.e.\ $[y]\car$---complies with her promise. 
The knowledge base formalizing this situation is the following set of formulae:
$$\Sigma:=\{  \otimes_y\noon , \Diam[y]\foot \land\Diam[y]\car \}$$

To see if the choice $[y]\car$ complies with $\Sigma$ 
we need to check whether the formula 
$$(\otimes_y\noon\land\Diam[y]\foot\land\Diam[y]\car) \rightarrow \otimes_y\lnot \car$$
is derivable. If not, then this means there is no implicit obligation for Yara to not take the car and, so, taking the car complies with the current normative situation expressed by $\Sigma$. 

Clearly, 
$[y]\car$ complies due to the syntactic independence with $\otimes_y \noon$. 
As can be seen, compliance checking is a form of duty checking, where it checks for a duty to the contrary of what the agent considers doing. 
Since we already discussed successful proof-search and counter-model extraction in the context of duty checking, our discussion above suffices, and we omit demonstration of failed proof-search.

\subsection{Joint Fulfillment Checking} 
\renewcommand{\clean}{\noon}
\renewcommand{\cl}{\clean}
\renewcommand{\po}{\post}

Last, 
we are interested in checking whether a consistent norm base together with a consistent factual context preserve consistency when put together. If the resulting set of formulae is consistent this tells us that given the context at hand the agents can jointly satisfy all of their obligations. If the resulting set is inconsistent this means that the factual context does not allow for jointly fulfilling all duties. We observe that, when taken together with a specified factual context, a norm base may imply obligations that are not implied when considering the norm base in isolation. Hence, in contrast to the first two reasoning tasks,  
this particular task is about logical consistency.

The following example illustrates this third reasoning task. Suppose that Yara promised her friend Lisa to pick up a package. That is, Yara has the obligation $\otimes_y\post$, where $\post$ is the proposition ``the package is picked up.'' Let us assume that Yara knows the post-office closes at noon and realizes that if she goes to pick up the package, she will not be able to return the borrowed hammer to Kai at noon, i.e.\ $\Box([y]\post\rightarrow \lnot [y]\clean)$. We may formalize the 
situation as follows:
$$
\{\Diam[y]\clean, \Diam[y]\lnot\clean, \Diam[y]\post, \Diam[y]\lnot\post, \Box([y]\clean \rightarrow \lnot[y]\post), \otimes_y\clean, \otimes_y\post\}
$$
The first four formulae describe Yara's available choices, namely, Yara has the choice to either return or not return the hammer to Kai before noon, and the choice to pick up or not pick up the package for Lisa. The fifth formula expresses that the choices available to Yara are mutually exclusive. The last two formulae encode the obligations applicable to Yara. 
Are Yara's obligations consistent with the facts of the situation? 

We can utilize our proof-search algorithm to derive $(\otimes_y\clean \land\otimes_y\post\land \Box([y]\post\rightarrow\lnot[y]\clean))\rightarrow \bot$ in $\mathsf{G3DS}^{0}_{1}$, where our set of agents is $\Ag =\{y\}$ and $k = 0$. 
By applying the hp-admissibility of $\wk$ to the derivation given below, 
we can indeed show that the above norm base is inconsistent, that is, the following formula is a theorem in $\mathsf{DS}_{1}^{0}$:
$$
(\Diam[y]\clean\land \Diam[y]\lnot\clean \land \Diam[y]\post \land \Diam[y]\lnot\post \land \Box([y]\clean \rightarrow \lnot[y]\post)\land \otimes_y\clean\land \otimes_y\post) \rightarrow \bot
$$
Due to the size of the proof of $(\otimes_y\clean \land\otimes_y\post\land \Box([y]\post\rightarrow\lnot[y]\clean))\rightarrow \bot$, we split the proof into various sections and present each below. We denote the first proof shown below as 
$\Pi_{0}$. The sub-derivation $\Pi'$ constituting the right branch of $\Pi_0$ is obtained in a similar way to the left branch.

\bigskip

\hspace{-1cm}{\scalebox{0.88}{
\AxiomC{}
\RL{$\id$}
\UIC{$I_{\otimes_{y}} v, I_{\otimes_{y}} z,  \ldots \Rightarrow \ldots, v : [y] \po, z : \po, z : \negnnf{\po}$}
\RightLabel{$(\ominus_{y})$}
\UIC{$I_{\otimes_{y}} v, I_{\otimes_{y}} z,  \ldots \Rightarrow \ldots, v : [y] \po, z : \po$}
\RightLabel{$(\mathsf{D3}_{y})$}
\UIC{$I_{\otimes_{y}} v, R_{[y]}vv, R_{[y]}vz, R_{[y]}zv, R_{[y]}zz , \ldots \Rightarrow \ldots, v : [y] \po, z : \po$}
\RightLabel{$(\mathsf{Ref}_{y}),(\mathsf{Euc}_{y})$}
\UIC{$I_{\otimes_{y}} v, R_{[y]}vv, R_{[y]}vz, \ldots \Rightarrow \ldots, v : [y] \po, z : \po$}
\RightLabel{$([y])$}
\UIC{$I_{\otimes_{y}} v, R_{[y]}vv, R_{[y]}vz, \ldots \Rightarrow \ldots, v : [y] \po$}

\AxiomC{$\Pi'$}
\UIC{$I_{\otimes_{y}} v, R_{[y]}vv, \ldots \Rightarrow \ldots, v : [y] \cl$}

\insertBetweenHyps{\hskip -8pt}
\RL{$\conr$}
\BIC{$I_{\otimes_{y}} v, R_{[y]}vv, \ldots \Rightarrow \ldots, v : [y] \po \land [y] \cl$}
\RL{$\settdiar$}
\UIC{$I_{\otimes_{y}} v, R_{[y]}vv, \ldots \Rightarrow \ldots, w: [y]\po, u : \po$}
\RL{$(\mathsf{Ref}_{y})$}
\UIC{$I_{\otimes_{y}} v, \ldots \Rightarrow \ldots, w: [y]\po, u : \po$}
\RL{$(\mathsf{D2}_{y})$}
\UIC{$R_{[y]}ww, R_{[y]}wu, R_{[y]}uw, R_{[y]}uu, I_{\otimes_{y}} v \Rightarrow \ldots, w: [y]\po, u : \po$}
\RL{$(\mathsf{Ref}_{y}), (\mathsf{Euc}_{y})$}
\UIC{$R_{[y]}ww, R_{[y]}wu \Rightarrow \ldots, w: [y]\po, u : \po$}
\RL{$([y])$}
\UIC{$R_{[y]}ww \Rightarrow \ldots, w: [y]\po$}
\DisplayProof
}}
\bigskip

 The derivation of $R_{[y]}ww \Rightarrow \ldots, w: [y]\cl$, referred to as $\Pi_{1}$, is similar to the derivation $\Pi_{0}$ applying the same rules bottom-up in the same order. By making use of $\Pi_{0}$ and $\Pi_{1}$, we can complete our proof as follows:

\begin{center}
\scalebox{0.9}{
\AXC{$\Pi_{0}$}
\AXC{$\Pi_{1}$}
\RL{$\conr$}

\insertBetweenHyps{\hskip 3cm}
\BIC{$R_{[y]}ww \Rightarrow w: \Ody\negnnf{\cl}, w: \Ody\negnnf{\po}, w: \Diam ([y]\po \land [y]\cl), w: [y]\po \land [y]\cl, w: \bot$}\RL{$\diar$}
\UIC{$R_{[y]}ww \Rightarrow w: \Ody\negnnf{\cl}, w: \Ody\negnnf{\po}, w: \Diam ([y]\po \land [y]\cl), w: \bot$}\RL{$\disr \times 3$}
\UIC{$R_{[y]}ww \Rightarrow w: \Ody\negnnf{\cl} \lor \Ody\negnnf{\po}\lor \Diam ([y]\po \land [y]\cl)\lor \bot$}\RL{$(\mathsf{Ref}_{y})$}
\UIC{$\Rightarrow w: \Ody\negnnf{\cl} \lor \Ody\negnnf{\po}\lor \Diam ([y]\po \land [y]\cl)\lor \bot$}\RL{=}
\dashedLine
\UIC{$\Rightarrow w: (\otimes_y\cl \land\otimes_y\po\land \Box([y]\po\rightarrow\lnot[y]\cl)\rightarrow \bot$}
\DisplayProof
}
\end{center}

The above proof illustrates how our 
proof-search algorithm may be used to check the consistency of knowledge bases including obligations, 
possibly relative to a factual context. 
\\
\\
As a final reflection, we observe that, from an abstract point of view, these specific reasoning tasks can be reduced to variations of validity checking. That is, they all use proof-search with automated counter-model extraction to check the validity of a specified formula.  Nevertheless, these three different applications conceptually highlight different reasoning tasks with varying deontic value. 
If the agent’s considered behavior is not compliant we know that there is a conflicting obligation implied by the norm base. If the agent’s considered behavior is compliant this tells us that the action is permissible, but it does not tell us anything about the deontic value of the action (the action may be merely compatible with the agent’s obligatory actions). The third application is also different, checking whether the obligations of a norm base can be consistently fulfilled in a specified factual context irrespective of any considered actions.

\section{Conclusion}\label{Sect:Conclusion}

In this paper, we 
demonstrated how to automate normative reasoning for deontic $\sti$ logics by means of proof-search algorithms founded upon proof systems developed in~\cite{Lyo21thesis} and inspired by the proof systems for related logics introduced in~\cite{BerLyo19a,BerLyo21,LyoBer19}. In doing so, this work is the first to provide terminating proof-search procedures for deontic as well as non-deontic multi-agent $\sti$ logics with (un)limited choice. 
The key to obtaining terminating proof-search was the introduction of a loop-checking mechanism, which 
ensured the introduction of only finitely many labels during proof-search. As argued for in \sect~\ref{Sect:Proof-search}, 
such a loop-checking mechanism is needed to effectively handle the general case of $\sti$ logics with unlimited and limited choice axioms. 

The present work provides a promising outlook for future research. For instance, the proof theory of (and automated reasoning with) extensions of deontic $\sti$ logics remains to be investigated: e.g. temporal~\cite{BerLyo19b,CiuLor17} and epistemic extensions~\cite{Bro11}. In particular, temporal extensions of deontic $\sti$ logics are of interest in the context of verification of obligations for autonomous vehicles \cite{SheAbb21}. Decidability remains an open question for various temporal extensions of deontic $\sti$ \cite{CiuLor17}. Additionally, in relation to normative reasoning tasks, an interesting 
direction would be 
to extend deontic $\sti$ logic with conditional obligations \cite[Ch. 5]{Hor01}. Such obligations are often used to consistently model challenging contrary-to-duty scenarios which are scenarios in which a violation---i.e.\ $[i]\phi\land\Oi\lnot\phi$---has occurred and the agent needs to find out her obligations \textit{given} the sub-ideality of the situation \cite{HilMcN13}.

From a technical point of view, we believe that an important step in automated agential normative reasoning can be taken by writing and evaluating theorem provers based on our proof-search algorithms for the logics in this paper. Due to the modularity of our approach such results may be adapted to extensions of the deontic $\sti$ logics discussed here. On a related note, derivations generated by our proof-search algorithm are not necessarily minimal and shorter proofs may exist. 
Optimization procedures may serve as a fruitful future research direction.

In the context of explainable normative reasoning, we point out that this work provides only a promising first step in showing how logical methods can be harnessed in the context of $\sti$ and Explainable AI. 
As argued, constructive proof-search procedures that yield counter-models in case of failed proof-search have the advantage of providing a justification of a formula's (non-)theoremhood and make the reasons for the conclusion interpretable by providing a step-by-step reasoning procedure. Such explanations take experts instead of laypersons as the intended explainee. However, certain refinements can be made to increase the perspicuity of our proofs and counter-models. For instance, relevance conditions can be imposed on our proof calculi; when we derive the obligation $\otimes_y\foot$ from the knowledge base $\{\otimes_y\noon,\Box(\foot\rightarrow\noon),\Box(\noon\rightarrow\foot)\}$ the premise $\Box(\foot\rightarrow\noon)$ is not strictly needed in the derivation and thus can be removed from the proof as it is irrelevant in explaining the derived duty.

 %\acknowledgments{
 \section*{Acknowledgments}
 This research was supported by the European Research Council (ERC) Consolidator
Grant 771779 
(DeciGUT) and the Austrian Science Fund (FWF) project W1255-N23.%}

\bibliography{bib}

\providecommand{\noopsort}[1]{}
\begin{thebibliography}{}

\bibitem[\protect\BCAY{Alvarez}{Alvarez}{2017}]{Alv17}
Alvarez, M. \BBOP2017\BBCP.
\newblock \BBOQ {Reasons for Action: Justification, Motivation,
  Explanation}\BBCQ\
\newblock In Zalta, E.~N.\BED, {\Bem The {Stanford} Encyclopedia of
  Philosophy\/} ({W}inter 2017 \BEd). Metaphysics Research Lab, Stanford
  University.

\bibitem[\protect\BCAY{Arkoudas, Bringsjord,\ \BBA\ Bello}{Arkoudas
  et~al.}{2005}]{ArkBriBel05}
Arkoudas, K., Bringsjord, S., \BBA\ Bello, P. \BBOP2005\BBCP.
\newblock \BBOQ Toward ethical robots via mechanized deontic logic\BBCQ\
\newblock In {\Bem AAAI Fall Symposium on Machine Ethics}, \BPGS\ 17--23. The
  AAAI Press Menlo Park, CA.

\bibitem[\protect\BCAY{Avron}{Avron}{1996}]{Avr96}
Avron, A. \BBOP1996\BBCP.
\newblock \BBOQ The method of hypersequents in the proof theory of
  propositional non-classical logics\BBCQ\
\newblock In Hodges, W., Hyland, M., Steinhorn, C., \BBA\ Truss, J.\BEDS, {\Bem
  From Foundations to Applications: European Logic Colloquium}, \BPG\ 1–32.
  Clarendon Press.

\bibitem[\protect\BCAY{Balbiani, Herzig,\ \BBA\ Troquard}{Balbiani
  et~al.}{2008}]{BalHerTro08}
Balbiani, P., Herzig, A., \BBA\ Troquard, N. \BBOP2008\BBCP.
\newblock \BBOQ Alternative axiomatics and complexity of deliberative {STIT}
  theories\BBCQ\
\newblock {\Bem Journal of Philosophical Logic}, {\Bem 37\/}(4), 387--406.

\bibitem[\protect\BCAY{Bartha}{Bartha}{1993}]{Bar93}
Bartha, P. \BBOP1993\BBCP.
\newblock \BBOQ Conditional obligation, deontic paradoxes, and the logic of
  agency\BBCQ\
\newblock {\Bem Annals of Mathematics and Artificial Intelligence}, {\Bem
  9\/}(1-2), 1--23.

\bibitem[\protect\BCAY{Belnap}{Belnap}{1982}]{Bel82}
Belnap, N. \BBOP1982\BBCP.
\newblock \BBOQ Display logic\BBCQ\
\newblock {\Bem Journal of Philosophical Logic}, {\Bem 11\/}(4), 375--417.

\bibitem[\protect\BCAY{Belnap, Perloff,\ \BBA\ Xu}{Belnap
  et~al.}{2001}]{BelPerXu01}
Belnap, N., Perloff, M., \BBA\ Xu, M. \BBOP2001\BBCP.
\newblock {\Bem Facing the future: agents and choices in our indeterminist
  world}.
\newblock Oxford University Press, Oxford.

\bibitem[\protect\BCAY{{\noopsort{Berkell}{van Berkel}}\ \BBA\
  Lyon}{{\noopsort{Berkell}{van Berkel}}\ \BBA\ Lyon}{2019a}]{BerLyo19a}
{\noopsort{Berkell}{van Berkel}}, K.\BBACOMMA\  \BBA\ Lyon, T. \BBOP2019a\BBCP.
\newblock \BBOQ Cut-free calculi and relational semantics for temporal {STIT}
  logics\BBCQ\
\newblock In Calimeri, F., Leone, N., \BBA\ Manna, M.\BEDS, {\Bem European
  Conference on Logics in Artificial Intelligence (JELIA)}, \BPGS\ 803--819,
  Cham. Springer.

\bibitem[\protect\BCAY{{\noopsort{Berkell}{van Berkel}}\ \BBA\
  Lyon}{{\noopsort{Berkell}{van Berkel}}\ \BBA\ Lyon}{2019b}]{BerLyo19b}
{\noopsort{Berkell}{van Berkel}}, K.\BBACOMMA\  \BBA\ Lyon, T. \BBOP2019b\BBCP.
\newblock \BBOQ A neutral temporal deontic {STIT} logic\BBCQ\
\newblock In {\Bem International Workshop on Logic, Rationality and
  Interaction}, \BPGS\ 340--354. Springer.

\bibitem[\protect\BCAY{{\noopsort{Berkell}{van Berkel}}\ \BBA\
  Lyon}{{\noopsort{Berkell}{van Berkel}}\ \BBA\ Lyon}{2021}]{BerLyo21}
{\noopsort{Berkell}{van Berkel}}, K.\BBACOMMA\  \BBA\ Lyon, T. \BBOP2021\BBCP.
\newblock \BBOQ The varieties of ought-implies-can and deontic {STIT}
  logic\BBCQ\
\newblock In Liu, F., Marra, A., Portner, P., \BBA\ Putte, F. V.~D.\BEDS, {\Bem
  Deontic Logic and Normative Systems: 15th International Conference
  (DEON2020/2021, Munich)}, \BPGS\ 57--76. College Publications.

\bibitem[\protect\BCAY{{\noopsort{Berkels}}van~Berkel\ \BBA\
  Straßer}{{\noopsort{Berkels}}van~Berkel\ \BBA\ Straßer}{2022}]{BerStr22}
{\noopsort{Berkels}}van~Berkel, K.\BBACOMMA\  \BBA\ Straßer, C.
  \BBOP2022\BBCP.
\newblock \BBOQ Reasoning with and about norms in logical argumentation\BBCQ\
\newblock In Toni, F., Polberg, S., Booth, R., Caminada, M., \BBA\ Kido,
  H.\BEDS, {\Bem Frontiers in Artificial Intelligence and Applications:
  Computational Models of Argument, proceedings (COMMA22)}, \lowercase{\BVOL}\
  353, \BPGS\ 332 -- 343. IOS press.

\bibitem[\protect\BCAY{{\noopsort{Berkels}}van~Berkel\ \BBA\
  Straßer}{{\noopsort{Berkels}}van~Berkel\ \BBA\ Straßer}{2024}]{BerStr24}
{\noopsort{Berkels}}van~Berkel, K.\BBACOMMA\  \BBA\ Straßer, C.
  \BBOP2024\BBCP.
\newblock \BBOQ Towards deontic explanations through dialogue\BBCQ\
\newblock In Čyras, K., Kampik, T., Cocarascu, O., \BBA\ Rago, A.\BEDS, {\Bem
  2nd International Workshop on Argumentation for eXplainable AI (ArgXAI)}.
  CEUR Workshop Proceedings (To Appear).

\bibitem[\protect\BCAY{Biran\ \BBA\ Cotton}{Biran\ \BBA\
  Cotton}{2017}]{BirCot17}
Biran, O.\BBACOMMA\  \BBA\ Cotton, C. \BBOP2017\BBCP.
\newblock \BBOQ Explanation and justification in machine learning: A
  survey\BBCQ\
\newblock In {\Bem IJCAI-17 Workshop on Explainable AI (XAI)}, \BPGS\ 8--13.

\bibitem[\protect\BCAY{Broersen}{Broersen}{2011}]{Bro11}
Broersen, J. \BBOP2011\BBCP.
\newblock \BBOQ Deontic epistemic {STIT} logic distinguishing modes of mens
  rea\BBCQ\
\newblock {\Bem Journal of Applied Logic}, {\Bem 9\/}(2), 137--152.

\bibitem[\protect\BCAY{Br{\"{u}}nnler}{Br{\"{u}}nnler}{2009}]{Bru09}
Br{\"{u}}nnler, K. \BBOP2009\BBCP.
\newblock \BBOQ Deep sequent systems for modal logic\BBCQ\
\newblock {\Bem Archive for Mathematical Logic}, {\Bem 48\/}(6), 551--577.

\bibitem[\protect\BCAY{Bull}{Bull}{1992}]{Bul92}
Bull, R.~A. \BBOP1992\BBCP.
\newblock \BBOQ Cut elimination for propositional dynamic logic without *\BBCQ\
\newblock {\Bem Mathematical Logic Quarterly}, {\Bem 38\/}(1), 85--100.

\bibitem[\protect\BCAY{Ciuni\ \BBA\ Lorini}{Ciuni\ \BBA\
  Lorini}{2017}]{CiuLor17}
Ciuni, R.\BBACOMMA\  \BBA\ Lorini, E. \BBOP2017\BBCP.
\newblock \BBOQ Comparing semantics for temporal {STIT} logic\BBCQ\
\newblock {\Bem Logique et Analyse}, {\Bem 61\/}(243), 299--339.

\bibitem[\protect\BCAY{Dalmonte, Lellmann, Olivetti,\ \BBA\ Pimentel}{Dalmonte
  et~al.}{2021}]{Daletal21}
Dalmonte, T., Lellmann, B., Olivetti, N., \BBA\ Pimentel, E. \BBOP2021\BBCP.
\newblock \BBOQ Hypersequent calculi for non-normal modal and deontic logics:
  countermodels and optimal complexity\BBCQ\
\newblock {\Bem Journal of Logic and Computation}, {\Bem 31\/}(1), 67--111.

\bibitem[\protect\BCAY{Fitting}{Fitting}{1972}]{Fit72}
Fitting, M. \BBOP1972\BBCP.
\newblock \BBOQ Tableau methods of proof for modal logics.\BBCQ\
\newblock {\Bem Notre Dame Journal of Formal Logic}, {\Bem 13\/}(2), 237--247.

\bibitem[\protect\BCAY{Fitting}{Fitting}{2014}]{Fit14}
Fitting, M. \BBOP2014\BBCP.
\newblock \BBOQ Nested sequents for intuitionistic logics\BBCQ\
\newblock {\Bem Notre Dame Journal of Formal Logic}, {\Bem 55\/}(1), 41--61.

\bibitem[\protect\BCAY{Gabbay}{Gabbay}{1996}]{Gab96}
Gabbay, D.~M. \BBOP1996\BBCP.
\newblock {\Bem Labelled deductive systems}, \lowercase{\BVOL}~33 of {\Bem
  Oxford Logic Guides}.
\newblock Clarendon Press/Oxford Science Publications.

\bibitem[\protect\BCAY{Gentzen}{Gentzen}{1935a}]{Gen35a}
Gentzen, G. \BBOP1935a\BBCP.
\newblock \BBOQ Untersuchungen {\"u}ber das logische {S}chlie{\ss}en. {I}\BBCQ\
\newblock {\Bem Mathematische Zeitschrift}, {\Bem 39\/}(1), 176--210.

\bibitem[\protect\BCAY{Gentzen}{Gentzen}{1935b}]{Gen35b}
Gentzen, G. \BBOP1935b\BBCP.
\newblock \BBOQ Untersuchungen {\"u}ber das logische {S}chlie{\ss}en.
  {II}\BBCQ\
\newblock {\Bem Mathematische Zeitschrift}, {\Bem 39\/}(1), 405--431.

\bibitem[\protect\BCAY{Gunning, Stefik, Choi, Miller, Stumpf,\ \BBA\
  Yang}{Gunning et~al.}{2019}]{Gunetal19}
Gunning, D., Stefik, M., Choi, J., Miller, T., Stumpf, S., \BBA\ Yang, G.-Z.
  \BBOP2019\BBCP.
\newblock \BBOQ Xai—explainable artificial intelligence\BBCQ\
\newblock {\Bem Science Robotics}, {\Bem 4\/}(37), eaay7120.

\bibitem[\protect\BCAY{Hein}{Hein}{2005}]{Hei05}
Hein, R. \BBOP2005\BBCP.
\newblock \BBOQ Geometric theories and proof theory of modal logic\BBCQ\
\newblock Master's thesis, Technische Universit\"at Dresden.

\bibitem[\protect\BCAY{Herzig\ \BBA\ Schwarzentruber}{Herzig\ \BBA\
  Schwarzentruber}{2008}]{HerSch08}
Herzig, A.\BBACOMMA\  \BBA\ Schwarzentruber, F. \BBOP2008\BBCP.
\newblock \BBOQ Properties of logics of individual and group agency\BBCQ\
\newblock In Areces, C.\BBACOMMA\  \BBA\ Goldblatt, R.\BEDS, {\Bem Proceedings
  of the 7th Conference on Advances in Modal Logic (AIML)},
  \lowercase{\BVOL}~7, \BPGS\ 133--149. College Publications.

\bibitem[\protect\BCAY{Hilpinen\ \BBA\ McNamara}{Hilpinen\ \BBA\
  McNamara}{2013}]{HilMcN13}
Hilpinen, R.\BBACOMMA\  \BBA\ McNamara, P. \BBOP2013\BBCP.
\newblock \BBOQ Deontic logic: A historical survey and introduction\BBCQ\
\newblock In {\Bem Handbook of Deontic Logic and Normative Systems},
  \lowercase{\BVOL}~1, \BPGS\ 3--136. College Publications Milton Keynes.

\bibitem[\protect\BCAY{Horrocks\ \BBA\ Sattler}{Horrocks\ \BBA\
  Sattler}{2004}]{HorSat04}
Horrocks, I.\BBACOMMA\  \BBA\ Sattler, U. \BBOP2004\BBCP.
\newblock \BBOQ Decidability of {SHIQ} with complex role inclusion axioms\BBCQ\
\newblock {\Bem Artificial Intelligence}, {\Bem 160\/}(1-2), 79--104.

\bibitem[\protect\BCAY{Horty}{Horty}{2001}]{Hor01}
Horty, J.~F. \BBOP2001\BBCP.
\newblock {\Bem Agency and deontic logic}.
\newblock Oxford University Press.

\bibitem[\protect\BCAY{Horty\ \BBA\ Belnap}{Horty\ \BBA\
  Belnap}{1995}]{HorBel95}
Horty, J.~F.\BBACOMMA\  \BBA\ Belnap, N. \BBOP1995\BBCP.
\newblock \BBOQ The deliberative {STIT}: A study of action, omission, ability,
  and obligation\BBCQ\
\newblock {\Bem Journal of Philosophical Logic}, {\Bem 24\/}(6), 583--644.

\bibitem[\protect\BCAY{Kashima}{Kashima}{1994}]{Kas94}
Kashima, R. \BBOP1994\BBCP.
\newblock \BBOQ Cut-free sequent calculi for some tense logics\BBCQ\
\newblock {\Bem Studia Logica}, {\Bem 53\/}(1), 119--135.

\bibitem[\protect\BCAY{Lahav}{Lahav}{2013}]{Lah13}
Lahav, O. \BBOP2013\BBCP.
\newblock \BBOQ From frame properties to hypersequent rules in modal
  logics\BBCQ\
\newblock In {\Bem 28th Annual {ACM/IEEE} Symposium on Logic in Computer
  Science ({LICS})}, \BPGS\ 408--417.

\bibitem[\protect\BCAY{Langley, Meadows, Sridharan,\ \BBA\ Choi}{Langley
  et~al.}{2017}]{LanMaeSriCho17}
Langley, P., Meadows, B., Sridharan, M., \BBA\ Choi, D. \BBOP2017\BBCP.
\newblock \BBOQ Explainable agency for intelligent autonomous systems\BBCQ\
\newblock In {\Bem Proceedings of the Thirty-First AAAI Conference on
  Artificial Intelligence}, AAAI'17, \BPG\ 4762–4763. AAAI Press.

\bibitem[\protect\BCAY{Lorini\ \BBA\ Sartor}{Lorini\ \BBA\
  Sartor}{2015}]{LorSar15}
Lorini, E.\BBACOMMA\  \BBA\ Sartor, G. \BBOP2015\BBCP.
\newblock \BBOQ Influence and responsibility: A logical analysis\BBCQ\
\newblock In Rotolo, A.\BED, {\Bem Frontiers in Artificial Intelligence and
  Applications: Legal Knowledge and Information Systems (JURIX 2015)},
  \lowercase{\BVOL}\ 279, \BPGS\ 51--60.

\bibitem[\protect\BCAY{Lyon}{Lyon}{2021}]{Lyo21thesis}
Lyon, T. \BBOP2021\BBCP.
\newblock {\Bem Refining Labelled Systems for Modal and Constructive Logics
  with Applications}.
\newblock Ph.D.\ thesis, Technische Universit\"at Wien.

\bibitem[\protect\BCAY{Lyon\ \BBA\ van Berkel}{Lyon\ \BBA\ van
  Berkel}{2019}]{LyoBer19}
Lyon, T.\BBACOMMA\  \BBA\ van Berkel, K. \BBOP2019\BBCP.
\newblock \BBOQ Automating agential reasoning: Proof-calculi and syntactic
  decidability for {STIT} logics\BBCQ\
\newblock In Baldoni, M., Dastani, M., Liao, B., Sakurai, Y., \BBA\
  Zalila~Wenkstern, R.\BEDS, {\Bem Principles and Practice of Multi-Agent
  Systems - 22nd International Conference ({PRIMA})}, \lowercase{\BVOL}\ 11873
  of {\Bem LNCS}, \BPGS\ 202--218. Springer.

\bibitem[\protect\BCAY{Miller}{Miller}{2019}]{Mil19}
Miller, T. \BBOP2019\BBCP.
\newblock \BBOQ Explanation in artificial intelligence: Insights from the
  social sciences\BBCQ\
\newblock {\Bem Artificial intelligence}, {\Bem 267}, 1--38.

\bibitem[\protect\BCAY{Miller, Howe,\ \BBA\ Sonenberg}{Miller
  et~al.}{2017}]{MilHowSon17}
Miller, T., Howe, P., \BBA\ Sonenberg, L. \BBOP2017\BBCP.
\newblock \BBOQ Explainable {AI}: Beware of inmates running the asylum or: How
  {I} learnt to stop worrying and love the social and behavioural
  sciences\BBCQ\
\newblock {\Bem arXiv preprint arXiv:1712.00547}, {\Bem abs/1712.00547}.

\bibitem[\protect\BCAY{Murakami}{Murakami}{2004}]{Mur05}
Murakami, Y. \BBOP2004\BBCP.
\newblock \BBOQ Utilitarian deontic logic\BBCQ\
\newblock In Schmidt, R.~A., Pratt{-}Hartmann, I., Reynolds, M., \BBA\ Wansing,
  H.\BEDS, {\Bem Poceedings of the 5th Conference on Advances in Modal Logic
  (AIML)}, \BPGS\ 211--230.

\bibitem[\protect\BCAY{Negri}{Negri}{2005}]{Neg05}
Negri, S. \BBOP2005\BBCP.
\newblock \BBOQ Proof analysis in modal logic\BBCQ\
\newblock {\Bem Journal of Philosophical Logic}, {\Bem 34}, 507.

\bibitem[\protect\BCAY{Negri\ \BBA\ Pavlovi{\'c}}{Negri\ \BBA\
  Pavlovi{\'c}}{2020}]{Negri2020}
Negri, S.\BBACOMMA\  \BBA\ Pavlovi{\'c}, E. \BBOP2020\BBCP.
\newblock \BBOQ Proof-theoretic analysis of the logics of agency: The
  deliberative {STIT}\BBCQ\
\newblock {\Bem Studia Logica}, {\Bem 109}, 1--35.

\bibitem[\protect\BCAY{Poggiolesi}{Poggiolesi}{2009}]{Pog09}
Poggiolesi, F. \BBOP2009\BBCP.
\newblock \BBOQ The method of tree-hypersequents for modal propositional
  logic\BBCQ\
\newblock In Makinson, D., Malinowski, J., \BBA\ Wansing, H.\BEDS, {\Bem
  Towards Mathematical Philosophy}, \lowercase{\BVOL}~28 of {\Bem Trends in
  logic}, \BPGS\ 31--51. Springer.

\bibitem[\protect\BCAY{Schmidt\ \BBA\ Tishkovsky}{Schmidt\ \BBA\
  Tishkovsky}{2011}]{SchTis11}
Schmidt, R.~A.\BBACOMMA\  \BBA\ Tishkovsky, D. \BBOP2011\BBCP.
\newblock \BBOQ Automated synthesis of tableau calculi\BBCQ\
\newblock {\Bem {Logical Methods in Computer Science}}, {\Bem 7\/}(2).

\bibitem[\protect\BCAY{Shea-Blymyer\ \BBA\ Abbas}{Shea-Blymyer\ \BBA\
  Abbas}{2020}]{SheAbb20}
Shea-Blymyer, C.\BBACOMMA\  \BBA\ Abbas, H. \BBOP2020\BBCP.
\newblock \BBOQ A deontic logic analysis of autonomous systems' safety\BBCQ\
\newblock In {\Bem Proceedings of the 23rd International Conference on Hybrid
  Systems: Computation and Control}, \BPGS\ 1--11, New York. Association for
  Computing Machinery.

\bibitem[\protect\BCAY{Shea-Blymyer\ \BBA\ Abbas}{Shea-Blymyer\ \BBA\
  Abbas}{2021}]{SheAbb21}
Shea-Blymyer, C.\BBACOMMA\  \BBA\ Abbas, H. \BBOP2021\BBCP.
\newblock \BBOQ Algorithmic ethics: Formalization and verification of
  autonomous vehicle obligations\BBCQ\
\newblock {\Bem ACM Transactions on Cyber-Physical Systems (TCPS)}, {\Bem
  5\/}(4), 1--25.

\bibitem[\protect\BCAY{Simpson}{Simpson}{1994}]{Sim94}
Simpson, A.~K. \BBOP1994\BBCP.
\newblock {\Bem The proof theory and semantics of intuitionistic modal logic}.
\newblock Ph.D.\ thesis, University of Edinburgh.

\bibitem[\protect\BCAY{Tiu, Ianovski,\ \BBA\ Gor{\'{e}}}{Tiu
  et~al.}{2012}]{TiuIanGor12}
Tiu, A., Ianovski, E., \BBA\ Gor{\'{e}}, R. \BBOP2012\BBCP.
\newblock \BBOQ Grammar logics in nested sequent calculus: Proof theory and
  decision procedures\BBCQ\
\newblock In {\Bem Proceedings of the 9th Conference on Advances in Modal Logic
  (AIML)}, \BPGS\ 516--537.

\bibitem[\protect\BCAY{Vigan{\`o}}{Vigan{\`o}}{2000}]{Vig00}
Vigan{\`o}, L. \BBOP2000\BBCP.
\newblock {\Bem Labelled Non-Classical Logics}.
\newblock Springer Science \& Business Media.

\bibitem[\protect\BCAY{Wansing}{Wansing}{1994}]{Wan94}
Wansing, H. \BBOP1994\BBCP.
\newblock \BBOQ Sequent calculi for normal modal propositional logics\BBCQ\
\newblock {\Bem Journal of Logic and Computation}, {\Bem 4\/}(2), 125--142.

\bibitem[\protect\BCAY{Wansing}{Wansing}{2002}]{Wan02}
Wansing, H. \BBOP2002\BBCP.
\newblock \BBOQ Sequent systems for modal logics\BBCQ\
\newblock In Gabbay, D.~M.\BBACOMMA\  \BBA\ Guenthner, F.\BEDS, {\Bem Handbook
  of Philosophical Logic: Volume 8}, \BPGS\ 61--145. Springer, Dordrecht.

\bibitem[\protect\BCAY{Wansing}{Wansing}{2006}]{Wan06}
Wansing, H. \BBOP2006\BBCP.
\newblock \BBOQ Tableaux for multi-agent deliberative-{STIT} logic.\BBCQ\
\newblock In {\Bem Proceedings of the 6th Conference on Advances in Modal Logic
  (AIML)}, \BPGS\ 503--520.

\bibitem[\protect\BCAY{Ye\ \BBA\ Johnson}{Ye\ \BBA\ Johnson}{1995}]{YeJoh95}
Ye, L.~R.\BBACOMMA\  \BBA\ Johnson, P.~E. \BBOP1995\BBCP.
\newblock \BBOQ The impact of explanation facilities on user acceptance of
  expert systems advice\BBCQ\
\newblock {\Bem Mis Quarterly}, {\Bem 19}, 157--172.

\end{thebibliography}
\bibliographystyle{theapa}

\end{document}